\documentclass[12pt]{article}

\usepackage[margin=2 cm]{geometry} 
\usepackage{amsmath,amsthm,amsfonts}
\usepackage{bbm}
\usepackage[round]{natbib}  
\usepackage{verbatim} 
\usepackage{subfig}
\usepackage[hidelinks]{hyperref} 

\usepackage{enumitem}
\usepackage{algorithm}
\usepackage{algpseudocode}
\usepackage{tikz}
\usepackage{setspace} 
\usepackage{authblk}
\usepackage{multirow}

\newif\ifplot

\plottrue

\numberwithin{equation}{section}

\DeclareMathOperator{\E}{\mathbb{E}}
\DeclareMathOperator{\argmax}{argmax}
\DeclareMathOperator{\argmin}{argmin}
\DeclareMathOperator{\Var}{Var}

\newcommand{\cont}[4]{\mathcal{C}_{#1,#2}^{#3}\left(#4\right)}
\newcommand{\contt}[4]{\tilde{\mathcal{C}}_{#1,#2}^{#3}\left(#4\right)}
\newcommand{\Pb}[1]{\ensuremath{\mathbb{P}\left(#1\right)}}

\newcommand{\abs}[1]{\left|#1\right|}
\newcommand{\ip}[2]{\left\langle #1, #2 \right\rangle}

\DeclareMathOperator{\sign}{sign}

\newcommand{\pkg}[1]{{\fontseries{b}\selectfont #1}} 
\let\proglang=\textsf
\let\code=\texttt

\newtheorem{Theorem}{Theorem}
\newtheorem{Lemma}{Lemma}
\newtheorem{Corollary}{Corollary}

\newcommand{\Mc}{\mathcal{M}}
\newcommand{\Nc}{\mathcal{N}}
\newcommand{\Dc}{\mathcal{D}}
\newcommand{\Sc}{\mathcal{S}}
\newcommand{\Oc}{\mathcal{O}}
\newcommand{\Ic}{\mathcal{I}}

\newcommand{\Tc}{\mathcal{T}}

\newcommand{\R}{\mathbb{R}}

\newcommand{\fb}{\mathbf{f}}
\newcommand{\vb}{\mathbf{v}}
\newcommand{\Yb}{\mathbf{Y}}
\newcommand{\psib}{\boldsymbol{\psi}}
\newcommand{\phib}{\boldsymbol{\phi}}

\newcommand{\Thetab}{\boldsymbol{\Theta}}
\newcommand{\varepsilonb}{\boldsymbol{\varepsilon}}
\newcommand{\gammab}{\boldsymbol{\gamma}}
\newcommand{\Gammab}{\boldsymbol{\Gamma}}

\newcommand{\fl}{\underline{f}}

\newcommand{\Cl}{\underline{C}}

\title{Narrowest-Over-Threshold Change-point Detection}

\author[1]{Rafal Baranowski}
\author[1]{Yining Chen}
\author[1]{Piotr Fryzlewicz}

\affil[1]{Department of Statistics, London School of Economics and Political Science, Houghton Street, London, WC2A 2AE, UK.}

\begin{document}

\maketitle 

\begin{abstract}
We propose a new, generic and flexible methodology for nonparametric function estimation, in which we first estimate the number and locations of any features that may be present in the function, and then estimate the function parametrically between each pair of neighbouring detected features. Examples of features handled by our methodology include change-points in the piecewise-constant signal model, kinks in the piecewise-linear signal model, and other similar irregularities, which we also refer to as generalised change-points.

Our methodology works with only minor modifications across a range of generalised change-point scenarios, and we achieve such a high degree of generality by proposing and using a new multiple generalised change-point detection device, termed Narrowest-Over-Threshold (NOT). The key ingredient of NOT is its focus on the smallest local sections of the data on which the existence of a feature is suspected. Crucially, this adaptive localisation technique prevents NOT from considering subsamples containing two or more features, a key factor that ensures the general applicability of NOT.

For selected scenarios, we show the consistency and near-optimality of NOT in detecting the number and locations of generalised change-points. Furthermore, we propose to select NOT's threshold via the strengthened Schwarz Information Criterion (sSIC) and give theoretical justifications. The NOT estimators are easy to implement and rapid to compute: the entire threshold-indexed solution path can be computed in close-to-linear time. Importantly, the NOT approach is easy to extend by the user to tailor to their own needs. There is no single competitor, but we show that the performance of NOT matches or surpasses the state of the art in the scenarios tested. Our methodology is implemented in the R package \textbf{not}.
\newline

\noindent \textbf{keywords}: Break-point detection, knots, piecewise-polynomial, segmentation, splines.
\end{abstract}


\section{Introduction}
\label{Sec:introduction}

This paper considers the canonical univariate statistical model
\begin{align}
	\label{Eq:signal+noise1}
	Y_{t}=f_{t}+\varepsilon_{t}, \quad t=1,\ldots,T,
\end{align}
where the deterministic and unknown signal $f_t$ is believed to
display some 
regularity across the index $t$, and the stochastic noise 
$\varepsilon_t$ is exactly or approximately centred at zero. Despite
the simplicity of model (\ref{Eq:signal+noise1}), inferring
information about $f_t$ remains a task of fundamental importance in
modern applied statistics and data science. 
When the interest is in the detection
of ``features'' in $f_t$ such as jumps or kinks, then non-linear techniques are usually required.

If $f_t$ is modelled as
piecewise-constant and it is of interest to detect its change-points, several
techniques are available, and we only mention a selection. For Gaussian noise $\varepsilon_t$, both non-penalised and penalised least squares approaches are considered by \citet{YaoAu1989}. For specific choices of penalty functions, see e.g. \citet{Yao1988}, \citet{Lavielle2005} and \citet{Davis2006}. The Gaussianity assumption on $\varepsilon_t$ is relaxed to exponential family distributions in \citet{Lee1997}, \citet{Hawkins2001} and \citet{frick2014multiscale}. In particular, \citet{frick2014multiscale} also provide confidence intervals for the location of the estimated change-points. Often this penalty-type approach requires a computational cost of at least $O(T^2)$. However, there are exceptions, such as the Pruned Exact Linear Time method (PELT, \citet{killick2012optimal}), which achieves a linear computational cost, but requires the further assumption that change-points are separated by time intervals drawn independently from some probability distribution, a scenario in which considerations of statistical consistency are not generally possible. A nonparametric version of PELT is investigated by \citet{haynes2016computationally}. Another general approach is based on the idea of Binary Segmentation (BS; \citealp{vostrikova1981detection}), which can be viewed as a greedy approach with a limited computational cost. Its popular variants include the Circular Binary Segmentation (CBS; \citealp{olshen2004circular}) and the Wild Binary Segmentation (WBS; \citealp{fryzlewicz2014wild}). A selection of publications and software can be found in the online repository \emph{changepoint.info} maintained by \citet{cprepository}. 

More general change-point problems, in which $f_t$ is modelled as piecewise-parametric (not necessarily piecewise-constant) between ``knots'', the number and locations of which are unknown and need to be estimated, have attracted less interest in the literature and overwhelmingly focus on linear trend detection. Among them, we mention the approach based on the least squares principle and Wald-type tests by \citet{BaiPerron1998}, dynamic programming using the $L_0$ penalty \citep{MFL2017}, and trend filtering \citep{tibshirani2014adaptive, lin2016approximate}. Finally, we mention a related problem of jump regression, where the aim is to estimate the points of sharp cusps or discontinuities of a regression function. As investigated in, e.g., \citet{Wang1995} and \citet{XiaQiu2015}, it proceeds by estimating the locations of features nonparametrically via wavelets or local kernel smoothing. However, this not only requires the choice of some tuning parameters (e.g. scale or bandwidth) but also results in scale/bandwidth-dependent (and occasionally sub-optimal) rates for the  
estimated locations of such points.

The aim of this work is to propose a new, generic approach to the problem of
detecting an unknown number of ``features'' occurring at unknown locations
in $f_t$. By a feature, we mean a characteristic of $f_t$, 
occurring at a location $t_0$, that is detectable
by considering a sufficiently large subsample of data $Y_t$ around $t_0$.
Examples include: change-points in $f_t$ when it is modelled as piecewise-constant,
change-points in the first derivative when $f_t$ is modelled as piecewise-linear
and continuous, and discontinuities in $f_t$ or its first derivative when $f_t$ 
is modelled as piecewise-linear but without the continuity constraint. 
We will provide a precise description of the type of features we are interested
in later on. Moving beyond $f_t$ only, our approach will also permit 
the detection of similar features present in some distributional aspects of
$\varepsilon_t$, for example in its variance. Since all types of features we consider describe changes in a parametric description of $f_t$, we use the terms ``feature detection'' and ``change-­point detection'' interchangeably throughout the paper. Occasionally, for precision, we will be referring to change-point detection in the piecewise-constant model as the ``canonical'' change-point problem, while our general feature detection problem will sometimes be referred to as a ``generalised'' change-point problem.

Core to our approach is a particular blend of ``global'' and ``local''
treatment of the data $Y_t$ 
in the search for the multiple features that may be present in $f_t$,
a combination that gives our method a multiscale character.
At the first ``global'' stage, we randomly draw a number of subsamples  $(Y_s, Y_{s+1}, \ldots, Y_e)'$, where $1 \le s < e \le T$. On each 
subsample, we assume, possibly erroneously, that {\em only one}
feature is present and use a tailor-made contrast function derived 
(according to a universal recipe we provide later) from
the likelihood theory to find the most likely location of the feature.
We retain those subsamples for which the contrast {\em exceeds
	a certain user-specified threshold}, and discard the others. Amongst
the retained subsamples, we search for the one drawn on the
{\em narrowest} interval, i.e. one for which $e-s$ is the smallest:
it is this step that gives rise to the name {\em Narrowest-Over-Threshold}
(NOT) for our methodology. The focus on the narrowest interval
constitutes the ``local'' part of the method, and is a key ingredient
of our approach which ensures that with high probability, at
most one feature is present in the selected interval. This key observation
gives our methodology a general character and allows it
to be used, only with minor
modifications, in a wide range of scenarios, including those described
in the previous paragraph. Having detected the first feature, the algorithm then proceeds 
recursively to the left and to the right of it, and stops, on any current interval,
if no contrasts can be found that exceed the threshold.

Besides its generic character, other benefits of the
proposed methodology include low computational complexity, ease of 
implementation, accuracy in the detection of the feature locations, and the fact
that it enables parametric (and hence: interpretable) estimation of the signal on each section delimited by
a pair of neighbouring estimated features. Regarding the computational complexity,
the facts that only a limited number of data subsamples, $M$, need to be drawn (we provide
precise bounds later; with finitely many change-points, one can take $M=O(\log T)$ in general), and that typical contrasts are computable in linear time,
lead to a computational complexity of $O(MT)$ for the entire procedure. Moreover,
the entire threshold-indexed solution path can also be computed efficiently, in typically close-to-linear time, as observed from our numerical experiments. Regarding the estimation accuracy, in the scenarios we consider theoretically, our procedure yields near-optimal rates of convergence
for the estimators of feature locations.

Importantly, the flexible character of our methodology leaves it open to possible
extensions and modifications. Indeed, borrowing words from \citet{sweldens2000building}, who 
advocated ``building your own wavelets at home'',
we also view our proposal as flexible enough to enable the user to
``construct their own feature detector at home'', e.g. by proposing their own 
specialised contrast functions, or by data-adaptively choosing the most suitable 
contrast function from a pre-specified dictionary (which would lead to mixed-type
feature detection). Although these extensions are not covered in
the current work, we view this modularity and flexibility offered by our methodology
as an important aspect of our proposal.

On a broader level, our methodology promotes the idea of ``fitting simple models
on subsets of the data (the local aspect), and then aggregating the results to obtain the overall fit
(the global aspect)'', an idea also present in
the Wild Binary Segmentation method of \citet{fryzlewicz2014wild}. However, we emphasise
that the way the simple models (here: models containing {\em at most one} change-point or other
feature) are fitted in the NOT and WBS methods are entirely different
and have different aims. Unlike the WBS, the NOT methodology focuses on the
{\em narrowest} intervals of the data on which it is possible to locate the feature of interest.
It is this focus that enables NOT to extend well beyond mere change-point detection for a piecewise-constant
$f_t$, the latter being the sole focus of the WBS method. The lack of the narrowest-interval focus in the WBS and BS
methods means that they are not applicable to more general feature detection, and we explain the
mechanics of this phenomenon briefly in the following simple example.

Consider a continuous piecewise-linear signal that has two change-points in its first derivative:
\begin{align}
\label{Eq:motivating_example}
f_{t} = \begin{cases}
\frac{1}{350} t , & t=1,\ldots, 350,\\
1 , & t=351,\ldots, 650,\\
\frac{1001}{350} -\frac{1}{350} t , & t=651,\ldots, 1000.
\end{cases}
\end{align}
\begin{figure}[!ht]
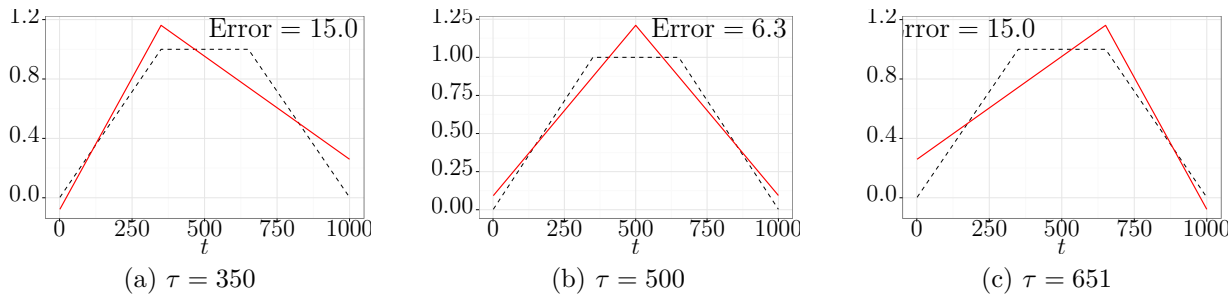

	\centering
	\null\hfill
	\subfloat[][$\tau=350$]{
			\resizebox{0.28\textwidth}{!}{\input{tikz/motivating_example_one_cpt_1.tex}}
\label{Fig:motivating_example_one_cpt_1}}
	\hfill
	\subfloat[][$\tau=500$]{
			\resizebox{0.28\textwidth}{!}{\input{tikz/motivating_example_one_cpt_2.tex}}
\label{Fig:motivating_example_one_cpt_2}}
	\hfill
	\subfloat[][$\tau=651$]{
			\resizebox{0.28\textwidth}{!}{\input{tikz/motivating_example_one_cpt_3.tex}}
\label{Fig:motivating_example_one_cpt_3}}
	\hfill\null
	\caption{\label{Fig:motivating_example_intro} 
		Best $\ell_2$ approximation of the true signal (dashed) via a triangular signal with a single change­-point, the location of which is fixed at the left change-­point (left panel), halfway between the true change-­points (middle panel) and at the right change-­point (right panel). Approximation errors (in terms of squared $\ell_2$ distance) are given in the top-right corners of the corresponding panels. }
\end{figure}

\noindent If we approximate $f_t$ using a piecewise-linear signal with only one change-point in its derivative, then the best approximation (in terms of minimising the $\ell_2$ distance) will result in an estimated change-point at $t=500$, which is away from the true ones at $t=350$ and $t=650$, as is illustrated in Figure~\ref{Fig:motivating_example_intro}.  Therefore, taking the entire sample of data starting at $s = 1$ and ending at $e = 1000$, and searching for one of its multiple change-points by fitting, via least squares, a triangular signal with a single change-­point, does not make sense. It is this issue that leads to the failure of the BS and WBS methods. On the other hand, NOT avoids this issue because of its unique feature of picking the {\em narrowest} intervals, which are likely to contain only one change-­point. To understand the mechanics of this key feature, imagine that now $f_t$ is observed with noise. Through its pursuit of the narrowest intervals, NOT will ensure that, with high probability, some suitably narrow intervals around the change-­points $t = 350$ and $t = 650$ are considered. More precisely, by construction, they will be {\em narrow enough to contain only one change-point each}, but wide enough for the designed contrast (see Section~\ref{Sec:scen:change_in_slope} for more on contrasts) to indicate the existence of the change-point within both of them. The designed contrast function will indicate the right location of the change-point (modulo the estimation error) if only one change-point is present in the data subsample considered, unlike in the situation described earlier in which multiple change-­points were included in the chosen interval. More details on this example are presented in Section~\ref{Sec:illustrative_example} of the online supplementary materials.

We note that this example is different from the canonical change-point detection problem (i.e. piecewise-constant signal with multiple change-points), where if we approximate the signal using a piecewise-constant function with only one change-point, the change-point of the fitted signal will always be among the true ones \citep{venkatraman1992consistency}. Since the latter property does not hold in most generalised change-point detection problems, this highlights the need for new methods with better localisation of the feature of interest, such as our NOT algorithm. In the final stages of preparing this manuscript, we learned that \citet{fang2016segmentation} independently considered a related shortest-interval idea in the context of the canonical change-point detection problem. However, they did not consider it as a springboard to more general feature detection problems, which is the key motivation behind NOT and its most valuable contribution.

To summarise, in the NOT approach, we propose a new ``modus operandi'' 
in statistical smoothing, by providing a novel, general, flexible framework for feature
detection and interpretable signal estimation. The procedure is fast, accurate, 
easy to code and to extend by the users to tailor to their own needs. Its implementation
is provided in the \proglang{R} package \pkg{not} \citep{baranowski2016notpackage}.

The remainder of this paper is organised as follows. In Section~\ref{Sec:method}, we give a more mathematical description of NOT. In particular, we consider NOT in four scenarios, each with a different form of structural change in the mean and/or variance. For the development of both theory and computation, in each scenario, we also introduce the tailor-made contrast function derived from the generalised likelihood ratio (GLR), which is used to detect features within each subsample. Theoretical properties of NOT, such as its consistency and convergence rates are also provided. In Section~\ref{Sec:NOTsSIC}, we propose to use NOT with the strengthened Schwarz Information Criterion (sSIC) and discuss its computational aspects. Section~\ref{Sec:noise} discusses possible extensions of NOT. A comprehensive simulation study is carried out in Section~\ref{Sec:simulation_study}, where we compare NOT with the state-of-art change-point detection tools. In Section~\ref{Sec:data}, we consider data examples of global temperature anomalies and London housing data. All proofs, as well as further discussion on computational aspects, additional simulations and real data example can be found in the online supplementary materials.


\section{The framework of NOT}
\label{Sec:method}

\subsection{Setup}
\label{Sec:setup}
To describe the main framework of NOT, we consider a simplified version of (\ref{Eq:signal+noise1}), where $\Yb = (Y_1,\ldots,Y_T)'$ is modelled through 
\begin{align}
\label{Eq:signal+noise}
 Y_{t}=f_{t}+\sigma_t \varepsilon_{t}, \quad t=1,\ldots,T,
\end{align}
where $f_{t}$ is the signal, and where $\sigma_t$ is the noise's standard deviation at time $t$. To facilitate the technical presentation of our results, in Sections \ref{Sec:method} and \ref{Sec:NOTsSIC}, we assume that $\varepsilon_{t} \stackrel{\mathrm{i.i.d.}}{\sim} \Nc(0,1)$. In Section~\ref{Sec:noise}, we extend our framework to dependent noise and other noise distributions. Numerical examples regarding all setups can be found in Section~\ref{Sec:simulation_study}.

We assume that $(f_t, \sigma_t)$ can be partitioned into $q+1$ segments, with $q$ unknown distinct change-points $0=\tau_{0} < \tau_{1} <\ldots < \tau_{q} <  \tau_{q+1} = T$. Here the value of $q$ is not pre-specified and can grow with $T$. For each $j = 1,\ldots, q+1$ and for $t = \tau_{j-1}+1,\ldots,\tau_j$, the structure of $(f_t,\sigma_t)$ is is modelled parametrically by a local (i.e. depending on $j$) real-valued $d$-dimensional parameter vector $\Thetab_j$ (with $\Thetab_j \neq \Thetab_{j-1}$), where $d$ is known and typically small. To fix ideas, in the following, we assume that each segment of $f_t$ and $\sigma_t$ follows a polynomial. In addition, we require the minimum distance between consecutive change-points to be $\ge d$ for the purpose of identifiability. (Otherwise, e.g. take $f_t$ to be  piecewise-linear with a known constant $\sigma_t$, in which case $d = 2$. If we had a segment of length 1, then we would not be able to define a line based on a single point.) In other words, $(f_t, \sigma_t)$ can be divided into $q$ different segments, each from the same parametric family of much simpler structure. Some commonly-encountered scenarios are listed below, where the following holds inside the $j$-th segment for each $j = 1,\ldots, q+1$:
\begin{enumerate}[label=(S\arabic*)]
	\item \label{Scen:change_in_mean} \textbf{Constant variance, piecewise-constant mean}:  
	
	$\sigma_t = \sigma_0$ and $f_t = \theta_j$ for $t = \tau_{j-1}+1,\ldots,\tau_j$. 
	
	\item  \label{Scen:change_in_slope}  \textbf{Constant variance, continuous and piecewise-linear mean}:
	
	$\sigma_t = \sigma_0$ and $f_t = \theta_{j,1} + \theta_{j,2} \; t $ for $t = \tau_{j-1}+1,\ldots,\tau_j$, with the additional constraint of 
	\[
	\theta_{j,1} + \theta_{j,2} \; \tau_{j} = \theta_{j+1,1} + \theta_{j+1,2} \;  \tau_{j} 
	\]
	for $j = 1,\ldots,q$.
		
	\item \label{Scen:change_in_mean_and_slope}  \textbf{Constant variance, piecewise-linear (but not necessarily continuous) mean}: 
	
	$\sigma_t = \sigma_0$ and $f_t = \theta_{j,1} + \theta_{j,2} \; t $ for $t = \tau_{j-1}+1,\ldots,\tau_j$.
	In addition,  $f_{\tau_j}+ \theta_{j,2} \neq f_{\tau_j+1}$ for $j = 1,\ldots,q$.
	
	\item \label{Scen:change_in_mean_and_variance} \textbf{Piecewise-constant variance,  piecewise-constant mean}: 
	
	$f_t = \theta_{j,1}$ and $\sigma_t = \theta_{j,2} > 0$ for $t = \tau_{j-1}+1,\ldots,\tau_j$.  
\end{enumerate}

Since $\sigma_0$ in \ref{Scen:change_in_mean}--\ref{Scen:change_in_mean_and_slope} acts as a nuisance parameter, in the rest of this manuscript, for simplicity we assume that its value is known. If it is unknown, then it can be estimated accurately using the Median Absolute Deviation (MAD) method \citep{hampel1974influence}.  More specifically, with i.i.d. Gaussian errors, the MAD estimator of $\sigma_0$ is defined as $\hat{\sigma} = \mathrm{Median} \{|Y_2-Y_1|,\ldots, |Y_T-Y_{T-1}|\}/\{\Phi^{-1}(3/4)\sqrt{2}\}$ in Scenario~\ref{Scen:change_in_mean}, and as $\hat{\sigma} = \mathrm{Median} \{|Y_1-2Y_2+Y_3|,\ldots, |Y_{T-2}-2Y_{T-1}+Y_T|\}/\{\Phi^{-1}(3/4)\sqrt{6}\}$ in Scenarios~\ref{Scen:change_in_slope} and \ref{Scen:change_in_mean_and_slope}. Here $\Phi^{-1}(\cdot)$ denotes the quantile function of the standard normal distribution. Note that the MAD estimator is robust to any change-points present in the underlying signal $f_t$, due to its combination of working with the differenced data, and its use of the median. Finally, we note that a different procedure is proposed to estimate $\sigma_0$ with dependent errors; see Section~\ref{Sec:dependentnoise} for more details.

Both the methodology and the theory developed below can readily be extended to handle more complicated cases in which the signal within the segments is non-linear (e.g. higher-order-polynomial, a case illustrated in Section~\ref{Sec:simulation_study}). In all of the above-listed scenarios, we focus on structure changes in the mean or the first two moments in the univariate setting. Nevertheless, our framework can be extended to handle multivariate observations, or other more complex structure changes such as autocovariance in time series. 

\subsection{Main idea}
\label{Sec:mainidea}
We now describe the main idea of NOT formally. In the first step, instead of directly using the entire data sample, we randomly extract subsamples, i.e. vectors $(Y_s, Y_{s+1},\ldots, Y_e)'$, where $(s,e)$ is drawn uniformly from the set of pairs of indices in $\{1,\ldots,T\} \times \{1,\ldots,T\}$ that satisfy $1 \le s < e \le T$ and $e - s > 2(d - 1)$. 
Let $\ell(Y_s,\ldots,Y_e; \Thetab)$ be the likelihood of $\Thetab$ given $(Y_s, \ldots, Y_e)'$.  We then compute the generalised log-likelihood ratio (GLR) statistic for all potential single change-points within the subsample and pick the maximum, that is,
\begin{align}
	\label{Eq:glr1}		
	\mathcal{R}_{s,e}^b(\Yb) &=  2 \log \bigg[\frac{\sup_{\Thetab^1, \Thetab^2}\big\{\ell(Y_s,\ldots,Y_b; \Thetab^1) \ell(Y_{b+1},\ldots,Y_e; \Thetab^2)\big\} }{\sup_{\Thetab}\ell(Y_s,\ldots,Y_e; \Thetab)}\bigg];  \\
	\notag 
	\mathcal{R}_{s,e}(\Yb) &= \max_{b \in \{s+d-1,\ldots,e-d\}} \mathcal{R}_{s,e}^b(\Yb).
\end{align}
If constraints are in place between $\Thetab_{j}$ and $\Thetab_{j+1}$ for any $j=1,\ldots,q$  (e.g. as in \ref{Scen:change_in_slope}), the supremum in the numerator of (\ref{Eq:glr1}) is taken over the set that only contains elements of form $\Thetab^1 \times \Thetab^2$ satisfying these constraints. Otherwise, as in \ref{Scen:change_in_mean}, \ref{Scen:change_in_mean_and_slope} and \ref{Scen:change_in_mean_and_variance}, \eqref{Eq:glr1} can be simplified to 
\begin{align*}
\mathcal{R}_{s,e}^b(\Yb) =  2 \log \bigg\{\frac{\sup_{\Thetab}\ell(Y_s,\ldots,Y_b; \Thetab) \sup_{\Thetab}\ell(Y_{b+1},\ldots,Y_e; \Thetab) }{\sup_{\Thetab}\ell(Y_s,\ldots,Y_e; \Thetab)}\bigg\}.
\end{align*}
The above procedure is repeated on $M$ randomly drawn pairs of integers $(s_1,e_1),\ldots,(s_M,e_M)$.

In the second step, we test all $\mathcal{R}_{s_m,e_m}(\Yb)$ for $m=1,\ldots,M$ against a given threshold $\zeta_{T}$. Among those significant ones, we pick the one corresponding to the interval $[s_{m^*}, e_{m^*}]$ that has the smallest length. Once a change-point is found in $[s_{m^*}, e_{m^*}]$ (i.e. $b^*$ that maximises $\mathcal{R}_{s_{m^*},e_{m^*}}^b(\Yb)$), the same procedure is then repeated recursively to the left and to the right of it, until no further significant GLRs can be found. Note that in each recursive step, one could reuse the previously drawn intervals, provided that they fall entirely within each current subsegment considered.

After the process of estimating the change-points is completed, one can estimate the signals within each segment using standard methods such as least squares or maximum likelihood. Note that the estimation of knot locations in spline regression can be viewed as a multiple change-point detection problem set in the context of polynomial segments that are continuously differentiable but have discontinuous higher order derivatives at the change-points between these segments; NOT can be used for this purpose. 

Admittedly, in our framework, one could also use a deterministic scheme (for example, that in \citet{RW2010}) to pick a sufficiently rich family of intervals for multiscale inference. However, one advantage of our approach is that through the use of randomness in drawing the intervals, we avoid having to make a subjective choice of a particular fixed design. In addition, if the number of intervals drawn later turns out to be insufficient, it is straightforward to add more intervals via our random scheme. Nevertheless, with a very large number drawn intervals, the difference in performance between the random and deterministic designs is likely to be minimal, an observation also made in \citet{fryzlewicz2014wild}. 

We end this section by remarking that \citet{CsorgoHorvath1997} present a thorough investigation of the problem of single change-point detection in the GLR framework and heuristically suggest binary segmentation as a possible device for extending this methodology to multiple change-point detection. However, as illustrated in our Section~\ref{Sec:introduction}, such an extension will only work correctly in the canonical change-point detection problem in Scenario~\ref{Scen:change_in_mean}. By contrast, our aim in introducing the NOT device is to enable the use of the GLR methodology in the problem of multiple change-point detection across a range of generalised change-point scenarios.

\subsection{Log-likelihood ratios and contrast functions}
\label{Sec:contrast_function}
In many applications, the GLR (\ref{Eq:glr1}) in NOT can be simplified with the help of ``contrast functions'' under the setting of Gaussian noise. More precisely, for every integer triple $(s,e,b)$ with $1 \leq s <  e \leq T$, our aim is to find $\mathcal{C}_{s,e}^b(\Yb)$ such that:
\begin{enumerate}[label=(\alph*)]
	\item $\argmax_{b} \mathcal{C}_{s,e}^b(\Yb) =  \argmax_{b} \mathcal{R}_{s,e}^b(\Yb)$,
	\item heuristically speaking, the value of  $ \mathcal{C}_{s,e}^b(\Yb)$ is relatively small if there is no change-point in $[s, e]$,
	\item the formulation of $\mathcal{C}_{s,e}^b(\Yb)$ mainly consists of taking inner products between the data and certain contrast vectors, which facilitates the development of both computation and theory, particularly if the contrast vectors can be taken to be mutually orthonormal.
\end{enumerate}
In the following, we give the contrast functions corresponding to \ref{Scen:change_in_mean}--\ref{Scen:change_in_mean_and_variance}. We note that this approach recovers the CUSUM statistic in \ref{Scen:change_in_mean}, which is popular in this canonical change-point detection setting. One can view the resulting statistics as generalisations of CUSUM to other scenarios.

\subsubsection{Scenario \ref{Scen:change_in_mean}}
Here $f_t$ is piecewise-constant. For any integer triple $(s,e,b)$ with $1 \leq s <  e \leq T$ and $s \le b \le e-1$, we define the contrast vector $\psib_{s,e}^{b} = \big(\psi_{s,e}^{b}(1),\ldots,\psi_{s,e}^{b}(T)\big)'$ as
\begin{align}
\label{Eq:basis_jump} 
\psi_{s,e}^{b}(t) = \begin{cases}
\sqrt{\frac{e-b}{l(b-s+1)}}, & t=s,\ldots,b\\
-\sqrt{\frac{b-s+1}{l(e-b)}}, & t=b+1,\ldots,e\\
0, & \mbox{otherwise},
\end{cases}
\end{align}
where $l=e-s+1$. Also, if $b \notin \{s,s+1,\ldots,e-1\}$, then we set $\psi_{s,e}^{b}(t) = 0$ for all $t$. As an illustration, plots of $\psib_{s,e}^{b}$ with different $(s,e,b)$ are shown in Figure~\ref{Fig:jump_basis_functions}. 

\begin{figure}[!htb]
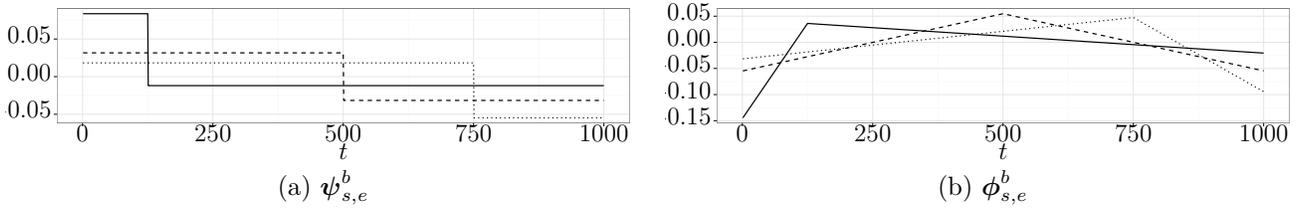

	\centering
	\null\hfill
	\subfloat[][$\psib_{s,e}^{b}$]{
			\resizebox{0.48\textwidth}{!}{\input{tikz/jump_basis_functions.tex}}
\label{Fig:jump_basis_functions}}
	\hfill
	\subfloat[][$\phib_{s,e}^{b}$]{
			\resizebox{0.48\textwidth}{!}{\input{tikz/kink_basis_functions.tex}}
\label{Fig:kink_basis_functions}}
	\hfill\null
	\caption{\label{Fig:jump_kink_vec_examples} Plots of $\psib_{s,e}^{b}$ and $\phib_{s,e}^{b}$ given by, respectively, \eqref{Eq:basis_jump} and \eqref{Eq:basis_function_kink} for $s=1$, $e=1000$ and several values of $b$. Solid line: $b=125$; dashed line: $b=500$; dotted line: $b=750$. }
\end{figure}

For any vector $\vb=(v_{1}, \ldots, v_{T})'$ we define the contrast function as $\cont{s}{e}{b}{\vb}= \abs{\ip{\vb}{ \psib_{s,e}^{b}}}$. Therefore, if $s \le b \le e-1$, then
\begin{align}
\label{Eq:contrast_const}
\cont{s}{e}{b}{\vb} =\abs{\sqrt{\frac{e-b}{l(b-s+1)}}\sum_{t=s}^{b}{v}_{t}-\sqrt{\frac{b-s+1}{l(e-b)}}\sum_{t=b+1}^{e}v_{t}}.
\end{align}
Otherwise, $\cont{s}{e}{b}{\vb} =0$. This recovers the well-known CUSUM statistic in the change-point detection literature. It can be shown that
$
	[\cont{s}{e}{b}{\Yb}]^2 = \sigma_0^2  \mathcal{R}_{s,e}^b(\Yb)
$
for every $(s,e,b)$ with $1 \le s \le b < e \le T$, thus $\cont{s}{e}{b}\cdot$ fulfills the aforementioned requirements for the contrast function. 

In addition, 
for any $1 \le s < e \le T$, we define the constant vector for the interval $[s,e]$ as
\begin{align*} 
	\mathbf{1}_{s,e}(t) = \begin{cases}	(e-s+1)^{-1/2}, & t=s,\ldots, e\\ 0, &\mbox{otherwise} \end{cases},
\end{align*}
and write $\mathbf{1}_{s,e} = \big(\mathbf{1}_{s,e}(1),\ldots,\mathbf{1}_{s,e}(T)\big)'$. Then it is easy to check that $\mathbf{1}_{s,e}$ and $\psib_{s,e}^{b}$ are orthonormal. This explains why the CUSUM is invariant to shifts in the mean.

\subsubsection{Scenario \ref{Scen:change_in_slope}}
\label{Sec:scen:change_in_slope}
Here $f_t$ is piecewise-linear and continuous. For any triple $(s,e,b)$ with $1 \leq s <  e \leq T$ and $s+1 \le b \le e - 1$, consider the contrast vector $\phib_{s,e}^{b} = \big(\phi_{s,e}^{b}(1),\ldots,\phi_{s,e}^{b}(T)\big)'$ with
\begin{align}
\label{Eq:basis_function_kink}
\phi_{s,e}^{b}(t) = \begin{cases} 
\alpha_{s,e}^{b}\beta_{s,e}^{b}  \Big[\big\{3 (b-s+1)+(e-b)-1\big\}t-\big\{b(e-s)+2s(b-s+1)\big\}\Big],   &t= s, \ldots, b \\
-\frac{\alpha_{s,e}^{b}}{\beta_{s,e}^{b}} \Big[\big\{3 (e-b)+(b-s+1)+1\big\}t-\big\{b(e-s) +2e(e-b+1)\big\}\Big],  & t= b+1,\ldots, e,\\
0, & \mbox{otherwise}.
\end{cases}
\end{align}
where $\alpha_{s,e}^{b} = \left(\frac{6}{l(l^{2}-1)(1+(e-b+1) (b-s+1)+(e-b) (b-s))}\right)^{1/2}$,	$\beta_{s,e}^{b} = \left(\frac{(e-b+1)(e-b)}{(b-s)(b-s+1)}\right)^{1/2}$ and $l = e-s+1$. If $b \notin \{s+1,\ldots,e-1\}$, then we set $\phi_{s,e}^{b}(t) = 0$ for all $t$. We illustrate the structure of $\phib_{s,e}^{b}$ in Figure~\ref{Fig:kink_basis_functions}. The contrast function is then defined as 
\begin{align}
\label{Eq:contrast_kink}
\cont{s}{e}{b}{\vb}=  \abs{\ip{\vb}{ \phib_{s,e}^{b}}}.
\end{align}

To explain the rationale behind $\phib_{s,e}^{b}$, we first define the ``linear'' vector for the interval $[s,e]$, $\gammab_{s,e}=\big(\gamma_{s,e}(1),\ldots,\gamma_{s,e}(T)\big)'$, as
\begin{align*}
\gamma_{s,e}(t) = \begin{cases}
\Big\{\frac{1}{12} (e-s+1) (e^2-2 e s+2 e+s^2-2 s)\Big\}^{-1/2} \big(t-\frac{e+s}{2}\big),  & t=s,\ldots,e\\
0, & \mbox{otherwise}
\end{cases}.
\end{align*} 
Then we have that  $\phib_{s,e}^{b}$ is orthonormal to both $\mathbf{1}_{s,e}$ and $\gammab_{s,e}$ (note that  $\gammab_{s,e}$ itself is orthonormal to $\mathbf{1}_{s,e}$). The orthonormality of the vectors $\mathbf{1}_{s,e}$, $\gammab_{s,e}$ and  $\phib_{s,e}^{b}$ is important in deriving the identity $\sigma_0^2 \mathcal{R}_{s,e}^b(\Yb) = \cont{s}{e}{b}{\Yb}^2$ below, and helps improve the numerical efficiency and stability in our implementation of NOT. In particular, it means that the contrast function is invariant to both mean shifts and slope shifts on a given interval. In fact, $\phib_{s,e}^{b}$ can be derived by (i) applying the Gram--Schmidt process on the following vector (linear with a kink at $b+1$ on $[s,e]$)
	\begin{align*}
	\tilde{\phi}_{s,e}^{b}(t) = \begin{cases} 
	t- b,   &t= b+1, \ldots, e \\
	0,      & \mbox{otherwise} \end{cases}
	\end{align*}
with respect to $\mathbf{1}_{s,e}$ and $\gammab_{s,e}$, and (ii) normalisation such that  $\|\cdot\|_{2}=1$. 	
Now write the restriction of $\vb$ on the interval $[s,e]$ as $\vb|_{[s,e]} = (0,\ldots,0,v_{s},\ldots,v_{e},0,\ldots,0)'$. Fix any $(s,e,b)$, given the restriction imposed on $\Thetab$ in \ref{Scen:change_in_slope}, the best approximation of $\Yb|_{[s,e]}$ (in the $\ell_2$ distance) with a single kink at $b$ is a linear combination of $\mathbf{1}_{s,e}$, $\gammab_{s,e}$ and  $\phib_{s,e}^{b}$ (all mutually orthonormal). Therefore,
	\begin{align*}
	 &\sigma_0^2 \mathcal{R}_{s,e}^b(\Yb) =\\
	 & \min_{a_0,a_1 \in \R} \| \Yb|_{[s,e]}- a_0 \mathbf{1}_{s,e} - a_1 \gammab_{s,e} \|_2^2 - \min_{a_0,a_1,a_2 \in \R} \| \Yb|_{[s,e]} - a_0 \mathbf{1}_{s,e} - a_1 \gammab_{s,e} - a_2 \phib_{s,e}^{b}\|_2^2 \\
	& =  \| \Yb|_{[s,e]} -  \langle \Yb,  \gammab_{s,e}\rangle  \gammab_{s,e} - \langle \Yb,  \mathbf{1}_{s,e} \rangle  \mathbf{1}_{s,e} \|^2 - \|	\Yb|_{[s,e]} -  \langle\Yb, \phib_{s,e}^{b}\rangle \phib_{s,e}^{b} - \langle \Yb,  \gammab_{s,e}\rangle  \gammab_{s,e} - \langle \Yb,  \mathbf{1}_{s,e} \rangle  \mathbf{1}_{s,e} \|^2 \\
	& = \langle\Yb, \phib_{s,e}^{b}\rangle ^2 = \cont{s}{e}{b}{\Yb}^2.
	\end{align*}
Thus the aforementioned requirements for the contrast function are satisfied.

\subsubsection{Scenario \ref{Scen:change_in_mean_and_slope}}
Here $f_{t}$ is a piecewise-linear but not necessarily continuous function. We use the following contrast function for any $s < b < e$:
\begin{align}
\label{Eq:basis_jump_plus_kink}
\cont{s}{e}{b}{\vb} =\left(\ip{\vb}{\psib_{s,e}^{b}}^2+\ip{\vb}{\gammab_{s,b}}^2 + \ip{\vb}{\gammab_{b+1,e}}^2 - \ip{\vb}{\gammab_{s,e}}^2\right)^{1/2}.
\end{align}
This construction is justified by noting that
\begin{align*}
\sigma_0^2 \mathcal{R}_{s,e}^b(\Yb)  &= \min_{a_0,a_1 \in \R} \| \Yb|_{[s,e]}- a_0 \mathbf{1}_{s,e} - a_1 \gammab_{s,e} \|_2^2 \\
& \quad- \bigg(\min_{a_0,a_1 \in \R} \| \Yb|_{[s,b]} - a_0 \mathbf{1}_{s,b} - a_1 \gammab_{s,b} \|_2^2 + \min_{a_0,a_1 \in \R} \| \Yb|_{[b+1,e]} - a_0 \mathbf{1}_{b+1,e} - a_1 \gammab_{b+1,e} \|_2^2\bigg)\\
& = \cont{s}{e}{b}{\Yb}^2,
\end{align*}
where we also used the orthonormality among $\mathbf{1}_{s,e}$, $\psib_{s,e}^{b}$, $\gammab_{s,b}$ and  $\gammab_{b+1,e}$ in the above derivation.

\subsubsection{Scenario \ref{Scen:change_in_mean_and_variance}}
Here both $f_t$ and $\sigma_{t}$ are piecewise-constant. For any $1 \leq s + 1 < b < e -1 \leq T$, we propose  
\begin{align}
\label{Eq:contrast_volatility}
\cont{s}{e}{b}{\Yb} =  (e-s+1)\log\left(\hat{\sigma}_{s,e}(\Yb)\right) - (b-s+1)\log\left(\hat{\sigma}_{s,b}(\Yb)\right) - (e-b)\log\left(\hat{\sigma}_{b+1,e}(\Yb)\right), 
\end{align}
where 
\[
\hat{\sigma}_{s,e}^2(\Yb) = \frac{1}{e-s+1}\sum_{t=s}^{e}\left(Y_{t}-\frac{1}{e-s+1}\sum_{t=s}^{e} Y_{t} \right)^2 = \ip{\Yb^2}{\mathbf{1}_{s,e}^2} - \ip{\Yb}{\mathbf{1}_{s,e}^2}^2.
\]
Otherwise, for $b\not\in \{s+2,\ldots, e-2\}$, we set $\cont{s}{e}{b}{\Yb}=0$.
In this Scenario, it is straightforward to verify that $\cont{s}{e}{b}{\Yb} = \mathcal{R}_{s,e}^b(\Yb)$. (N.B. $\mathbf{1}_{s,e}^2 \neq \mathbf{1}_{s,e}$ due to the normalising constant.) In practice, for numerical stability, we use $\log_\epsilon(\cdot) := \log\{\max(\cdot,\epsilon)\}$ instead of $\log(\cdot)$ in \eqref{Eq:contrast_volatility} with a small given $\epsilon>0$.

\subsection{The NOT algorithm}
\label{Sec:not_algorithm}

	Here we present a generic version of the NOT algorithm. Its pseudo-code can be found below.   The main ingredient of the NOT procedure is a contrast function $\cont{s}{e}{b}{\cdot}$, chosen by the user, depending on the assumed nature of change-points in the data, e.g. as exemplified by our scenarios \ref{Scen:change_in_mean}--\ref{Scen:change_in_mean_and_variance} above. In addition, some tuning parameters are needed: $\zeta_{T}>0$ is the threshold with respect to which the contrast should be tested, while $M$ is the number of the intervals drawn in the procedure. Guidance on the choice of $\zeta_{T}$ and $M$ is given in Section~\ref{Sec:NOTsSIC}.

\begin{algorithm}[!htb]
	\caption{NOT}
	\label{Alg:not_algorithm}

	\begin{algorithmic}
		\Require Data vector $\Yb=(Y_{1},\ldots, Y_{T})'$,  $F_{T}^{M}$ being a set of $M$ intervals, with each pair of start- and end- points drawn independently and uniformly from the set of pairs of indices in $\{1,\ldots,T\} \times \{1,\ldots,T\}$ that satisfy the conditions outlined at the beginning of Section~\ref{Sec:mainidea}, $\Sc=\emptyset$.
		\Ensure Set of estimated change-points $\Sc\subset\{1,\ldots,T\}$. \\
		
	    \vspace{0.01cm} \\
	    \hspace{-0.43cm} \textbf{To start the algorithm}: Call  \Call{NOT}{1, $T$, $\zeta_T$}\\

		\Procedure{NOT}{$s,e,\zeta_{T}$}	
		\If{$e-s < 1$} STOP
		\Else
		\State $\Mc_{s,e}:=\left\{m:[s_{m},e_{m}]\in F_{T}^{M}, [s_{m},e_{m}]\subset[s,e]\right\}$
		\If{$\Mc_{s,e}=\emptyset$} STOP
		\Else
		\State $\Oc_{s,e}:=\left\{m\in\Mc_{s,e}: \max_{s_{m}\le b \le e_{m}} \cont{s_m}{e_m}{b}{\Yb}  >  \zeta_{T}\right\}$
		\If{$\Oc_{s,e}=\emptyset$} STOP
		\Else
		\State $m^{*}:\in\argmin_{m\in\Oc_{s,e}}|e_{m}-s_{m}|$ 
		\State $b^{*}:=\argmax_{s_{m^*} \le b \le e_{m^*}} \cont{s_{m^*}}{e_{m^*}}{b}{\Yb}$
		\State $\Sc:=\Sc\cup\{b^{*}\}$
		\State \Call{NOT}{$s,b^{*},\zeta_{T}$}
		\State \Call{NOT}{$b^{*}+1,e,\zeta_{T}$}
		\EndIf
		\EndIf
		\EndIf
		\EndProcedure
	\end{algorithmic}
	
\end{algorithm}

To sum up, the input include the data vector $\mathbf{Y}$,  the set of $F_T^M$ that contains all randomly drawn sub-intervals for testing, and the global variable $\mathcal{S}$ for the set of estimated change-points initialised with $\mathcal{S} = \emptyset$. Then NOT is started recursively with $[s,e] = [1,T]$ and a given $\zeta_T$.

Here the entire set of $F_{T}^{M}$  that contains all random intervals is generated before we start running Algorithm~\ref{Alg:not_algorithm}. In this way, we are better able to control the computational complexity of the entire procedure. If we were to draw new intervals
each time after a change-point was detected, the computational complexity would
depend to a larger extent on the number of change-points. Furthermore, if we were to draw anew after each detection, we would likely be forfeiting some of the intervals drawn before, which would result in a waste of computational effort.


\subsection{Theoretical properties of NOT}
\label{Sec:theory}
In this section, we analyse the theoretical behaviour of the NOT algorithm in Scenarios \ref{Scen:change_in_mean} and \ref{Scen:change_in_slope}. We cover the case of infill asymptotics, which is standard in the literature on a posteriori change-point detection. An attractive feature of our methodology is that proofs for other scenarios can in principle be constructed ``at home'' by the user, by following the same generic proof strategy as the one we use for these two scenarios. 

First, we revisit the canonical change-point detection problem, \ref{Scen:change_in_mean}, where the signal vector $\fb=(f_{1},\ldots,f_{T})'$ is piecewise-constant.  Here $\sigma_0$ is assumed to be known. Otherwise, one can plug in the MAD estimator, described in Section~\ref{Sec:setup}, without affecting the correctness of our theory. For notational convenience, we set $\sigma_{0} = 1$. For other values of $\sigma_0$, our theorems are still valid with only minor adjustments to the constants therein. Explicit expressions for the constants are given in Section~\ref{Sec:proof_ext} of the online supplementary materials. 

\begin{Theorem}
\label{Thm:consistency_gaussian_const}
Suppose $Y_{t}$ follow model (\ref{Eq:signal+noise}) in Scenario \ref{Scen:change_in_mean}. Let $\delta_{T} = \min_{j=1,\ldots,q+1}(\tau_{j}-\tau_{j-1})$, $\Delta_j^\fb = |f_{\tau_{j}+1}-f_{\tau_{j}}|$, $\fl_{T} = \min_{j = 1,\ldots,q}\Delta_j^\fb$. Let $\hat{q}$ and $\hat{\tau}_{1},\ldots,\hat{\tau}_{\hat{q}}$ denote, respectively, the number and locations of change-points, sorted in increasing order, estimated by Algorithm~\ref{Alg:not_algorithm} with the contrast function given by \eqref{Eq:contrast_const}.  Then there exist constants $\Cl$, $C_1, C_2, C_3 > 0$ (not depending on $T$) such that given $\delta_{T}^{1/2} \fl_{T} \geq \Cl \sqrt{\log T}$, $C_1 \sqrt{\log T } \leq \zeta_{T} < C_2 \delta_{T}^{1/2} \fl_{T}$ and  $M \geq 36 T^2 \delta_T^{-2} \log (T^2 \delta_T^{-1})$, as $T \rightarrow \infty$,
\begin{align}
\label{Eq:consistency_gaussian_const}
\Pb{\hat{q}=q,\; \max_{j=1,\ldots,q} \Big(|\hat{\tau}_{j}-\tau_{j}|(\Delta_j^\fb)^2 \Big) \leq C_3 \log T} \rightarrow 1.
\end{align}
\end{Theorem}

In the simplest case where we have finitely many change-points with $\delta_T \sim T$ and $\fl_{T} \sim 1$, then $\delta_{T}^{1/2} \fl_{T} \sim \sqrt{T}$ so the condition $\delta_{T}^{1/2} \fl_{T} \geq \Cl \sqrt{\log T}$ is always satisfied for a sufficiently large $T$. We need $M = O(\log T)$ many random intervals for consistent detection of all the change-points, which leads to a total computational cost of $O(T \log T)$ for the entire procedure. Furthermore,  $\max_{j=1,\ldots,q} \Big(|\hat{\tau}_{j}-\tau_{j}|\Big) = O_{p}(\log T)$, which trails the minimax rate of $O_p(1)$ by only a logarithmic factor. In addition, we note that the NOT procedure allows for $\delta_{T}^{1/2} \fl_{T}$, a quantity that characterises the difficulty level of the problem, to be of order $\sqrt{\log T}$. As argued in \cite{chan2013detection}, this is the smallest rate that permits change-point detection for any method from a minimax perspective.

Next, we revisit Scenario~\ref{Scen:change_in_slope}, in which the signal is piecewise-linear and continuous. Again, we set $\sigma_{0}=1$ for notational convenience. Explicit expressions of the constants in the following theorem can be found in Section~\ref{Sec:proof_of_consistency_theorem_slope} of the online supplementary materials.

\begin{Theorem}
\label{Thm:consistency_gaussian_linear}
Suppose $Y_{t}$ follow model (\ref{Eq:signal+noise}) in Scenario~\ref{Scen:change_in_slope}. Let $\delta_{T} = \min_{j=1,\ldots,q+1}(\tau_{j}-\tau_{j-1})$, $\Delta_j^\fb =  |2f_{\tau_{j}}-f_{\tau_{j}-1}-f_{\tau_{j}+1}|$, $\fl_{T} = \min_{j = 1,\ldots,q}\Delta_j^\fb$. Let $\hat{q}$ and $\hat{\tau}_{1},\ldots,\hat{\tau}_{\hat{q}}$ denote, respectively, the number and locations of change-points, sorted in increasing order, estimated by Algorithm~\ref{Alg:not_algorithm} with the contrast function given by \eqref{Eq:contrast_kink}.  Then there exist constants $\Cl, C_1, C_2, C_3 > 0$  (not depending on $T$)  such that given $\delta_{T}^{3/2} \fl_{T} \geq  \Cl \sqrt{\log T}$, $C_1 \sqrt{\log T } \leq \zeta_{T} < C_2 \delta_{T}^{3/2} \fl_{T}$ and $M \geq 36 T^2 \delta_T^{-2} \log (T^2 \delta_T^{-1})$, as $T \rightarrow \infty$,
\begin{align}
\label{Eq:consistency_gaussian_linear}
\Pb{\hat{q}=q,\; \max_{j=1,\ldots,q} \Big(|\hat{\tau}_{j}-\tau_{j}|(\Delta_j^\fb)^{2/3}
 \Big) \leq C_3 (\log T)^{1/3}} \rightarrow 1.
\end{align}
\end{Theorem}

In the case in which we have finitely many change-points with $\delta_T \sim T$, we again need $M = O(\log T)$ random intervals for consistent estimation of all the change-points, leading to the total computational cost of $O(T \log T)$. In addition, when $\fl_{T} \sim T^{-1}$ (a case in which $f_t$ is bounded), our theory indicates that the resulting change-point detection rate is $O_p(T^{2/3} (\log T)^{1/3})$, which is different from the rate of $O_p(T^{2/3})$ derived by \citet{raimondo1998minimax} by only a logarithmic factor; moreover, under additional assumptions and with a more careful but restrictive choice of $\zeta_{T}$, this rate can be further improved to $O_p(T^{1/2} (\log T)^{1/2})$; see Section~\ref{Sec:sSIC_theory} and Lemma~\ref{lem:consistency_gaussian_linear} in the online supplementary materials for more details. Furthermore, we remark that in more general cases  (i.e. number of change-points increasing with $T$)  in Scenario~\ref{Scen:change_in_slope}, the difficulty level of the problem in Scenario~\ref{Scen:change_in_slope} can be charaterised by $\delta_{T}^{3/2} \fl_{T}$, a quantity analogous to $\delta_{T}^{1/2} \fl_{T}$ in the setting of \ref{Scen:change_in_mean}.

Finally, we emphasise again that in contrast, the WBS will fail to estimate change-point consistently in Scenario~\ref{Scen:change_in_slope}, for reasons described in Section~\ref{Sec:introduction}.

\section{NOT with the strengthened Schwarz Information Criterion (sSIC)}
\label{Sec:NOTsSIC}
\subsection{Motivation}
\label{Sec:NOTsSIC_motivation}
The success of Algorithm~\ref{Alg:not_algorithm} depends on the choice of the threshold $\zeta_T$. Although Theorem~\ref{Thm:consistency_gaussian_const} and Theorem~\ref{Thm:consistency_gaussian_linear} state that there exists $\zeta_T$ that guarantee consistent estimation of the change-points, this choice still typically depends on some unobserved quantities; furthermore, there are many more general scenarios where a theoretical optimal threshold might be difficult to derive.

Note that for a given $\Yb$ and $F_T^M$, each threshold $\zeta_T$ corresponds to a candidate model produced by NOT. Therefore, if we could produce a ``solution path'' of candidate models obtained from NOT along all possible thresholds, we could then try to select the best model along the solution path via minimising an information-based criterion. In this sense, here the task of selecting the best threshold is equivalent to selecting the best model.  

The idea of a ``solution path'' has also been widely used in high-dimensional statistics. See, for instance, the work of \citet{EHJT2004} for the lasso and \citet{TibshiraniTaylor2011} for the generalised lasso. However, since our NOT procedure does not have a convex objective function to optimise, the algorithm we developed in the following is different from those developed for the high-dimensional problems.

\subsection{The NOT solution path algorithm}
\label{Sec:solution_path_algorithm} 
Denote by $\Tc(\zeta_{T})=\{\hat{\tau}_{1}(\zeta_{T}), \ldots, \hat{\tau}_{\hat{q}(\zeta_T)}(\zeta_T)\}$ the locations of change-points estimated by Algorithm~\ref{Alg:not_algorithm} with threshold $\zeta_T$ 
and define the threshold-indexed solution path as the family of sets $\{\Tc(\zeta_{T})\}_{\zeta_T\geq 0}$. 
Note that this threshold-indexed solution path has the following important properties. First, being seen as the function $\zeta_{T}\mapsto\Tc(\zeta_{T})$, it changes its value only at discrete points, i.e. there exist $0= \zeta_{T}^{(0)} < \zeta_{T}^{(1)}< \ldots < \zeta_{T}^{(N)}$, such that  $\Tc(\zeta_{T}^{(i)}) \neq \Tc(\zeta_{T}^{(i+1)})$ for any $i=0,1,\ldots,N-1$, and $\Tc(\zeta_{T}) = \Tc(\zeta_{T}^{(i)})$ for any $\zeta_{T} \in [\zeta_{T}^{(i)},\zeta_{T}^{(i+1)})$; and second,  $T(\zeta_{T})= \emptyset$ for any $\zeta_{T} \geq \zeta_{T}^{(N)}$. 

However, the thresholds $\zeta_{T}^{(i)}$ are unknown and depend on the data, therefore naively applying Algorithm~\ref{Alg:not_algorithm} on a range of pre-specified thresholds typically does not recover the entire solution path. Moreover, from the computational point of view, repeated application of Algorithm~\ref{Alg:not_algorithm} to find the solution path is not optimal either, because intuitively one would expect the solutions for $\zeta_{T}^{(i+1)}$ and  $\zeta_{T}^{(i)}$  to be similar for most $i$.  
These issues are circumvented via our newly developed Algorithm~\ref{Alg:not_solution_path}, which is able to compute the entire threshold-indexed solution path quickly, thus facilitating the study of a data-driven approach to the choice of $\zeta_{T}$ in Section~\ref{Sec:sSIC}. The key idea of Algorithm~\ref{Alg:not_solution_path} is to make use of information from $\Tc(\zeta_{T}^{(i)})$ to compute both $\zeta_{T}^{(i+1)}$ and $\Tc(\zeta_{T}^{(i+1)})$ iteratively for every $i=0,\ldots,N-1$. The pseudo-code of Algorithm~\ref{Alg:not_solution_path}, as well as other relevant details, can be found in Section~\ref{Sec:solution_path_algorithm_details} of the online supplementary materials.

\subsection{Choice of $\zeta_{T}$ via the strengthened Schwarz Information Criterion (sSIC)}
\label{Sec:sSIC}

Suppose we have $\Tc(\zeta^{(1)}), \ldots,\Tc(\zeta^{(N)})$ that form the NOT solution path, i.e. the collection of candidate models produced by Algorithm~\ref{Alg:not_solution_path}. We propose to select $\Tc(\zeta^{(k)})$ that minimises the strengthened Schwarz Information Criterion (sSIC; \citet{LWZ1997}, \citet{fryzlewicz2014wild}) defined as follows. Let  $k=1,\ldots,N$, $\hat{q}_{k}=|\Tc(\zeta^{(k)}_T)|$ and $\hat{\Thetab}_{1}, \ldots, \hat{\Thetab}_{\hat{q}_{k}+1}$ be the maximum likelihood estimators of the segment parameters in model \eqref{Eq:signal+noise} with the estimated change-points $\hat{\tau}_{1}, \ldots, \hat{\tau}_{\hat{q}_{k}} \in \Tc(\zeta^{(k)}_T)$. Here for notational convenience, we have suppressed the dependence of $\hat{\tau}_{1}, \ldots, \hat{\tau}_{\hat{q}_{k}}$ on  $\zeta^{(k)}_T$. Further, denote by $n_{k}$ the total number of estimated parameters, including the number of free parameters in $\Thetab_{1}, \ldots, \Thetab_{\hat{q}_{k}+1}$ (N.B. this can be different from the dimensionality of each $\Thetab_{j}$ multiplied by the number of segments, as e.g. in \ref{Scen:change_in_slope}). Then the strengthened Schwarz Information Criterion (sSIC) is 
\begin{align}
\label{Eq:SIC} 
\mbox{sSIC}(k) = -2 \sum_{j=1}^{\hat{q}_{k}+1} \ell(Y_{\hat{\tau}_{j-1}+1},\ldots,Y_{\hat{\tau}_{j}}; \hat{\Thetab}_{j})+ n_{k} \log^\alpha(T),
\end{align}
for some pre-given $\alpha \ge 1$, with $\hat{\tau}_{0}=0$ and $\hat{\tau}_{\hat{q}_{k}+1}=T$. When $\alpha=1$, we recover the well-known Schwarz Information Criterion (SIC). 

One of the reasons we use sSIC here is to facilitate our theoretical development below. In fact, once we obtain the NOT solution path via Algorithm~\ref{Alg:not_solution_path}, other information criteria, such as MBIC \citep{ZhangSiegmund2007} or Minimum Description Length (MDL; \citet{Davis2006}), could conceivably be used for model (or equivalently, threshold) selection.

\subsection{Theoretical properties of NOT with the sSIC}
\label{Sec:sSIC_theory}
In this section, we analyse the theoretical behaviour of NOT with the sSIC in Scenarios \ref{Scen:change_in_mean} and \ref{Scen:change_in_slope}. Here we focus on the situation where the number of change-points $q$ is fixed (i.e. does not increase with $T$) and the spacings between consecutive change-points are large (i.e. $\sim T$). This is typical for the theoretical development of information-criteron-based approaches, and reflects the fact that such approaches tend to work better in practice for signals with a moderate number of change-points with large spacings between them. See also \citet{Yao1988}. Again, for notational convenience, we set $\sigma_0 = 1$. 
Our results below provide theoretical justifications for using NOT with the sSIC. In contrast to Algorithm~\ref{Alg:not_algorithm}, here one does not need to supply a threshold. 

\begin{Theorem}
	\label{Thm:consistency_gaussian_const_sSIC}
	Suppose $Y_{t}$ follow model (\ref{Eq:signal+noise}) in Scenario \ref{Scen:change_in_mean}. Let $\delta_{T} = \min_{j=1,\ldots,q+1}(\tau_{j}-\tau_{j-1})$, $\Delta_j^\fb = |f_{\tau_{j}+1}-f_{\tau_{j}}|$ and $\fl_{T} = \min_{j = 1,\ldots,q}\Delta_j^\fb$. Furthermore, assume that $q$ is fixed, $\delta_{T}/T \ge {\Cl}_1$, $\fl_{T} \ge {\Cl}_2$ and $\max_{t=1,\ldots,T} |f_t| \le \bar{C}$ for some ${\Cl}_1,{\Cl}_2, \bar{C} > 0$. Let $\hat{q}$ and $\hat{\tau}_{1},\ldots,\hat{\tau}_{\hat{q}}$ denote, respectively, the number and locations of change-points, sorted in increasing order, estimated by NOT (via Algorithm~\ref{Alg:not_solution_path}) with the contrast function given by \eqref{Eq:contrast_const} and $\zeta_T$ picked via the sSIC using $\alpha > 1$. Then there exists a constant $C$ (not depending on $T$) such that given $M \geq 36 {\Cl}_1^{-2} \log ({\Cl}_1^{-1} T)$, 
	\begin{align*}
	\Pb{\hat{q}=q,\; \max_{j=1,\ldots,q} |\hat{\tau}_{j}-\tau_{j}| \leq C \log T} \rightarrow 1,
	\end{align*}
	as $T \rightarrow \infty$.
\end{Theorem}

\begin{Theorem}
	\label{Thm:consistency_gaussian_linear_sSIC}
	Suppose $Y_{t}$ follow model (\ref{Eq:signal+noise}) in Scenario~\ref{Scen:change_in_slope}. Let $\delta_{T} = \min_{j=1,\ldots,q+1}(\tau_{j}-\tau_{j-1})$, $\Delta_j^\fb =  |2f_{\tau_{j}}-f_{\tau_{j}-1}-f_{\tau_{j}+1}|$, $\fl_{T} = \min_{j = 1,\ldots,q}\Delta_j^\fb$. Furthermore, assume that $q$ is fixed, $\delta_{T}/T \ge {\Cl}_1$, $\fl_{T} T \ge {\Cl}_2$ and $\max_{t=1,\ldots,T} |f_t| \le \bar{C}$ for some ${\Cl}_1,{\Cl}_2, \bar{C} > 0$. Let $\hat{q}$ and $\hat{\tau}_{1},\ldots,\hat{\tau}_{\hat{q}}$ denote, respectively, the number and locations of change-points, sorted in increasing order, estimated  by NOT (via Algorithm~\ref{Alg:not_solution_path}) with the contrast function given by \eqref{Eq:contrast_kink} and $\zeta_T$ picked via the sSIC using $\alpha > 1$.  Then there exists a constant $C$  (not depending on $T$)  such that given $M \geq 36 {\Cl}_1^{-2} \log ({\Cl}_1^{-1} T)$, 
	\begin{align*}
	\Pb{\hat{q}=q,\; \max_{j=1,\ldots,q} |\hat{\tau}_{j}-\tau_{j}| \leq C \sqrt{T\log T}} \rightarrow 1,
	\end{align*}
	as $T \rightarrow \infty$.
\end{Theorem}

For a discussion of the optimality of the rates obtained in Theorems~\ref{Thm:consistency_gaussian_const_sSIC} and \ref{Thm:consistency_gaussian_linear_sSIC} regarding the accuracy of the estimated change-point locations, see Section~\ref{Sec:theory}.

\subsection{Computational complexity}
\label{Sec:computational_complexity}
Here we elaborate on the computational complexity of Algorithms~\ref{Alg:not_algorithm} and \ref{Alg:not_solution_path}. For both algorithms, the task of computation can be divided into two main parts. First, we need to evaluate a chosen contrast function for all points in the $M$ randomly picked intervals with their endpoints in $\{1,\ldots,T\}$. In the second part, we find potential locations of the change-points for a single threshold $\zeta_{T}$ in the case of Algorithm~\ref{Alg:not_algorithm} and for all possible thresholds in the case of Algorithm~\ref{Alg:not_solution_path}.

Naturally, the computational complexity of the first part depends on the cost of computing the contrast function for a single interval. In all scenarios studied in this paper, this cost is linear in the length of the interval, i.e.  the cost of computing $\{ \cont{s}{e}{b}{\Yb} \}_{b=s}^{e-1}$ is $O(e-s+1)$. This is explained in detail in Section~\ref{Sec:contrast_fun_linear_time} of the online supplementary materials. The intervals drawn in the procedures have approximately $O(T)$ points on average, therefore the computational complexity of the first part of the computations is $O(MT)$ in a typical application. Importantly, as the calculations for one interval are completely independent of the calculations for another, it is straightforward to run these computations in parallel.  
In addition, for the second part, as mentioned in detail in the Section~\ref{Sec:solution_path_algorithm_details} of online supplementary materials, its computational complexity is typically less than $O(MT)$, thus bringing the total computational complexity of both Algorithm~\ref{Alg:not_algorithm} and Algorithm~\ref{Alg:not_solution_path} to $O(MT)$.

Figure~\ref{Fig:computation_time} shows execution times for the implementation of Algorithm~\ref{Alg:not_solution_path} available in the \proglang{R} package \pkg{not}, with the data $Y_{t}$, $t=1,\ldots, T$, being i.i.d. $\mathcal{N}(0,1)$. The running times appears to scale linearly both in $T$ (Figure~\ref{Fig:computation_time_sol_path_fixed_M}) and in $M$ (Figure~\ref{Fig:computation_time_sol_path_fixed_T}), which provides evidence that the computational complexity of Algorithm~\ref{Alg:not_solution_path} in this particular example is practically of order $O(MT)$. 

Finally, we remark that the memory complexity of Algorithm~\ref{Alg:not_solution_path} is also $O(MT)$, which combined with its low computational complexity implies that our approach can handle problems of size $T$ in the range of millions.

\begin{figure}[!ht]
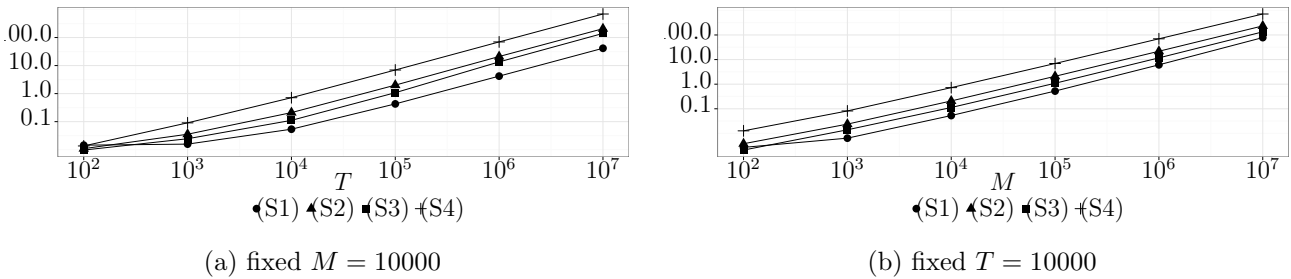

	\null\hfill
	\subfloat[][fixed $M=10000$]{
			\resizebox{0.48\textwidth}{!}{\input{tikz/computation_time_sol_path_fixed_M.tex}}
\label{Fig:computation_time_sol_path_fixed_M}}
	\hfill
	\subfloat[][fixed $T=10000$]{
			\resizebox{0.48\textwidth}{!}{\input{tikz/computation_time_sol_path_fixed_T.tex}}
\label{Fig:computation_time_sol_path_fixed_T}}
	\hfill\null
	\caption{\label{Fig:computation_time} Execution times (in seconds) for the implementation of Algorithm~\ref{Alg:not_solution_path} available in \proglang{R} package \pkg{not} \citep{baranowski2016notpackage}, for various feature detection problems with the data $Y_{t}$, $t=1,\ldots,T$ being i.i.d. $\Nc(0,1)$. In a single run, computations for the input of the algorithm are performed in parallel, using 8 virtual cores of an Intel Xeon 3.6 GHz CPU with 16 GB of RAM. The computation times are averaged over 10 runs in each case.}
\end{figure}

\subsection{Other practical considerations}
\subsubsection{Choice of $M$}
\label{Sec:parameter_choice_M}
As can be seen in Theorem~\ref{Thm:consistency_gaussian_const} and Theorem~\ref{Thm:consistency_gaussian_linear}, the minimum required value for $M$ typically grows with $T$ (i.e. for a fixed number of change-points, at $O(\log T)$). In practice, when the number of observations is of the order of thousands, we would recommend setting $M=10000$. 
 With this value of $M$, the implementation of Algorithm~\ref{Alg:not_algorithm} provided in the \proglang{R} \pkg{not} package \citep{baranowski2016notpackage} achieves the average computation time not longer than $2$ seconds in all examples in Section~\ref{Sec:simulation_study} using a single core of an Intel Xeon 3.6 GHz CPU. This can be accelerated further, as the \pkg{not} package allows for computing the contrast function over the intervals drawn in parallel using all available CPU cores.
However, caution must be exercised for signals with a large expected number of change-points, for which $M$ may need to be increased. For example, \citet{MFL2017} found that NOT with $M=10^5$ offered better practical performance on the change-point-rich signals they considered. 

\subsubsection{Early stopping for NOT with the sSIC}
Note that if the number of change-points in the data is expected to be rather moderate, then it may not be necessary to calculate sSIC for all $k$. In practice, solutions on the path corresponding to very small values of $\zeta_{T}$ contain many estimated change-points.
Such solutions are unlikely to minimise \eqref{Eq:SIC}. Therefore by considering $|\Tc(\zeta^{(k)}_T)|\leq q_{max}$ we could achieve some computational gains, without adversely impacting the overall performance of the methodology. As such, in all applications presented in this work we compute sSIC only for $k$ such that $|\Tc(\zeta^{(k)}_T)| \leq q_{max}$ with $q_{max}=25$.

\section{NOT with dependent or heavy-tailed noise}
\label{Sec:noise}
\subsection{NOT with dependent noise}
\label{Sec:dependentnoise}
When the errors $\varepsilon_t$ in model (\ref{Eq:signal+noise}) are dependent with $\mathbb{E}\varepsilon_t=0$ and $\mathrm{Var}(\varepsilon_t)=1$, the aforementioned NOT procedure can still be applied as a quasi-likelihood-type procedure. Conceivably, using NOT here would incur information loss. As is shown in Corollaries~\ref{Cor:consistency_gaussian_const} and~\ref{Cor:consistency_gaussian_linear} in Scenarios~\ref{Scen:change_in_mean} and \ref{Scen:change_in_slope}, NOT is still consistent if we replace the noise's i.i.d. assumption in Theorems~\ref{Thm:consistency_gaussian_const} and~\ref{Thm:consistency_gaussian_linear} by stationarity with short-memory. This new dependence assumption is satisfied by a large class of stationary time series models, including autoregressive moving average (ARMA) models. See also numerical examples in Section~\ref{Sec:sim_additional} of the online supplementary materials. Again we assume that $\sigma_0$ is known. However, if not, MAD-type estimators based on the simple differencing are no longer appropriate for dependent data. We comment on this issue after the corollaries.
\begin{Corollary}
	\label{Cor:consistency_gaussian_const}
	Suppose $Y_{t}$ follow model (\ref{Eq:signal+noise}) in Scenario \ref{Scen:change_in_mean}, but with $\{\varepsilon_t\}$ being a stationary short-memory Gaussian process, i.e. the auto-correlation function of $\{\varepsilon_t\}$, denoted by $\rho_k$ for any lag $k \in \mathbb{Z}$, satisfies $\sum_{k=-\infty}^{\infty} |\rho_k| < \infty$. Then, the conclusion of Theorem~\ref{Thm:consistency_gaussian_const} still holds (with different constants).
\end{Corollary}

\begin{Corollary}
	\label{Cor:consistency_gaussian_linear}
	Suppose $Y_{t}$ follow model (\ref{Eq:signal+noise}) in Scenario \ref{Scen:change_in_slope}, but with $\{\varepsilon_t\}$ being a stationary short-memory Gaussian process. The conclusion of Theorem~\ref{Thm:consistency_gaussian_linear} holds (with different constants).
\end{Corollary}

In our theoretical development for the dependent noise setting, the smallest permitted threshold to be used in the NOT algorithm depends linearly on $\sigma_0(\sum_{k=-\infty}^{\infty} |\rho_k|)^{1/2}$. 
This quantity can also be viewed as a generalisation to the independent noise setting, where the threshold is proportional to $\sigma_0$ (since $\sum_{k=-\infty}^{\infty} |\rho_k| = 1$). More details of its derivation is provided in Section~\ref{Sec:proof_of_consistency_cor_const} of the online supplementary materials. 

This poses a few challenges in the practical application of NOT to signals with dependent noise: (i) the (pre-)estimation of the residuals $\varepsilon_t$; (ii) the estimation $\sigma_0$, if unknown; and (iii) the estimation of $\sigma_0(\sum_{k=-\infty}^{\infty} |\rho_k|)^{1/2}$. These problems are known to be difficult in time series analysis in general. Possible solutions are outlined below. 

For (i), we have had some success with the wavelet-based method of \citet{JohnstoneSilverman1997}, which was implemented in \proglang{R} package \pkg{wavethresh} \citep{Nason2016}; its advantages are that it is specifically designed for dependent noise and that, being based on nonlinear wavelet shrinkage, it is particularly suited for signals with irregularities, such as (generalised) change-points. Here the Haar wavelet transform of the data is appropriate in Scenario~\ref{Scen:change_in_mean}, while a transform with respect to any wavelet that annihilates linear functions is appropriate in Scenarios~\ref{Scen:change_in_slope} and \ref{Scen:change_in_mean_and_slope}. Once the empirical residuals are obtained from (i), we could then estimate $\sigma_0$ in (ii) by its sample version, and estimate $\sigma_0(\sum_{k=-\infty}^{\infty} |\rho_k|)^{1/2}$ in (iii) in a model-based way (e.g. using the autoregressive model with its order $p$ chosen by an information criterion).

\subsection{Extension of NOT to heavy-tailed noise}
NOT appears to be relatively robust under noise misspecification. As is demonstrated later in Section~\ref{Sec:simulation_study}, it offers reasonable estimates when the noise is non-Gaussian but the Gaussian contrast functions are used. We now discuss how its performance can be improved further in the presence of heavy-tailed noise. 

In Scenario \ref{Scen:change_in_mean}, we propose to apply the following new contrast function, defined for $\Yb$ and $1\leq s\leq b<e<T$ as 
\begin{align}
\label{Eq:contrast_const_ht}
\contt{s}{e}{b}{\Yb} = \ip{\mathcal{S}_{s,e}(\Yb)}{\psib_{s,e}^{b}}
\end{align}
in our NOT procedure. Here for any vector $\vb=(v_{1},\ldots,v_{T})'$, the $i$-component of $\mathcal{S}_{s,e}(\vb)$ is given by $\mathcal{S}_{s,e}(\vb)_{i} = \sign\left(v_{i} - (e-s+1)^{-1}\sum_{t=s}^{e}v_{t}\right)$ and $\psib_{s,e}^{b}$ is defined by \eqref{Eq:basis_jump}. (For certain noise distributions, subtracting the sample median of $\vb$ instead of the sample mean would appear more appropriate.) The rationale behind \eqref{Eq:contrast_const_ht} is to assign $Y_s - \bar{\Yb}_{s,e},\ldots, Y_e - \bar{\Yb}_{s,e}$ (i.e.  residuals for fitting a curve with no change-point on a given interval) into two classes ($\pm1$, i.e. a two-point distribution, thus with light tails) and apply the contrast function to their $\pm1$ labels. Empirical performance of NOT (via Algorithm~\ref{Alg:not_solution_path}) combined with \eqref{Eq:contrast_const_ht} and sSIC is also illustrated in Section~\ref{Sec:simulation_study}. 

\section{Simulation study}
\label{Sec:simulation_study}
\subsection{Settings}
\label{Sec:sim_settings}
We consider examples following \ref{Scen:change_in_mean}--\ref{Scen:change_in_mean_and_variance} introduced in Section~\ref{Sec:contrast_function}, as well as an extra example satisfying
\begin{enumerate}[label=(S\arabic*)]
	\setcounter{enumi}{4}
	\item  \label{Scen:change_in_mean_slope_and_quad} $\sigma_{t}=\sigma_{0}$ and $f_{t}$ is a piecewise-quadratic function of $t$.
\end{enumerate}
Calculations required to derive the contrast function in \ref{Scen:change_in_mean_slope_and_quad} are similar to those shown in Section~\ref{Sec:contrast_function} for \ref{Scen:change_in_mean_and_slope}; we omit them here.

We simulate data according to Equation~\eqref{Eq:signal+noise} using the test signals \ref{Model:teeth} \texttt{teeth}, \ref{Model:blocks} \texttt{blocks}, \ref{Model:wave1} \texttt{wave1}, \ref{Model:wave2} \texttt{wave2}, \ref{Model:mix} \texttt{mix}, \ref{Model:vol} \texttt{vol} and \ref{Model:quad} \texttt{quad}, with the noise following 
\begin{enumerate}
	\item i.i.d. $\Nc(0,1)$;
	\item i.i.d. $\Nc(0,2)$;
	\item i.i.d. scaled Laplace distribution with zero-mean and unit-variance;
	\item i.i.d. scaled Student-$t_{5}$ distribution with unit-variance;
	\item a stationary Gaussian AR(1) process of $\varphi = 0.3$, with zero-mean and unit-variance.
\end{enumerate}
A detailed specification can be found in Section~\ref{Sec:simulation_models} of the online supplementary materials. Figure~\ref{Fig:model_examples} shows the examples of the data generated from models \ref{Model:teeth}--\ref{Model:quad}, as well as the estimates produced by NOT in a typical run. 

\begin{figure}[!ht]
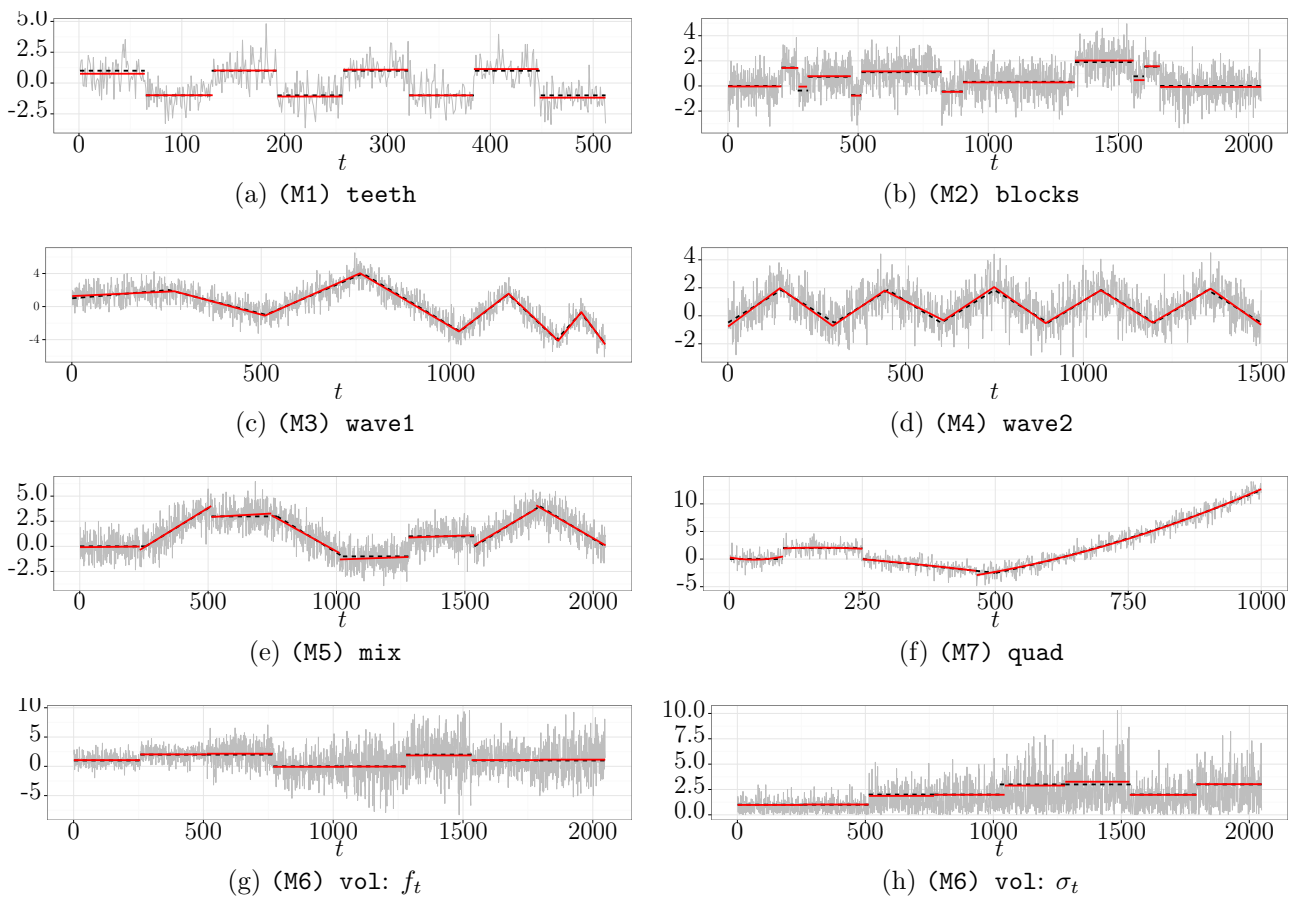

	\centering
	\subfloat[\texttt{(M1) teeth}]{
			\resizebox{0.48\textwidth}{!}{\input{tikz/teeth_example.tex}}
\label{Fig:teeth_example}}
	\subfloat[\texttt{(M2) blocks}]{
			\resizebox{0.48\textwidth}{!}{\input{tikz/blocks_example.tex}}
\label{Fig:blocks_example}}\\
	\subfloat[\texttt{(M3) wave1}]{
			\resizebox{0.48\textwidth}{!}{\input{tikz/wave1_example.tex}}
\label{Fig:wave1_example}}
	\subfloat[\texttt{(M4) wave2}]{
			\resizebox{0.48\textwidth}{!}{\input{tikz/wave2_example.tex}}
\label{Fig:wave2_example}}\\
	\subfloat[\texttt{(M5) mix}]{
			\resizebox{0.48\textwidth}{!}{\input{tikz/mix_example.tex}}
\label{Fig:mix_example}}
	\subfloat[\texttt{(M7) quad}]{
			\resizebox{0.48\textwidth}{!}{\input{tikz/quad_example.tex}}
\label{Fig:quad_example}}\\
	\subfloat[\texttt{(M6) vol}: $f_{t}$]{
			\resizebox{0.48\textwidth}{!}{\input{tikz/vol_example_mean.tex}}
\label{Fig:vol_example_mean}}
	\subfloat[\texttt{(M6) vol}: $\sigma_{t}$]{
			\resizebox{0.48\textwidth}{!}{\input{tikz/vol_example_volatility.tex}}
\label{Fig:vol_example_volatility}}\\
	\caption{Examples of data generated from simulation models studied in Section~\ref{Sec:simulation_models}.  Figure~\ref{Fig:teeth_example}-- \ref{Fig:vol_example_mean}: data series $Y_{t}$ (thin grey), true signal $f_{t}$ (dashed black), $\hat{f}_{t}$ being the least squares (LS) estimate of $f_{t}$ with the change-points estimated by NOT (thick red). Figure~\ref{Fig:vol_example_volatility}: centered data $|Y_{t}-\hat{f}_{t}|$ (thick grey), true standard deviation $\sigma_{t}$ (dashed black) and the estimated standard deviation $\hat{\sigma}_{t}$  between the change-points detected by NOT (thick red). \label{Fig:model_examples}}
\end{figure} 

\subsection{Estimators}
We apply Algorithm~\ref{Alg:not_solution_path} to compute the NOT solution path and pick the solution minimising the sSIC introduced in Section~\ref{Sec:sSIC} with $\alpha = 1$ (which is equivalent to SIC). In each simulated example, we use the contrast function designed to detect change-points in the scenario that the example follows, derived in Section~\ref{Sec:contrast_function} under the assumption that $\varepsilon_{t}$ is i.i.d. Gaussian. The resulting method is referred to simply as `NOT'. In addition, for Scenario \ref{Scen:change_in_mean} only, we also apply Algorithm~\ref{Alg:not_solution_path} combined with \eqref{Eq:contrast_const_ht} and SIC, which we call `NOT HT'. Here `HT' stands for `heavy tails'. The number of intervals drawn in the procedure and the maximum number of change-points for SIC are set to $M=10000$ and $q_{max}=25$, respectively.

We then compare the performance of NOT and NOT HT against the best competitors available on CRAN.
To the best of our knowledge, none of the competing packages can be applied in all of Scenarios \ref{Scen:change_in_mean}--\ref{Scen:change_in_mean_slope_and_quad}. 

For change-point detection in the mean, the selected competitors from CRAN are: \pkg{changepoint} \citep{killickeckley2014changepoint,  khe2016changepoint} implementing the PELT methodology proposed by \cite{killick2012optimal}, \pkg{changepoint.np} \citep{haynes2015changepointnp} implementing a nonparametric extension of the PELT methodology studied in \cite{haynes2016computationally}, \pkg{wbs} \citep{baranowski2015wbspackage} implementing the Wild Binary Segmentation proposed by \cite{fryzlewicz2014wild},
\pkg{ecp} \citep{jamesmatteson2014ecppackage} implementing the e.cp3o method proposed by \cite{james2015change}, \pkg{strucchange} \citep{zeileisetal2002strucchange} implementing the methodology of \cite{BaiPerron2003}, \pkg{Segmentor3IsBack} \citep{cleynenetal2013segmentor} implementing the technique proposed by \cite{rigaill2010pruned}, \pkg{nmcdr}  \citep{zoulancezhange2014nmcdr} implementing the NMCD methodology of \cite{zou2014nonparametric}, \pkg{stepR} \citep{hotzandsieling2016smuce} implementing the SMUCE method proposed by \cite{frick2014multiscale}, and \pkg{FDRSeg} \citep{LSA2017} implementing the FDRSeg method proposed by \citet{LMS2016}. We refer to the corresponding methods as, respectively, PELT, NP-PELT, WBS, e.cp3o, B{\&}P, S3IB, NMCD, SMUCE and FDRSeg. 

Note that e-cp3o, NMCD, NOT, PELT and NP-PELT can be also used for change-point detection in Scenario~\ref{Scen:change_in_mean_and_variance}, where change-points occur in the mean and variance of the data. In addition, for  Scenario~\ref{Scen:change_in_mean_and_variance}, we also include the SegNeigh method \citep{AugerLawrence1989} implemented in \pkg{changepoint}  \citep{killickeckley2014changepoint,  khe2016changepoint}. 

Only the B{\&}P method allows for change-point detection in piecewise-linear and piecewise-quadratic signals (in particular, the WBS is not suitable for these settings as described in Sections~\ref{Sec:introduction} and \ref{Sec:theory}), hence we also study the performance of the trend filtering methodology of \cite{kim2009ell_1} termed as TF hereafter, using the implementation available from the \proglang{R} package \pkg{genlasso} \citep{tibshirani2014genlasso}, to have a broader comparison. See also \cite{lin2016approximate}. The TF method aims to estimate a piecewise-polynomial signal from the data, not focusing on the change-point detection problem directly.  Let $\hat{f}_{t}^{(TF)}$ denote the TF estimate of the true signal $f_{t}$, then the TF estimates of the change-points in Scenario~\ref{Scen:change_in_slope} are defined as those $\tau$ for which $|2\hat{f}_{\tau}^{(TF)}-\hat{f}_{\tau-1}^{(TF)}-\hat{f}_{\tau+1}^{(TF)}|>\epsilon$, where $\epsilon>0$ is a very small number being the numerical tolerance level (more precisely, we set $\epsilon = 1.11\times 10^{-15}$ in our study). In the piecewise-quadratic case, the change-points are defined as those $\tau$ for which the third order differences $|\hat{f}_{\tau+2}^{(TF)}-3\hat{f}_{\tau+1}^{(TF)} + 3\hat{f}_{\tau}^{(TF)}-\hat{f}_{\tau-1}^{(TF)}|>\epsilon$. We note that both {B}\&{P} and TF require a substantial amount of computational resources, with {B}\&{P} being the slowest among all methods considered in this study. 

Finally, we remark that the tuning parameters for the competing methods are set to the values recommended by the corresponding \proglang{R} packages, and the \proglang{R} code for all simulations can be downloaded from our GitHub repository \citep{baranowski2016notcode}.

\subsection{Results}
\label{Sec:sim_res}
\begin{table}
\caption{Distribution of $\hat{q}-q$ for data generated according to \eqref{Eq:signal+noise} with the noise term $\varepsilon_{t}$ being i.i.d. $\mathcal{N}(0,1)$ for various choices of $f_{t}$ and $\sigma_{t}$ given in Section~\ref{Sec:simulation_models} of the online supplementary materials  and competing methods listed in Section~\ref{Sec:simulation_study}. Also, the average Mean-Square Error of the resulting estimate of the signal $f_{t}$, average Hausdorff distance $d_H$ given by \eqref{Eq:hausdorff_distance} and average computation time in seconds using a single core of an Intel Xeon 3.6 GHz CPU with 16 GB of RAM, all calculated over $100$ simulated data sets. Bold: methods with the largest empirical frequency of $\hat{q}-q=0$ or smallest average $d_{H}$ and those within $10\%$ of the highest, or, respectively, within $10\%$ of the lowest. \label{Table:sim_results_gaussian_sigma_1}} 

\centering
\footnotesize
\fbox{
	\begin{tabular}{*{11}{c|}c}
	&&\multicolumn{7}{c|}{$\hat{q}-q$}&&&\\
	Method & Model &$\leq-3$ & $-2$ & $-1$ & $0$ & $1$ & $2$ & $\geq 3$ & $\mbox{MSE}$ & $d_{H} \times 10^2$ & time\\
	\hline
	 B\&P & \multirow{11}{*}{\ref{Model:teeth}} & 70 & 8 & 1 & 21 & 0 & 0 & 0 & 0.703 & 11.39 & 0.27 \\ 
  e-cp3o &  & 0 & 0 & 0 & \textbf{100} & 0 & 0 & 0 & 0.052 & \textbf{0.48} & 2.32 \\ 
  FDRSeg &  & 0 & 0 & 0 & 78 & 16 & 4 & 2 & 0.085 & 1.39 & 0.16 \\ 
  NMCD &  & 0 & 0 & 0 & \textbf{96} & 4 & 0 & 0 & 0.093 & 0.76 & 1.38 \\ 
  NOT &  & 0 & 0 & 0 & \textbf{99} & 1 & 0 & 0 & 0.053 & 0.54 & 0.08 \\ 
  NOT HT &  & 0 & 0 & 0 & \textbf{99} & 1 & 0 & 0 & 0.055 & \textbf{0.51} & 0.1 \\ 
  NP-PELT &  & 0 & 0 & 0 & 86 & 11 & 2 & 1 & 0.068 & 0.85 & 0.03 \\ 
  PELT &  & 0 & 0 & 0 & \textbf{100} & 0 & 0 & 0 & 0.052 & \textbf{0.48} & 0 \\ 
  S3IB &  & 0 & 0 & 0 & \textbf{92} & 6 & 2 & 0 & 0.055 & 0.67 & 0.11 \\ 
  SMUCE &  & 0 & 0 & 0 & \textbf{100} & 0 & 0 & 0 & 0.083 & 0.57 & 0.22 \\ 
  WBS &  & 0 & 0 & 0 & \textbf{97} & 3 & 0 & 0 & 0.054 & 0.58 & 0.11 \\ 
  
	\hline
	 B\&P & \multirow{11}{*}{\ref{Model:blocks}} & 100 & 0 & 0 & 0 & 0 & 0 & 0 & 0.314 & 12.56 & 4.29 \\ 
  e-cp3o &  & 100 & 0 & 0 & 0 & 0 & 0 & 0 & 0.127 & 5.69 & 188.84 \\ 
  FDRSeg & & 0 & 1 & 33 & \textbf{52} & 10 & 3 & 1 & 0.03 & 1.82 & 2.43 \\ 
  NMCD &  & 0 & 5 & 64 & 31 & 0 & 0 & 0 & 0.035 & 1.82 & 4.92 \\ 
  NOT &  & 0 & 4 & 61 & 35 & 0 & 0 & 0 & 0.026 & 1.56 & 0.11 \\ 
  NOT HT &  & 2 & 8 & 54 & 28 & 8 & 0 & 0 & 0.033 & 2.08 & 0.23 \\ 
  NP-PELT &  & 0 & 0 & 27 & 44 & 15 & 9 & 5 & 0.029 & 2.13 & 0.49 \\ 
  PELT &  & 11 & 33 & 45 & 11 & 0 & 0 & 0 & 0.035 & 2.97 & 0.01 \\ 
  S3IB &  & 0 & 2 & 49 & \textbf{49} & 0 & 0 & 0 & 0.024 & \textbf{1.42} & 0.51 \\ 
  SMUCE &  & 59 & 36 & 5 & 0 & 0 & 0 & 0 & 0.069 & 3.44 & 0.03 \\ 
  WBS &  & 0 & 1 & 45 & \textbf{53} & 0 & 1 & 0 & 0.026 & \textbf{1.31} & 0.22 \\

	\hline
	 B\&P & \multirow{3}{*}{\ref{Model:wave1}} & 0 & 0 & 100 & 0 & 0 & 0 & 0 & 0.218 & 3.78 & 147.23 \\ 
  NOT &  & 0 & 0 & 0 & \textbf{99} & 1 & 0 & 0 & 0.015 & \textbf{0.99} & 0.63 \\ 
  TF &  & 0 & 0 & 0 & 0 & 0 & 0 & 100 & 0.019 & 8.33 & 63.98 \\ 
  
	\hline
	 B\&P & \multirow{3}{*}{\ref{Model:wave2}} & 0 & 1 & 3 & \textbf{96} & 0 & 0 & 0 & 0.072 & 2.59 & 168.12 \\ 
  NOT &  & 0 & 0 & 0 & \textbf{100} & 0 & 0 & 0 & 0.016 & \textbf{1.21} & 0.53 \\ 
  TF &  & 0 & 0 & 0 & 0 & 0 & 0 & 100 & 0.016 & 4.3 & 64.81 \\ 
  
	\hline
	 B\&P & \multirow{3}{*}{\ref{Model:mix}} & 0 & 0 & 0 & \textbf{100} & 0 & 0 & 0 & 0.02 & \textbf{2.42} & 382.96 \\ 
  NOT &  & 0 & 0 & 0 & \textbf{99} & 1 & 0 & 0 & 0.02 & \textbf{2.42} & 0.51 \\ 
  TF &  & 0 & 0 & 0 & 0 & 0 & 0 & 100 & 0.026 & 6.03 & 77.09 \\ 
  
	\hline
	 e-cp3o & \multirow{6}{*}{\ref{Model:vol}} & 94 & 3 & 0 & 3 & 0 & 0 & 0 & 0.378 & 16.83 & 11.35 \\ 
  NMCD &  & 0 & 0 & 7 & 83 & 8 & 2 & 0 & 0.057 & 2.54 & 4.8 \\ 
  NOT &  & 0 & 0 & 4 & \textbf{94} & 2 & 0 & 0 & 0.049 & \textbf{1.69} & 1.22 \\ 
  NP-PELT &  & 0 & 0 & 0 & 20 & 30 & 19 & 31 & 0.123 & 2.96 & 0.61 \\ 
  PELT &  & 9 & 15 & 28 & 48 & 0 & 0 & 0 & 0.074 & 8 & 0.02 \\ 
  SegNeigh &  & 0 & 0 & 8 & 60 & 17 & 10 & 5 & 0.054 & 2.5 & 38.02 \\ 
  
	\hline
	 B\&P & \multirow{3}{*}{\ref{Model:quad}} & 0 & 0 & 0 & \textbf{100} & 0 & 0 & 0 & 0.021 & \textbf{1.94} & 44.14 \\ 
  NOT &  & 0 & 0 & 0 & \textbf{100} & 0 & 0 & 0 & 0.02 & \textbf{1.78} & 0.31 \\ 
  TF &  & 0 & 0 & 0 & 0 & 0 & 0 & 100 & 0.049 & 23.33 & 59.56 \\ 
  
\end{tabular}}
\end{table}

\begin{table}
\caption{Distribution of $\hat{q}-q$ for data generated according to \eqref{Eq:signal+noise} with the noise term $\varepsilon_{t}$ being i.i.d. $(3/5)^{1/2}t_{5}$ for various choices of $f_{t}$ and $\sigma_{t}$ given in Section~\ref{Sec:simulation_models} of the online supplementary materials and competing methods listed in Section~\ref{Sec:simulation_study}. Also, the average Mean-Square Error of the resulting estimate of the signal $f_{t}$, average Hausdorff distance $d_H$ given by \eqref{Eq:hausdorff_distance} and average computation time in seconds using a single core of an Intel Xeon 3.6 GHz CPU with 16 GB of RAM, all calculated over $100$ simulated data sets. Bold: methods with the largest empirical frequency of $\hat{q}-q=0$ or smallest average $d_{H}$ and those within $10\%$ of the highest, or, respectively, within $10\%$ of the lowest. \label{Table:sim_results_student_sigma_1}} 	
	
\centering
\footnotesize
\fbox{
\begin{tabular}{*{11}{c|}c}
	&&\multicolumn{7}{c|}{$\hat{q}-q$}&&&\\
	Method & Model &$\leq-3$ & $-2$ & $-1$ & $0$ & $1$ & $2$ & $\geq 3$ & $\mbox{MSE}$ & $d_{H} \times 10^2$ & time\\
	\hline
	 B\&P & \multirow{11}{*}{\ref{Model:teeth}} & 65 & 12 & 0 & 23 & 0 & 0 & 0 & 0.67 & 10.76 & 0.26 \\ 
  e-cp3o &  & 0 & 0 & 0 & \textbf{100} & 0 & 0 & 0 & 0.044 & \textbf{0.39} & 2.22 \\ 
  NMCD &  & 0 & 0 & 0 & \textbf{94} & 6 & 0 & 0 & 0.092 & 0.81 & 1.31 \\ 
  FDRSeg &  & 0 & 0 & 0 & 6 & 7 & 10 & 77 & 0.11 & 4.47 & 0.05 \\    
  NOT &  & 0 & 0 & 0 & \textbf{94} & 5 & 1 & 0 & 0.046 & 0.57 & 0.08 \\ 
  NOT HT &  & 0 & 0 & 0 & \textbf{98} & 2 & 0 & 0 & 0.045 & 0.47 & 0.1 \\ 
  NP-PELT &  & 0 & 0 & 0 & 73 & 14 & 11 & 2 & 0.082 & 1.37 & 0.03 \\ 
  PELT &  & 0 & 0 & 0 & 63 & 6 & 16 & 15 & 0.092 & 1.68 & 0 \\ 
  S3IB &  & 0 & 0 & 0 & 54 & 7 & 20 & 19 & 0.096 & 1.84 & 0.11 \\ 
  SMUCE &  & 0 & 0 & 0 & 45 & 22 & 19 & 14 & 0.091 & 2.53 & 0.21 \\ 
  WBS &  & 0 & 0 & 0 & 44 & 3 & 28 & 25 & 0.105 & 2.44 & 0.11 \\ 
	\hline
	 B\&P & \multirow{11}{*}{\ref{Model:blocks}} & 100 & 0 & 0 & 0 & 0 & 0 & 0 & 0.302 & 11.98 & 4.28 \\ 
  e-cp3o &  & 100 & 0 & 0 & 0 & 0 & 0 & 0 & 0.126 & 5.87 & 197.26 \\ 
  FDRSeg &  & 0 & 0 & 0 & 0 & 0 & 1 & 99 & 0.044 & 6.98 & 1.44 \\   
  NMCD &  & 0 & 4 & 66 & \textbf{29} & 0 & 1 & 0 & 0.032 & \textbf{1.92} & 5.13 \\ 
  NOT &  & 2 & 16 & 33 & \textbf{31} & 14 & 3 & 1 & 0.032 & 4.09 & 0.11 \\ 
  NOT HT &  & 1 & 7 & 62 & \textbf{28} & 2 & 0 & 0 & 0.027 & \textbf{1.9} & 0.23 \\ 
  NP-PELT &  & 0 & 0 & 6 & 22 & 20 & 23 & 29 & 0.048 & 3.91 & 0.46 \\ 
  PELT &  & 0 & 3 & 16 & 19 & 20 & 12 & 30 & 0.066 & 3.98 & 0.01 \\ 
  S3IB &  & 29 & 10 & 26 & 20 & 4 & 11 & 0 & 0.065 & 4.38 & 0.49 \\ 
  SMUCE &  & 0 & 5 & 11 & 25 & 14 & 13 & 32 & 0.056 & 5.36 & 0.03 \\ 
  WBS &  & 0 & 3 & 15 & 11 & 21 & 15 & 35 & 0.067 & 4.7 & 0.22 \\ 
 
	\hline
	 B\&P & \multirow{3}{*}{\ref{Model:wave1}} & 0 & 0 & 100 & 0 & 0 & 0 & 0 & 0.217 & 3.63 & 149.51 \\ 
  NOT &  & 0 & 0 & 0 & \textbf{99} & 1 & 0 & 0 & 0.015 & \textbf{1} & 0.63 \\ 
  TF &  & 0 & 0 & 0 & 0 & 0 & 0 & 100 & 0.017 & 8.4 & 66.66 \\ 
  
	\hline
	 B\&P & \multirow{3}{*}{\ref{Model:wave2}} & 0 & 0 & 10 & \textbf{90} & 0 & 0 & 0 & 0.081 & 2.78 & 175.34 \\ 
  NOT &  & 0 & 0 & 0 & \textbf{94} & 5 & 1 & 0 & 0.019 & \textbf{1.51} & 0.54 \\ 
  TF &  & 0 & 0 & 0 & 0 & 0 & 0 & 100 & 0.017 & 4.44 & 68.33 \\ 
  
	\hline
	 B\&P & \multirow{3}{*}{\ref{Model:mix}} & 0 & 0 & 0 & \textbf{100} & 0 & 0 & 0 & 0.019 & \textbf{2.29} & 392 \\ 
  NOT &  & 0 & 0 & 0 & \textbf{96} & 4 & 0 & 0 & 0.019 & \textbf{2.33} & 0.53 \\ 
  TF &  & 0 & 0 & 0 & 0 & 0 & 0 & 100 & 0.026 & 6.01 & 80.41 \\ 
  
	\hline
	  e-cp3o & \multirow{6}{*}{\ref{Model:vol}} & 91 & 2 & 2 & 4 & 0 & 1 & 0 & 0.327 & 14.05 & 11.51 \\ 
  NMCD &  & 0 & 12 & 47 & \textbf{36} & 5 & 0 & 0 & 0.053 & 8.56 & 4.94 \\ 
  NOT &  & 0 & 4 & 17 & \textbf{35} & 25 & 12 & 7 & 0.08 & 6.1 & 1.26 \\ 
  NP-PELT &  & 0 & 0 & 2 & 9 & 22 & 19 & 48 & 0.205 & \textbf{5.1} & 0.66 \\ 
  PELT &  & 7 & 14 & 26 & \textbf{33} & 15 & 5 & 0 & 0.112 & 8.88 & 0.03 \\
  SegNeigh &  & 2 & 1 & 4 & 25 & 17 & 24 & 27 & 0.128 & \textbf{4.86} & 31.34 \\    
	\hline
	 B\&P & \multirow{3}{*}{\ref{Model:quad}} & 0 & 0 & 0 & \textbf{99} & 1 & 0 & 0 & 0.021 & \textbf{2.5} & 45.59 \\ 
  NOT &  & 0 & 0 & 8 & 79 & 11 & 2 & 0 & 0.03 & 4.28 & 0.32 \\ 
  TF &  & 0 & 0 & 0 & 0 & 0 & 0 & 100 & 0.05 & 23.32 & 62.79 \\ 
  	
\end{tabular}}

\end{table}

Here we only present the results under the setting where the noise is (a) i.i.d. standard normal in Table~\ref{Table:sim_results_gaussian_sigma_1}, and (d) i.i.d. scaled Student-$t_{5}$ in Table~\ref{Table:sim_results_student_sigma_1}. Additional results under the other above-mentioned noise settings can be found in Section~\ref{Sec:sim_additional} of the online supplementary materials. 

For each method, we show a frequency table for the distribution of $\hat{q}-q$, where $\hat{q}$ is the number of the estimated change-points and $q$ denotes the true number of change-points. We also report Monte-Carlo estimates of the Mean Squared Error of the estimated signal, given by $\mbox{MSE} = \mathbb{E} \Big\{\frac{1}{T}\sum_{t=1}^{T}(f_{t}-\hat{f}_{t})^2\Big\}$. For all methods but TF, $\hat{f}_{t}$ is calculated by finding the least squares (LS) approximation of the signal of the appropriate type depending on the true $f_{t}$, between each consecutive pair of estimated change­-points. For TF, $\hat{f}_{t}$ used in the definition of the MSE is the penalised least squares estimate of $f_{t}$ returned by the TF algorithm.  

To assess the performance of each method in terms of the accuracy of the estimated locations of the change-points, we also report estimates of the (scaled) Hausdorff distance defined as 
\begin{align}
	\label{Eq:hausdorff_distance} 
	d_{H} = T^{-1} \E  \max\left\{\max_{j=0,\ldots,q+1} \min_{k=0,\ldots,\hat{q}+1}|\tau_{j}-\hat{\tau}_{k}|, \max_{k=0,\ldots,\hat{q}+1} \min_{j=0,\ldots,q+1}|\hat{\tau}_{k}-\tau_{j}|\right\},
\end{align} 
where $0=\tau_{0} < \tau_{1} < \ldots \tau_{q} < \tau_{q+1}=T$ and $0=\hat{\tau}_{0} < \hat{\tau}_{1} < \ldots \hat{\tau}_{q} < \hat{\tau}_{q+1}=T$ denote, respectively, true and estimated locations of the change-points. From the definition above, it follows  that $0 \leq d_{H} \leq 1$. 
An estimator is regarded to perform well when its $d_{H}$ is close to $0$. However, $d_{H}$ would be large when the number of change-points is under-estimated or some of the estimated change-points are far away from the real ones. 

We find that in most of the simulated scenarios,  NOT is among the most competitive methods in terms of the estimation of the number of change-points, their locations, as well as the true signal. Importantly, it is very fast to compute, which gives it a particular advantage over its competitors in Scenarios \ref{Scen:change_in_slope}, \ref{Scen:change_in_mean_and_slope} and \ref{Scen:change_in_mean_slope_and_quad}. 
Finally, NOT with the contrast function derived under the assumption that the noise is i.i.d. Gaussian is relatively robust against the misspecification in $\varepsilon_{t}$, when the truth is either correlated or heavy-tailed.

\subsection{More on model misspecification and model selection}
\label{Sec:simulation_miss}

We have demonstrated that NOT is relatively robust against the misspecification in the distribution of $\varepsilon_{t}$, when the truth is either correlated or heavy-tailed. Now we investigate the case where the signal $f_t$ is misspecified. In particular, we focus on the misspecification of the degree of the polynomials between consecutive change-points. 

We simulate data according to \eqref{Eq:signal+noise} using the signal \ref{Model:smile} \texttt{smile} and noise of (a) i.i.d. $\Nc(0,1)$ and (b) i.i.d. $\Nc(0,2)$.  Here the true signal is piecewise-linear but not necessarily continuous (i.e. from Scenario \ref{Scen:change_in_mean_and_slope}). We test NOT with sSIC using contrast functions constructed from Scenarios \ref{Scen:change_in_mean}, \ref{Scen:change_in_mean_and_slope} and  \ref{Scen:change_in_mean_slope_and_quad}, where the estimators are denoted by $\mathrm{NOT}_0$, $\mathrm{NOT}_1$ and $\mathrm{NOT}_2$, respectively. Again we take $\alpha = 1$. Figure~\ref{Fig:model_miss} shows a typical realisation of the estimates produced by NOT with different contrast functions, while Table~\ref{Table:sim_results_miss} summarises the results.

For $\mathrm{NOT}_0$ (suitable for piecewise-constant signal), we see that unsurprisingly  $\mathrm{NOT}_0$ significantly overestimates the number of change-points $q$. This is due to the bias-variance tradeoff in the sSIC, where the bias term only approaches zero as the estimated number of change-points $\hat{q} \rightarrow \infty$. Nevertheless, we note that the set of change-point estimates from $\mathrm{NOT}_0$ typically includes the true change-points with jump, even though the construction of the contrast function (wrongly) assumes that the signal is piecewise-constant in the neighbourhood of these change-points. Furthermore, under the higher signal-to-noise ratio setting, $\mathrm{NOT}_2$, which is designed for piecewise-quadratic signal, is able to estimate the number of change-points $q$ correctly most of the time. However, since $\mathrm{NOT}_2$ is over-parameterised in this setting of Scenario \ref{Scen:change_in_mean_and_slope}, it tends to perform slightly worse than $\mathrm{NOT}_1$ in terms of both the MSE for the estimated signal, and the accuracy of the estimated locations of the change-points. Finally, under the lower signal-to-noise ratio setting,
$\mathrm{NOT}_2$ tends to underestimate the number of change-points, thanks to the bias-variance tradeoff in the sSIC. Nevertheless, as is illustrated in Figures~\ref{Fig:miss_not_2_2}, the estimated $f_t$ is quite close to the truth in terms of the $\ell_2$ distance. These findings suggest that NOT could still provide valuable insights in certain misspecified circumstances. 

In the same example, we also demonstrate that one could empirically select the degree of the polynomial for the NOT's contrast function via sSIC.  
Denote the sSIC scores corresponding to the estimates from $\mathrm{NOT}_0$, $\mathrm{NOT}_1$ and $\mathrm{NOT}_2$ by sSIC($\mathrm{NOT}_0$),  sSIC($\mathrm{NOT}_1$) and  sSIC($\mathrm{NOT}_2$) respectively. We propose to pick the estimator produced by  $\mathrm{NOT}_{i^*}$ with 
\[
i^* = \mathrm{argmin}_{i \in \{0,1,2\}} \mathrm{sSIC}(\mathrm{NOT}_i).
\]
As shown in Table~\ref{Table:sim_results_miss}, empirical results suggest that we are able to select the correct order of the polynomial for our NOT approach using sSIC, especially when the signal-to-noise ratio is high.

\begin{figure}[!ht]
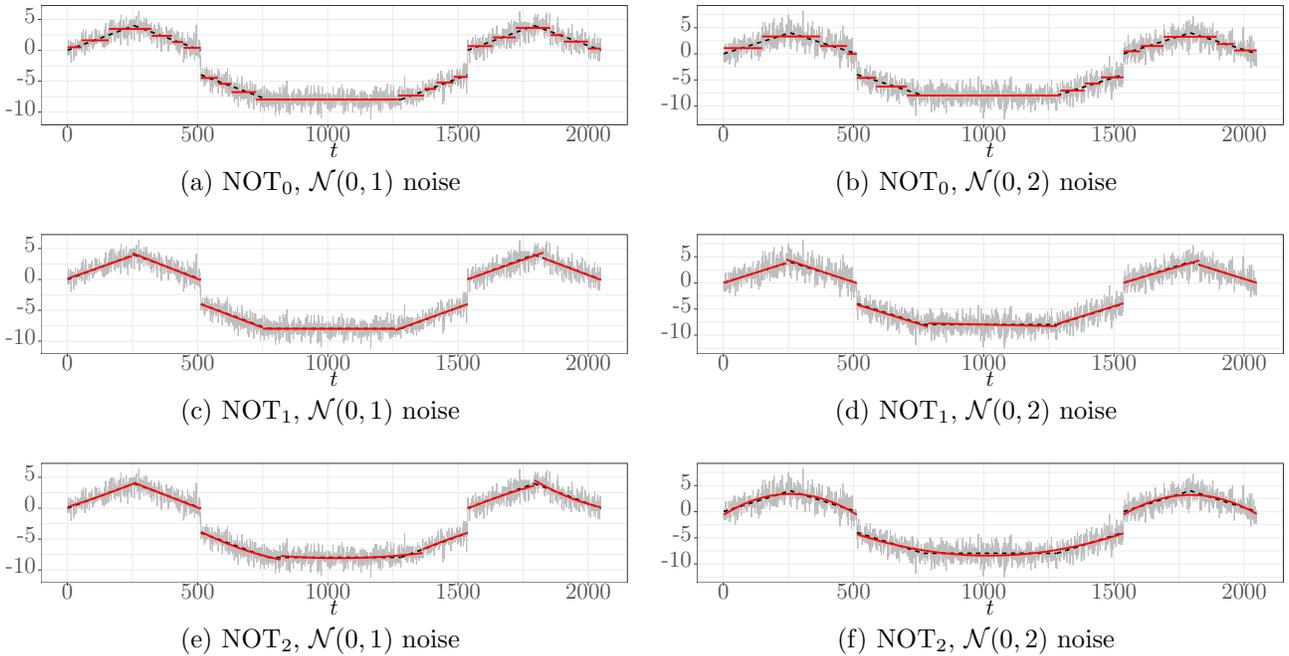

	\centering
	\subfloat[$\mathrm{NOT}_0$, $\Nc(0,1)$  noise]{
			\resizebox{0.48\textwidth}{!}{\input{tikz/miss_pcwsConstMean_1.tex}}
\label{Fig:miss_not_0_1}}
	\subfloat[$\mathrm{NOT}_0$, $\Nc(0,2)$ noise]{
			\resizebox{0.48\textwidth}{!}{\input{tikz/miss_pcwsConstMean_2.tex}}
\label{Fig:miss_not_0_2}}	
	\\
	\subfloat[$\mathrm{NOT}_1$, $\Nc(0,1)$ noise]{
			\resizebox{0.48\textwidth}{!}{\input{tikz/miss_pcwsLinMean_1.tex}}
\label{Fig:miss_not_1_1}}
	\subfloat[$\mathrm{NOT}_1$, $\Nc(0,2)$ noise]{
			\resizebox{0.48\textwidth}{!}{\input{tikz/miss_pcwsLinMean_2.tex}}
\label{Fig:miss_not_1_2}}
	\\
	\subfloat[$\mathrm{NOT}_2$, $\Nc(0,1)$ noise]{
			\resizebox{0.48\textwidth}{!}{\input{tikz/miss_pcwsQuadMean_1.tex}}
\label{Fig:miss_not_2_1}}
	\subfloat[$\mathrm{NOT}_2$, $\Nc(0,2)$ noise]{
			\resizebox{0.48\textwidth}{!}{\input{tikz/miss_pcwsQuadMean_2.tex}}
\label{Fig:miss_not_2_2}}	
	\\
	\caption{Typical realisation of the estimates produced by different NOTs, with data generated from \ref{Model:smile} \texttt{smile}. Figure~\ref{Fig:miss_not_0_1}-- \ref{Fig:miss_not_2_2}: data series $Y_{t}$ (thin grey), true signal $f_{t}$ (dashed black), $\hat{f}_{t}$ being the LS estimate of $f_{t}$ with the change-points estimated by NOT (thick red). Higher signal-to-noise ratio setting (with $\Nc(0,1)$ errors) in Figures~\ref{Fig:miss_not_0_1}, \ref{Fig:miss_not_1_1} and \ref{Fig:miss_not_2_1}; lower signal-to-noise ratio setting (with $\Nc(0,2)$ errors) in Figures~\ref{Fig:miss_not_0_2}, \ref{Fig:miss_not_1_2} and \ref{Fig:miss_not_2_2}. Here $\mathrm{NOT}_0$, $\mathrm{NOT}_1$ and $\mathrm{NOT}_2$ denote estimates from NOT with sSIC using contrast functions constructed from Scenarios \ref{Scen:change_in_mean}, \ref{Scen:change_in_mean_and_slope} and  \ref{Scen:change_in_mean_slope_and_quad}, respectively. \label{Fig:model_miss}}
\end{figure}

\begin{table}
\caption{Distribution of $\hat{q}-q$ obtained by $\mathrm{NOT}_0, \mathrm{NOT}_1, \mathrm{NOT}_2$ for data generated according to \eqref{Eq:signal+noise} with the signal \ref{Model:smile} and the noise $\varepsilon_{t} \stackrel{\mathrm{i.i.d.}}{\sim} \Nc(0,1)$ and $\Nc(0,2)$, the average Mean-Square Error of the resulting estimate of the signal over 100 simulations. The number of times each method selected by sSIC is also reported. 
\label{Table:sim_results_miss}} 	
	
\centering
\fbox{
\begin{tabular}{*{10}{c|}c}
	&&\multicolumn{7}{c|}{$\hat{q}-q$}&& Number of times\\
	Noise & Method & $\leq-3$ & $-2$ & $-1$ & $0$ & $1$ & $2$ & $\geq 3$ & $\mbox{MSE}$ & selected by sSIC \\
	\hline 
	& $\mathrm{NOT}_0$ & 0 & 0 & 0 & 0 & 0 & 0 & 100 & 0.120  & 0\\ 
	$\Nc(0,1)$ & $\mathrm{NOT}_1$   & 0 & 0 & 0 & \textbf{99} & 1 & 0 & 0 & \textbf{0.015}  & \textbf{100} \\ 
	& $\mathrm{NOT}_2$   & 0 & 4 & 18 & 78 & 0 & 0 & 0 & 0.024  & 0\\ \hline 
	& $\mathrm{NOT}_0$ & 0 & 0 & 0 & 0 & 0 & 0 & 100 & 0.188 & 0 \\ 
	$\Nc(0,2)$ & $\mathrm{NOT}_1$    & 0 & 0 & 0 & \textbf{100} & 0 & 0 & 0 & \textbf{0.032} & \textbf{94} \\ 
	& $\mathrm{NOT}_2$   & 57 & 23 & 14 & 6 & 0 & 0 & 0 & 0.078  & 6\\
\end{tabular}}

\end{table}

\section{Real data analysis}
\label{Sec:data}

\subsection{Temperature anomalies}
\label{Sec:temp_anomalies_analysis}

We analyse the GISS Surface Temperature anomalies data set available from \cite{gistemp2016anomdata}, consisting of monthly temperature anomalies recorded from January 1880 to June 2016. The anomaly here is defined as the difference between the average global temperature in a given month and the baseline value, being the average calculated for that time of the year over the 30-year period from 1951 to 1980; for more details see \cite{hansen2010global}. This and similar anomalies series are frequently studied in the literature with a particular focus on identifying change-points in the data, see e.g. \cite{ruggieri2013bayesian} or \cite{james2015change}. 

\begin{figure}[!ht]
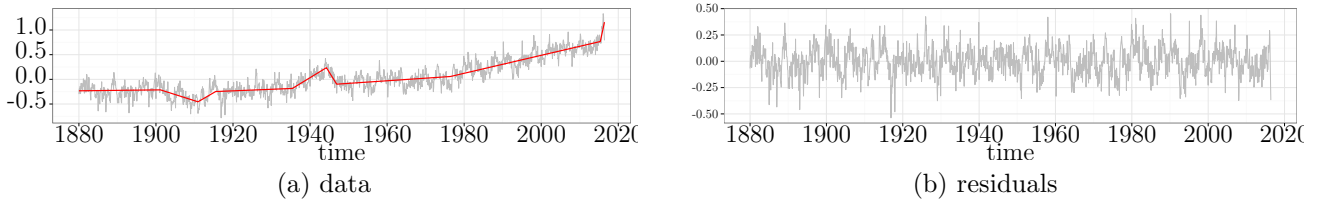

	\centering
	\subfloat[][data]{
			\resizebox{0.48\textwidth}{!}{\input{tikz/temp_anomalies_fitted.tex}}
\label{Fig:temp_anomalies_analysis_fitted}} \subfloat[][residuals]{	
			\resizebox{0.48\textwidth}{!}{\input{tikz/temp_anomalies_residuals.tex}}
\label{Fig:temp_anomalies_analysis_residuals}}\\
	\caption{Change-point analysis for the GISSTEMP data set introduced in Section~\ref{Sec:temp_anomalies_analysis}. Figure~\ref{Fig:temp_anomalies_analysis_fitted}: the data series $Y_{t}$ (thin grey) and $\hat{f}_{t}$ estimated using change-points returned by NOT (thick red). Figure~\ref{Fig:temp_anomalies_analysis_residuals}: residuals $\hat{\varepsilon}_{t}=Y_{t}-\hat{f}_{t}$. \label{Fig:temp_anomalies_analysis}}
\end{figure}

The plot of the data (Figure~\ref{Fig:temp_anomalies_analysis_fitted})  indicates the presence of a linear trend with several change-points in the temperature anomalies series. The corresponding changes are not abrupt, therefore we believe that Scenario \ref{Scen:change_in_slope} with change-points in the slope of the trend is the most appropriate here. To detect the locations of the change-points, we apply NOT (via Algorithm~\ref{Alg:not_solution_path}) with the contrast given by \eqref{Eq:contrast_kink}, combined with the SIC to determine the best model on the solution path. 

The NOT estimate of the piecewise-linear trend and the corresponding  empirical residuals are shown in Figure~\ref{Fig:temp_anomalies_analysis}. We identify 8 change-points located at the following dates: March 1901, December 1910, July 1915,  June 1935, April 1944,  December 1946,  June 1976 and May 2015.  Previous studies conducted on similar temperature anomalies series (observed at a yearly frequency and obtained from a different source), report change-points around 1910, 1945 and 1976 (see \cite{ruggieri2013bayesian} for an overview of a number of related analyses). In addition to the change-points around these dates, NOT identifies two periods, 1901--1915 and 1935--1946, where local deviations from the baseline. 
We also observe a long-lasting upward trend in the anomalies series starting in December 1946. Finally, NOT indicates that the slope of the trend is increasing, with the most recent change-point in May 2015. 

\subsection{UK House Price Index}
\label{Sec:house_price_index}

We analyse monthly percentage changes in the UK House Price Index (HPI), which provides an overall estimate of the changes in house prices across the UK. The data and a detailed description of how the index is calculated are available online from \cite{hpi2016data}. \cite{fryzlewicz2016tail}, who proposes a method for signal estimation and change-point detection in Scenario~\ref{Scen:change_in_mean}, used this data set to illustrate the performance of his methodology. We perform a similar analysis, assuming the more flexible Scenario~\ref{Scen:change_in_mean_and_variance}, allowing for changes both in the mean and the variance, which, we argue, leads to additional insights and better-interpretable estimates for this dataset.

\begin{figure}[!ht]
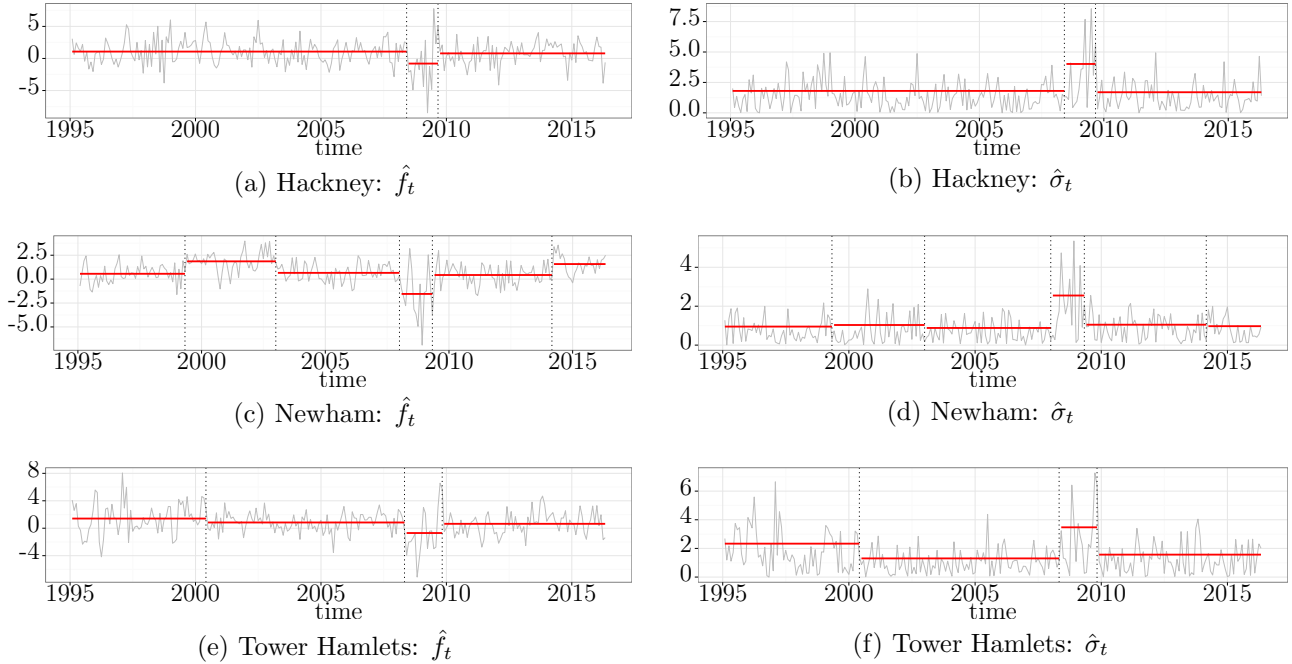

	\centering
	\subfloat[][Hackney: $\hat{f}_{t}$]{
			\resizebox{0.48\textwidth}{!}{\input{tikz/hpi_Hackney_mean.tex}}
\label{Fig:hpi_Hackney_mean}}
	\subfloat[][Hackney: $\hat{\sigma}_{t}$]{
			\resizebox{0.48\textwidth}{!}{\input{tikz/hpi_Hackney_vol.tex}}
\label{Fig:hpi_Hackney_vol}}\\
	\subfloat[][Newham: $\hat{f}_{t}$]{
			\resizebox{0.48\textwidth}{!}{\input{tikz/hpi_Newham_mean.tex}}
\label{Fig:hpi_Newham_mean}}
	\subfloat[][Newham: $\hat{\sigma}_{t}$]{
			\resizebox{0.48\textwidth}{!}{\input{tikz/hpi_Newham_vol.tex}}
\label{Fig:hpi_Newham_vol}}\\
	\subfloat[][Tower Hamlets: $\hat{f}_{t}$]{
			\resizebox{0.48\textwidth}{!}{\input{tikz/hpi_Tower_Hamlets_mean.tex}}
\label{Fig:hpi_Tower_Hamlets_mean}}
	\subfloat[][Tower Hamlets:  $\hat{\sigma}_{t}$]{
			\resizebox{0.48\textwidth}{!}{\input{tikz/hpi_Tower_Hamlets_vol.tex}}
\label{Fig:hpi_Tower_Hamlets_vol}}\\
	\caption{\label{Fig:hpi_analysis} Change-point analysis for the monthly percentage changes in the UK House Price Index from January 1995 to May 2016. Figure~\ref{Fig:hpi_Hackney_mean}, \ref{Fig:hpi_Newham_mean} and \ref{Fig:hpi_Tower_Hamlets_mean}: the monthly percentage changes $Y_{t}$ and the fitted piecewise-constant mean $\hat{f}_{t}$, between the change-points estimated with NOT. Figure~\ref{Fig:hpi_Hackney_vol}, \ref{Fig:hpi_Newham_vol} and \ref{Fig:hpi_Tower_Hamlets_vol}: $|Y_{t}-\hat{f}_{t}|$ and the fitted piecewise-constant standard deviation $\hat{\sigma}_{t}$, between the change-points estimated with NOT. }
\end{figure} 

As in \cite{fryzlewicz2016tail},  we analyse the percentage changes in the HPI for three London boroughs, namely Hackney, Newham and Tower Hamlets, all of which are located in East London. Hackney and Tower of Hamlets border on the City of London, a major business and financial district, with the latter being home to Canary Wharf, another important financial centre. On the other hand, Newham, located to the east of Hackney and Tower Hamlets, hosted the London 2012 Olympic Games which involved large-scale investment in that borough.

Figure~\ref{Fig:hpi_analysis} shows monthly percentage changes in HPI for the analysed boroughs and the corresponding NOT estimates, obtained using the contrast function \eqref{Eq:contrast_volatility}. As recommended in Section~\ref{Sec:sSIC} and \ref{Sec:parameter_choice_M}, we set the number of intervals drawn in the procedure to $M=10000$ and choose the threshold that minimises the SIC. For better comparability, NOT is applied with the same random seed for each data series.

In contrast to \cite{fryzlewicz2016tail}, whose TGUH method estimates at least 10 change-points in each HPI series, we detect just a few change-points in the data, facilitating the interpretation of the results. Furthermore, for all three boroughs, NOT estimates two change-points (one around March 2008 and one around September 2009) that could perhaps be linked to the 2008--2009 financial crisis and the concurrent collapse of the housing market. Estimated standard deviations for that period are much larger than the estimates corresponding to the other segments of piecewise-constancy, suggesting that the market is more volatile during 2008--2009, and thus in this example Scenario \ref{Scen:change_in_mean_and_variance} may be more relevant than \ref{Scen:change_in_mean} considered in \cite{fryzlewicz2016tail}. It is also interesting to observe that, with the exception of Tower Hamlets from January 1995 to April 2000 and the 2008--2009 financial crisis
for all boroughs, the estimated standard deviations appear to oscillate around a baseline level. 

The period of a larger volatility for Tower Hamlets in Figure~\ref{Fig:hpi_Tower_Hamlets_vol}, observed from January 1995 to April 2000, somewhat coincides with the developments in Canary Wharf,  which in the past was a dock complex closed in 1980. \cite{gordon2001resurrection} claims that the project of converting Canary Wharf into a business district ``was politically controversial and widely regarded as a planning disaster'' which  ``(in 1992) failed as a result of six factors: a recession in the London property market, competition from the City of London, poor transport links, few British tenants, complicated finances and developer overconfidence''. Over the 1995--2000 period, the situation in the London property reversed, which combined with a development of new public transport lines in Canary Wharf led to the success of the project. Indeed, according to \cite{gordon2001resurrection}, ``when the Jubilee underground line opened in 2000, Canary Wharf's resurrection was complete''.

\section*{Acknowledgements}
\label{Sec:acknowledgements}
We thank Paul Fearnhead for his helpful comments on an earlier draft, and on the implementation of our  \proglang{R} package. We also thank the associate editor and four anonymous referees for their comments and  suggestions. Piotr Fryzlewicz's work was supported by the Engineering and Physical Sciences Research Council grant No. EP/L014246/1. 

\newpage
\appendix
\section*{{Online supplementary materials for \\ `Narrowest-Over-Threshold Detection of Multiple Change-points and Change-point-like Features'}}

\vspace{0.5cm}
\noindent This document contains the following parts:
\begin{enumerate}
	\item[\textbf{A.}] Simulation models
	\item[\textbf{B.}] More details on the computational aspects of NOT and its solution path
	\item[\textbf{C.}] Additional simulation results
	\item[\textbf{D.}] Additional real data example: oil pirce
	\item[\textbf{E.}] Proofs
\end{enumerate}


\section{Simulation models}
\label{Sec:simulation_models}
\begin{enumerate}[label=(M\arabic*)]
	\item\label{Model:teeth} \texttt{teeth}: piecewise-constant $f_{t}$ (in Scenario~\ref{Scen:change_in_mean}), $T=512$, $q=7$ change-points at $\tau=64, 128, \ldots, 448$, with the corresponding jump sizes $-2,2,-2,\ldots,-2$, starting intercept $f_{1}=1$, $\sigma_{t}=1$ for $t=1,\ldots,T$.
	\item \label{Model:blocks} \texttt{blocks}: piecewise-constant $f_{t}$ (in Scenario~\ref{Scen:change_in_mean}), $T=2024$, $q=11$ change-points at $\tau=205, 267, 308, 472, 512, 820, 902, 1332, 1557, 1598, 1659$, with the corresponding jump sizes $1.464, -1.830,  1.098, -1.464,  1.830, -1.537,  0.768,  1.574, -1.135,  0.769, -1.537$, starting intercept $f_{1}=0$, $\sigma_{t}=1$ for $t=1,\ldots,T$. This signal is widely analysed in the literature, see e.g. \cite{fryzlewicz2014wild}. 
	\item\label{Model:wave1} \texttt{wave1}: piecewise-linear $f_{t}$ without jumps in the intercept (in Scenario~\ref{Scen:change_in_slope}), $T=1408$, $q=7$ change-points at $\tau=256, 512, 768, 1024, 1152, 1280, 1344$, with the corresponding changes in slopes $1 \cdot 2^{-6}, -2 \cdot 2^{-6}, 3 \cdot 2^{-6} \ldots, -7 \cdot 2^{-6}$, starting intercept $f_{1}=1$ and slope $f_{2}-f_{1}= 2^{-8}$, $\sigma_{t}=1$ for $t=1,\ldots,T$.
	\item\label{Model:wave2} \texttt{wave2}: piecewise-linear $f_{t}$ without jumps in the intercept (in Scenario~\ref{Scen:change_in_slope}), $T=1500$, $q=9$ change-points at $\tau=150, 300,\ldots, 1350$, with the corresponding changes in slopes $2^{-5}, -2^{-5}, 2^{-5}, \ldots,  -2^{-5}$, starting intercept $f_{1}=2^{-1}$ and slope $f_{2}-f_{1}= 2^{-6}$, $\sigma_{t}=1$ for $t=1,\ldots,T$.
	\item\label{Model:mix} \texttt{mix}: piecewise-linear $f_{t}$ with possible jumps at change-points (in Scenario~\ref{Scen:change_in_mean_and_slope}), length $T=2048$, $q=7$ change-points at $\tau=256, 512, \ldots, 1792$, with the corresponding sizes of jump $0, -1, 0, 0, 2, -1, 0$ and changes in the slope $2^{-6}, -2^{-6}, -2^{-6}, 2^{-6}, 0, 2^{-6}, -2^{-5}$, starting value for the intercept $f_{1}=0$ and slope $f_{2}-f_{1}=0$, $\sigma_{t}=1$ for $t=1,\ldots,T$.
	\item\label{Model:vol} \texttt{vol}: piecewise-constant $f_{t}$ and $\sigma_{t}$ (in Scenario~\ref{Scen:change_in_mean_and_variance}), $T=2048$, $q=7$ change-points at $\tau=256, 512, \ldots, 1792$ with the corresponding jumps in $f_t$ and $\sigma_t$ being $1,0,-2,0,2,-1,0$ and $0,1,0,1,0,-1,1$, respectively, initial values $f_{1}=\sigma_{1}=1$.
	\item\label{Model:quad} \texttt{quad}: piecewise-quadratic $f_{t}$ (in Scenario~\ref{Scen:change_in_mean_slope_and_quad}), $T=1000$, $q=3$ change-points at $\tau=100, 250, 500$, with the corresponding changes in the intercept $2, -2, 0$, in the slope $0, -10^{-1}, 10^{-1}$ and in the quadratic coefficient $0, 0,2  \times 10^{-5}$, the initial values $f_{1}=f_{2}-f_{1}=f_{3}-2f_{2}+f_{1}=0$, $\sigma_{t}=1$ for all $t=1,\ldots,T$.
	\item\label{Model:smile} \texttt{smile}: piecewise-linear $f_{t}$ with possible jumps at change-points (designed to test NOT under misspecification), $T=2048$, $q=6$ change-points at $\tau=256, 512, 768, 1280, 1536, 1792$, with the corresponding sizes of jump $0, -4, 0, 0, 4, 0$ and changes in the slope $-2^{-5}, 0, 2^{-6}$, $2^{-6}, 0, -2^{-5}$, starting value for the intercept $f_{1}=0$ and slope $f_{2}-f_{1}=2^{-6}$, $\sigma_{t}=1$ for $t=1,\ldots,T$.
\end{enumerate}

\section{More details on the compututational aspects of NOT and its solution path}
\subsection{Computing contrast functions in a linear time}
\label{Sec:contrast_fun_linear_time}
The practical performance (in terms of computational cost) of Algorithm~\ref{Alg:not_algorithm} relies on the fast computation of the contrast functions discussed in Section~\ref{Sec:contrast_function} on any given interval $[s,e]$. Here we show that in all scenarios listed in Section~\ref{Sec:contrast_function}, the cost of computing $\{ \cont{s}{e}{b}{\Yb} \}_{b=s}^{e-1}$ is $O(e-s+1)$.

Note that the key ingredients in $ \cont{s}{e}{b}{\Yb}$ under the different scenarios are functions of the inner products, i.e. $\ip{\Yb}{\psib_{s,e}^b}$, $\ip{\Yb}{\phib_{s,e}^b}$, $\ip{\Yb}{\gammab_{s,b}}$, $\ip{\Yb}{\gammab_{b+1,e}}$, $\ip{\Yb}{\mathbf{1}^2_{s,b}}$, $\ip{\Yb}{\mathbf{1}^2_{b+1,e}}$, $\ip{\Yb^2}{\mathbf{1}^2_{s,b}}$ and $\ip{\Yb^2}{\mathbf{1}^2_{b+1,e}}$ for $b = s,\ldots, e-1$. 
For a fixed interval $[s,e]$, by simple algebra, we observe that $\ip{\Yb}{\psib_{s,e}^b}$ and $\ip{\Yb}{\phib_{s,e}^b}$ can be decomposed as 
\begin{align*}
\ip{\Yb}{\psib_{s,e}^b} &=  \overleftarrow{a}_{\psib,b}\sum_{t=s}^{b}Y_{t}-\overrightarrow{a}_{\psib,b}\sum_{t=b+1}^{e}Y_{t} \\
&:=  \overleftarrow{a}_{\psib,b}\overleftarrow{\pi}_{b}^{(0)}(\Yb) -  \overrightarrow{a}_{\psib,b} \overrightarrow{\pi}_{b}^{(0)}(\Yb),\\
\ip{\Yb}{\phib_{s,e}^b} &=  \overleftarrow{a}_{\phib,b}^{(1)}\sum_{t=s}^{b}t Y_{t}-\overrightarrow{a}_{\phib,b}^{(1)}\sum_{t=b+1}^{e}t Y_{t} + \overleftarrow{a}_{\phib,b}^{(0)}\sum_{t=s}^{b}Y_{t}-\overrightarrow{a}_{\phib,b}^{(0)}\sum_{t=b+1}^{e}Y_{t} \\
&:=  \overleftarrow{a}_{\phib,b}^{(1)} \overleftarrow{\pi}_{b}^{(1)}(\Yb) -  \overrightarrow{a}_{\phib,b}^{(1)} \overrightarrow{\pi}_{b}^{(1)}(\Yb) + \overleftarrow{a}_{\phib,b}^{(0)}\overleftarrow{\pi}_{b}^{(0)}(\Yb) -  \overrightarrow{a}_{\phib,b}^{(0)} \overrightarrow{\pi}_{b}^{(0)}(\Yb),
\end{align*}
where $\overleftarrow{a}_{\psib,b}, \overrightarrow{a}_{\psib,b}, \overleftarrow{a}_{\phib,b}^{(1)}, \overrightarrow{a}_{\phib,b}^{(1)}, \overleftarrow{a}_{\phib,b}^{(0)}$ and $\overrightarrow{a}_{\phib,b}^{(0)}$ are scalars that do not depend on $\Yb$, and can all be computed at the cost of $O(1)$ using equations given in Section~\ref{Sec:contrast_function}. Here for notational convenience, we use overhead arrows to indicate whether a scalar or a function is associated with observations to the left of $b$ (i.e. $[s,b]$, using $\overleftarrow{\cdot}$) or to the right of  $b$ (i.e. $[b+1,e]$, using $\overrightarrow{\cdot}$). We also suppress their dependence on $s$ and $e$ in the notation. In addition, the following recursive formulae hold
\begin{align*}
\overleftarrow{\pi}_{b+1}^{(k)}(\Yb) &= \overleftarrow{\pi}_{b}^{(k)}(\Yb) + (b+1)^{k} Y_{b+1},\\
\overrightarrow{\pi}_{b}^{(k)}(\Yb) &= \overrightarrow{\pi}_{b+1}^{(k)}(\Yb) + (b+1)^{k} Y_{b+1},
\end{align*}
with $\overleftarrow{\pi}_{s}^{(k)}(\Yb) = \overrightarrow{\pi}_{e}^{(k)}(\Yb) = 0$ for $k=0,1$. Consequently, $\overleftarrow{\pi}_{b}^{(k)}(\Yb)$ and $\overrightarrow{\pi}_{b}^{(k)}(\Yb)$ for all $b\in\{s,\ldots,e-1\}$ and $k=0,1$ (thereby  $\ip{\Yb}{\psib_{s,e}^b}$ and $\ip{\Yb}{\phib_{s,e}^b}$) can be computed in a single pass through $Y_{s}, \ldots, Y_{e}$. Similar approach can be applied to the remaining inner products involved in the definitions of the contrast functions given in Section~\ref{Sec:contrast_function}, which demonstrates that in all these cases the computation of $\{ \cont{s}{e}{b}{\Yb} \}_{b=s}^{e-1}$ scales linearly with the number of observations.

\subsection{Details of the NOT solution path algorithm}
\label{Sec:solution_path_algorithm_details} 
As mentioned in Section~\ref{Sec:solution_path_algorithm} of the main paper, we have developed Algorithm~\ref{Alg:not_solution_path} that computes the entire threshold-indexed solution path $\{\Tc(\zeta_{T})\}_{\zeta_T\geq 0}$ quickly, and have implemented it in our \proglang{R} package \pkg{not}.  We now provide its detailed pseudo-code on the next page. 

The construction of Algorithm~\ref{Alg:not_solution_path} stems from the following two observations. First, for any fixed threshold $\zeta_{T}$, Algorithm~\ref{Alg:not_algorithm} implies a binary tree data structure that is constructed according to the order of the detection of each change-point. More specifically, in our implementation, each tree node $\mathtt{N}$ contains information on the location of the detected change-point $\mathtt{N.b}$ over the interval of interest, $[\mathtt{N.s}, \mathtt{N.e}]$, along with the  maximum achieved value of the contrast function over all intervals in $F_T^M$ that are subsets of $[\mathtt{N.s}, \mathtt{N.e}]$ (the largest value and its location are denoted by $\mathtt{N.c}$ and $\mathtt{N.b}$, respectively). Moreover, we define $\mathtt{N.Left}$ and $\mathtt{N.Right}$ pointing to the nodes of the next detected change-points in $[\mathtt{N.s}, \mathtt{N.b}]$ and $[\mathtt{N.b}+1, \mathtt{N.e}]$, respectively. We then treat the first detected change-point over $[1,T]$ as the root of the tree and construct its branches in a recursive fashion afterwards. Second, suppose that we have already constructed the tree for $\zeta_{T}$ with root $\mathtt{N_{r}}$. For $\zeta_{T}'>\zeta_{T}$, the new tree's root is unchanged if $\mathtt{N_{r}.c} > \zeta_{T}'$. This observation remains valid for $\mathtt{N_{r}.Left}$ and $\mathtt{N_{r}.Right}$ and all subsequent nodes. Therefore, a branch of the tree has to be reconstructed only if $\mathtt{N.c} \leq \zeta_{T}'$ for some node $\mathrm{N}$. In this way, the tree constructed for $\zeta_T$ can be used as a starting point to finding the tree corresponding to $\zeta_T'$, thus significantly reducing the computational time in comparison to constructing the tree from scratch. 

Next, we elaborate on the complexity of Algorithm~\ref{Alg:not_solution_path}. As explained previously, finding solutions of Algorithm~\ref{Alg:not_algorithm} for a single threshold $\zeta_T$ is equivalent to the construction of a binary tree, which can be performed with the \textproc{BuildBinaryTree} routine given in Algorithm~\ref{Alg:not_solution_path}. Computational cost of this operation is no larger than $O(M K_{\zeta_{T}})$, where $K_{\zeta_{T}}$ denotes the height of the constructed binary tree with the threshold $\zeta_T$. The computational complexity of finding the entire solution path using Algorithm~\ref{Alg:not_solution_path} is therefore (in the worst case) $O(M K N)$, where $N$ and $K$ are, respectively, the number of solutions and the maximum tree depth over the entire solution path. However, this is a rough estimate which assumes that for each threshold on the path the binary tree has a different root node, which, from our empirical experience, is highly unlikely to occur in practice. Typically, the consecutive trees on the path differ just slightly (see e.g. our next Section~\ref{Sec:illustrative_example}), which significantly reduces the amount of computation that Algorithm~\ref{Alg:not_solution_path} requires. As such,  we find that the computational complexity of Algorithm~\ref{Alg:not_solution_path} is more like $O(MT)$ in practice. 

\newpage
\begin{algorithm}[!hb]
	\caption{NOT solution path}
	\label{Alg:not_solution_path}
	\small
	\begin{algorithmic}
		
		\Require Intervals $[s_m,e_m]$ and
		\begin{align*}
		b_{m} := \argmax_{s_m \le b \le e_m} \cont{s_m}{e_m}{b}{\Yb}, \quad
		c_{m} := \cont{s_m}{e_m}{b_m}{\Yb}, \quad
		l_{m} := e_m-s_m+1
		\end{align*}
		for all $m\in F_{T}^{M}$.
		\Ensure Thresholds $0 = \zeta^{(1)}_T < \ldots < \zeta^{(N)}_T$ and sets of estimated change-points $\Tc(\zeta^{(1)}_T), \ldots,\Tc(\zeta^{(N)}_T)$.\\
		\vspace{0.1cm} \\
		\hspace{-0.42cm} \textbf{To start the algorithm}: Call \Call{SolutionPath}{}()\\

		\Procedure{BuildBinaryTree}{$s$, $e$, $\zeta_{T}$, $\mathtt{N}$}
		\State $\Mc_{s,e} :=$ set of those $m\in\{1,\ldots,M\}$ such that $[s_{m}, e_{m}] \subset [s,e]$
		\State $\Oc_{s,e}:=$ set of $m\in\Mc_{s,e}$ such that $c_{m}>\zeta_{T}$
		\If{$\Oc_{s,e}=\emptyset$} $\mathtt{N}=\mathtt{NULL}$
		\Else
		\State $k:=\mbox{any elements of } \argmin_{m\in\Oc_{s,e}} l_{m}$
		\State $\mathtt{N.b} := b_{k}$, $\mathtt{N.c} := c_{k}$, $\mathtt{N.Left} := \mathtt{NULL}$, $\mathtt{N.Right} := \mathtt{NULL}$
		\State \Call{BuildBinaryTree}{$s$, $\mathtt{N.b}$, $\zeta_{T}$, $\mathtt{N.Left}$}
		\State \Call{BuildBinaryTree}{$\mathtt{N.b}+1$, $e$, $\zeta_{T}$,  $\mathtt{N.Right}$}
		\EndIf
		\EndProcedure\\
		
		\Procedure{UpdateBinaryTree}{$s$, $e$, $\zeta_{T}$, $\mathtt{N}$}
		
		\If{$\mathtt{N.c} \leq \zeta_{T}$}
		\State \Call{BuildBinaryTree}{$s$, $e$, $\zeta_{T}$, $\mathtt{N}$}
		\Else
		
		\If{$\mathtt{N.Left} \neq \mathtt{NULL}$}
		\State \Call{UpdateBinaryTree}{$s$, $\mathtt{N.b}$, $\zeta_{T}$, $\mathtt{N.Left}$}
		\EndIf
		
		\If{$\mathtt{N.Right} \neq \mathtt{NULL}$}
		\State \Call{UpdateBinaryTree}{$\mathtt{N.b}+1$, $e$, $\zeta_{T}$,  $\mathtt{N.Right}$}
		\EndIf		
		
		\EndIf
		\EndProcedure\\
		
		\Procedure{SolutionPath}{$ $}
		\State Set $\mathtt{N_r} := \mathtt{NULL}$, $i:=1$, $\zeta_T^{(1)}:=0$
		\State \Call{BuildBinaryTree}{$1$, $T$, $\zeta_T^{(1)}$, $\mathtt{N_r}$}
		\While{$\mathtt{N_r} \neq \mathtt{NULL}$}
		\State $\Dc:= \{\mathtt{N_r} \mbox{ and all its children nodes} \}$
		\State $\Tc(\zeta^{(i)}_T):= \{\mathtt{N.b} | \mathtt{N} \in \Dc \}$  
		\State $\zeta^{(i+1)}_T:= \min_{\mathtt{N} \in \Dc} \{\mathtt{N.c} \}$
		\State \Call{UpdateBinaryTree}{$1$, $T$, $\zeta^{(i+1)}_T$, $\mathtt{N_r}$}
		\State $i:=i+1$
		\EndWhile
		
		\EndProcedure	
	\end{algorithmic}
\end{algorithm}
\newpage

\subsection{An illustrative example}
\label{Sec:illustrative_example}

In this part, we revisit the example shown in the Introduction of our paper, and provide a simple illustration of how Algorithm~\ref{Alg:not_algorithm} and Algorithm~\ref{Alg:not_solution_path} work on a simulated dataset. Figure~\ref{Fig:illustrative_example_contrasts_and_fitted} shows the generated data $\{Y_{t}\}_{t=1}^{1000}$ following Scenario \ref{Scen:change_in_slope}, where the signal $f_t$ is as in \eqref{Eq:motivating_example} and $\sigma_t = 0.05$. The contrast function \eqref{Eq:contrast_kink} is evaluated for $5$ intervals. We observe that the contrast function corresponding to  $[1,1000]$, being the longest interval here, attains its maximum at $b=490$, which is far from the true change-points located at $\tau=350$ and $\tau=650$. Furthermore,  $\max_{1\leq b \leq 1000}\cont{s}{e}{b}{\Yb}$ is much larger than the corresponding value for the other intervals considered in Table~\ref{Table:illustrative_example_contrasts}. However, thanks to the fact that we focus on the narrowest-over-threshold intervals, Algorithm~\ref{Alg:not_algorithm} (for any $\zeta_{T}\in(0.08, 0.83)$) picks at its first iteration an interval with exactly one change-point (depending on $\zeta_{T}$, it is either $[225, 450]$ or $[500, 750]$) and the maximum of the contrast function computed is close to one of the true change-points.

\begin{figure}[!ht]
	\centering
	\null\hfill
	\subfloat[][]{
			\resizebox{0.48\textwidth}{!}{\input{tikz/illustrative_example_contrasts.tex}}
\label{Fig:illustrative_example_contrasts}}
	\hfill
	\subfloat[][]{
			\resizebox{0.48\textwidth}{!}{\input{tikz/illustrative_example_fitted.tex}}
\label{Fig:illustrative_example_fitted}}\\
	\hfill\null
	\caption{\label{Fig:illustrative_example_contrasts_and_fitted} An application of the NOT methodology to $Y_{t}$ generated from model \eqref{Eq:signal+noise1} with the signal $f_{t}$ given by $\eqref{Eq:motivating_example}$ and i.i.d. $\varepsilon_{t}\sim\Nc(0,0.05^2)$. Figure~\ref{Fig:illustrative_example_contrasts}: contrast function $\cont{s}{e}{b}{\Yb}$ given by \eqref{Eq:contrast_kink} evaluated for all $b \in[s,e]$ and intervals $[s,e]$ specified in Table~\ref{Table:illustrative_example_contrasts}. For intervals containing one change-point, $\cont{s}{e}{b}{\Yb}$ attains its maximum at $b$ close to the change-point. When there are two change-points (black solid line), the maximum is far from both change-points, despite $\max_{s\leq b \leq e}\cont{s}{e}{b}{\Yb}$ being large. Figure~\ref{Fig:illustrative_example_fitted}:  observed $Y_{t}$ (thin grey), true signal (thick dashed black), signal estimated picking the change-point candidate based on the interval corresponding to the largest contrast function (dotted-dashed navy) and the \textit{narrowest-over-threshold} intervals (dashed red).}
\end{figure}

Figure~\ref{Fig:illustrative_example_sol_path} shows how Algorithm~\ref{Alg:not_solution_path} proceeds in the example presented in Figure~\ref{Fig:illustrative_example_contrasts_and_fitted}. At the initial stage that can be seen in Figure~\ref{Fig:illustrative_example_sol_path_1}, the threshold is set to $\zeta_{T}^{(1)}=0$ and $b=471$, the maximum of the contrast function computed for the shortest interval $[450, 550]$ is taken as the root of the binary tree. Then we construct its left and right branches by considering only those intervals specified in Table~\ref{Table:illustrative_example_contrasts} whose endpoints  $[s,e]\subset[1,471]$ and $[s,e]\subset[472,1000]$, respectively, and the procedure continues for the resulting nodes. Next, the node with the smallest value of the contrast function is determined ($b=746$) and the threshold is set to the corresponding minimum $\zeta_{T}^{(2)}=0.03$. This guarantees that as Algorithm~\ref{Alg:not_solution_path} proceeds, there will be at least one update in the binary tree. In our example, the $b=746$ node  is removed and, as the maximum for $[500, 750]\subset[472,1000]$ exceeds the threshold, the $b=651$ node is inserted its place. Subsequently, we identify the node with the smallest contrast again ($b=471$), update the threshold to $\zeta_{T}^{(3)}=0.07$ and reconstruct the entire tree, as $b=471$ in Figure~\ref{Fig:illustrative_example_sol_path_2} constitutes its root. Algorithm~\ref{Alg:not_solution_path} keeps running until the resulting tree shrinks to $\mathtt{NULL}$.  In this example,  the fourth solution on the path (Figure~\ref{Fig:illustrative_example_sol_path_4}) contains exactly two nodes being close to the true change-points. 

\begin{table}
\caption{\label{Table:illustrative_example_contrasts} Intervals considered in Figure~\ref{Fig:illustrative_example_contrasts} and corresponding maxima of the contrast function  $\cont{s}{e}{b}{\Yb}$ given by \eqref{Eq:contrast_kink}, all calculated for a sample path of $Y_{t}$, $t=1,\ldots,1000$ generated from model \eqref{Eq:signal+noise1} with the signal $f_{t}$ given by \eqref{Eq:motivating_example} and the noise $\varepsilon_{t}\sim\Nc(0,0.05^2)$.} 
\centering
\fbox{%
\begin{tabular}{*{5}{c}}
$s$ & $e$ & $e-s+1$ & $\argmax_{s\leq b \leq e}\cont{s}{e}{b}{\Yb}$ & $\max_{s\leq b \leq e}\cont{s}{e}{b}{\Yb}$ \\ 
\hline
   1 & 1000 & 1000 & 490 & 10.19 \\ 
   10 & 245 & 236 &  43 & 0.08 \\ 
  225 & 450 & 226 & 344 & 0.76 \\ 
  500 & 750 & 251 & 651 & 0.83 \\ 
  740 & 950 & 211 & 746 & 0.03 \\ 
  450 & 550 & 101 & 471 & 0.07 \\ 
  
\end{tabular}}
\end{table}

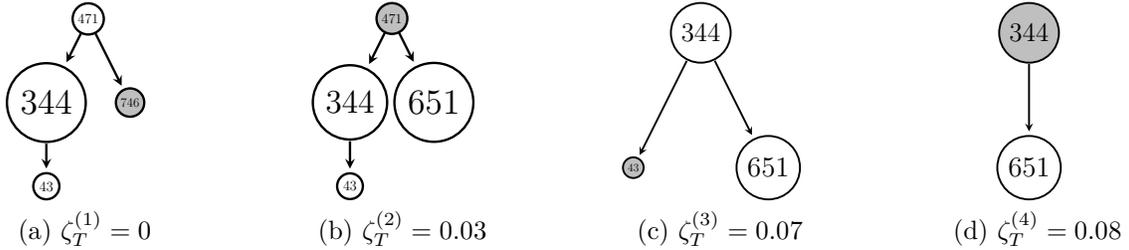
\begin{figure}[!htb]
	\centering
	\null\hfill
	\subfloat[\mbox{$\zeta_{T}^{(1)}=0$}]{\parbox{0.2\textwidth}{\centering
			\resizebox{!}{0.15\textwidth}{\begin{tikzpicture}[scale=3]
\tikzset{
	node/.style={circle,inner sep=1mm, draw,very thick,black,fill=white,text=black},
	smallestNode/.style={circle,inner sep=1mm, draw,very thick,black,fill=lightgray,text=black},
	arrow/.style={very thick,black,shorten >=2pt,-stealth},
}

	\node [node, scale=0.593829] (1) at (0.000000, 1.000000)	{471};
	\node [node, scale=1.500000] (2) at (-0.250000, 0.500000)	{344};
	\node [node, scale=0.607532] (3) at (-0.250000, 0.000000)	{43};
	\node [smallestNode, scale=0.539925] (4) at (0.250000, 0.500000)	{746};
	\path [arrow] (1) edge (2);
	\path [arrow] (1) edge (4);
	\path [arrow] (2) edge (3);

\end{tikzpicture}}
\label{Fig:illustrative_example_sol_path_1}}}
	\hfill
	\subfloat[\mbox{$\zeta_{T}^{(2)}=0.03$}]{\parbox{0.2\textwidth}{\centering
			\resizebox{!}{0.15\textwidth}{\begin{tikzpicture}[scale=3]
\tikzset{
	node/.style={circle,inner sep=1mm, draw,very thick,black,fill=white,text=black},
	smallestNode/.style={circle,inner sep=1mm, draw,very thick,black,fill=lightgray,text=black},
	arrow/.style={very thick,black,shorten >=2pt,-stealth},
}

	\node [smallestNode, scale=0.585590] (1) at (0.000000, 1.000000)	{471};
	\node [node, scale=1.412190] (2) at (-0.250000, 0.500000)	{344};
	\node [node, scale=0.598089] (3) at (-0.250000, 0.000000)	{43};
	\node [node, scale=1.500000] (4) at (0.250000, 0.500000)	{651};
	\path [arrow] (1) edge (2);
	\path [arrow] (1) edge (4);
	\path [arrow] (2) edge (3);

\end{tikzpicture}}
\label{Fig:illustrative_example_sol_path_2}}}
	\hfill
	\subfloat[\mbox{$\zeta_{T}^{(3)}=0.07$}]{\parbox{0.2\textwidth}{\centering
			\resizebox{!}{0.15\textwidth}{\begin{tikzpicture}[scale=3]
\tikzset{
	node/.style={circle,inner sep=1mm, draw,very thick,black,fill=white,text=black},
	smallestNode/.style={circle,inner sep=1mm, draw,very thick,black,fill=lightgray,text=black},
	arrow/.style={very thick,black,shorten >=2pt,-stealth},
}

	\node [node, scale=1.412190] (1) at (0.000000, 1.000000)	{344};
	\node [smallestNode, scale=0.598089] (2) at (-0.500000, 0.000000)	{43};
	\node [node, scale=1.500000] (3) at (0.500000, 0.000000)	{651};
	\path [arrow] (1) edge (2);
	\path [arrow] (1) edge (3);

\end{tikzpicture}}
\label{Fig:illustrative_example_sol_path_3}}}
	\hfill
	\subfloat[\mbox{$\zeta_{T}^{(4)}=0.08$}]{\parbox{0.2\textwidth}{\centering
			\resizebox{!}{0.15\textwidth}{\begin{tikzpicture}[scale=3]
\tikzset{
	node/.style={circle,inner sep=1mm, draw,very thick,black,fill=white,text=black},
	smallestNode/.style={circle,inner sep=1mm, draw,very thick,black,fill=lightgray,text=black},
	arrow/.style={very thick,black,shorten >=2pt,-stealth},
}

	\node [smallestNode, scale=1.412190] (1) at (0.000000, 1.000000)	{344};
	\node [node, scale=1.500000] (2) at (0.000000, 0.000000)	{651};
	\path [arrow] (1) edge (2);

\end{tikzpicture}}
\label{Fig:illustrative_example_sol_path_4}}}
	\hfill\null
	\caption{\label{Fig:illustrative_example_sol_path} First four segmentation trees obtained by Algorithm~\ref{Alg:not_solution_path} applied to a $Y_{1},\ldots, Y_{1000}$ presented in Figure~\ref{Fig:illustrative_example_contrasts_and_fitted}. The larger the node, the larger the corresponding value of  $\max_{s\leq b \leq e}\cont{s}{e}{b}{\Yb}$ given by \eqref{Eq:contrast_kink}. The grey nodes correspond to the smallest contrast function for each tree and are updated as Algorithm~\ref{Alg:not_solution_path} proceeds. }
\end{figure}

\section{Additional simulation results}
\label{Sec:sim_additional}
In addition to the results presented in Section~\ref{Sec:simulation_study}, here we present Tables \ref{Table:sim_results_gaussian_sigma__1.4142135623731}--\ref{Table:sim_results_ar_sigma_1} that  summarise the results for three different distributions of the noise $\varepsilon_{t}$, where (b) $\varepsilon_{t}\stackrel{i.i.d.}{\sim} \Nc(0,2)$, (c) $\varepsilon_{t}\stackrel{i.i.d.} \sim \mathrm{Laplace}(0,2^{-1/2})$ and (e) $\varepsilon_{t}$ follows zero-mean unit-variance Gaussian AR(1) with $\varphi = 0.3$.

\begin{table}
\caption{Distribution of $\hat{q}-q$ for data generated according to \eqref{Eq:signal+noise} with the noise term $\varepsilon_{t}$ being i.i.d. $\mathcal{N}(0,2)$ for various choices of $f_{t}$ and $\sigma_{t}$ given in Section~\ref{Sec:simulation_models} and competing methods listed in Section~\ref{Sec:simulation_study}. Also, the average Mean-Square Error of the resulting estimate of the signal $f_{t}$, average Hausdorff distance $d_H$ given by \eqref{Eq:hausdorff_distance} and average computation time in seconds using a single core of an Intel Xeon 3.6 GHz CPU with 16 GB of RAM, all calculated over $100$ simulated data sets. Bold: methods with the largest empirical frequency of $\hat{q}-q=0$ or smallest average $d_{H}$ and those within $10\%$ of the highest, or, respectively, within $10\%$ of the lowest. \label{Table:sim_results_gaussian_sigma__1.4142135623731}} 
\footnotesize
\centering
\fbox{
	\begin{tabular}{*{11}{c|}c}
	&&\multicolumn{7}{c|}{$\hat{q}-q$}&&&\\
	Method & Model &$\leq-3$ & $-2$ & $-1$ & $0$ & $1$ & $2$ & $\geq 3$ & $\mbox{MSE}$ & $d_{H} \times 10^2$ & time\\
	\hline
	 B\&P & \multirow{11}{*}{\ref{Model:teeth}} & 82 & 9 & 2 & 7 & 0 & 0 & 0 & 0.832 & 14.15 & 0.26 \\ 
  e-cp3o &  & 0 & 0 & 0 & \textbf{100} & 0 & 0 & 0 & 0.109 & \textbf{1.02} & 2.15 \\ 
  FDRSeg &  & 0 & 0 & 0 & 82 & 12 & 4 & 2 & 0.136 & 1.65 & 0.09 \\   
  NMCD &  & 0 & 0 & 0 & \textbf{98} & 2 & 0 & 0 & 0.149 & 1.43 & 1.28 \\ 
  NOT &  & 0 & 0 & 0 & \textbf{99} & 1 & 0 & 0 & 0.112 & \textbf{1.05} & 0.08 \\ 
  NOT HT &  & 0 & 0 & 0 & \textbf{97} & 3 & 0 & 0 & 0.127 & 1.35 & 0.09 \\ 
  NP-PELT &  & 0 & 0 & 0 & 73 & 24 & 2 & 1 & 0.131 & 1.43 & 0.04 \\ 
  PELT &  & 0 & 0 & 0 & \textbf{100} & 0 & 0 & 0 & 0.11 & \textbf{1.04} & 0 \\ 
  S3IB &  & 0 & 0 & 0 & \textbf{94} & 5 & 1 & 0 & 0.113 & 1.17 & 0.11 \\ 
  SMUCE &  & 0 & 1 & 15 & 84 & 0 & 0 & 0 & 0.192 & 2.23 & 0.23 \\ 
  WBS &  & 0 & 0 & 0 & \textbf{98} & 2 & 0 & 0 & 0.11 & \textbf{1.05} & 0.11 \\ 
 
	\hline
	 B\&P & \multirow{11}{*}{\ref{Model:blocks}} & 100 & 0 & 0 & 0 & 0 & 0 & 0 & 0.358 & 14.34 & 5.64 \\ 
  e-cp3o &  & 100 & 0 & 0 & 0 & 0 & 0 & 0 & 0.142 & 8.12 & 194.18 \\ 
  FDRSeg &  & 7 & 30 & 42 & 15 & 6 & 0 & 0 & 0.063 & \textbf{3.19} & 3.27 \\   
  NMCD &  & 37 & 31 & 26 & 5 & 1 & 0 & 0 & 0.073 & 4.02 & 5.06 \\ 
  NOT &  & 27 & 28 & 25 & 17 & 2 & 1 & 0 & 0.062 & \textbf{3.48} & 0.11 \\ 
  NOT HT &  & 42 & 27 & 23 & 7 & 1 & 0 & 0 & 0.076 & 4.23 & 0.23 \\ 
  NP-PELT &  & 1 & 12 & 26 & \textbf{25} & 17 & 16 & 3 & 0.067 & 3.91 & 0.54 \\ 
  PELT &  & 92 & 7 & 0 & 1 & 0 & 0 & 0 & 0.106 & 7.28 & 0.01 \\ 
  S3IB &  & 35 & 23 & 24 & 17 & 0 & 1 & 0 & 0.065 & 3.94 & 0.53 \\ 
  SMUCE &  & 100 & 0 & 0 & 0 & 0 & 0 & 0 & 0.139 & 5.72 & 0.04 \\ 
  WBS &  & 30 & 26 & 27 & 16 & 1 & 0 & 0 & 0.064 & 3.64 & 0.22 \\ 
   
	\hline
	 B\&P & \multirow{3}{*}{\ref{Model:wave1}} & 0 & 0 & 100 & 0 & 0 & 0 & 0 & 0.246 & 3.94 & 146.74 \\ 
  NOT &  & 0 & 0 & 0 & \textbf{99} & 1 & 0 & 0 & 0.032 & \textbf{1.47} & 0.54 \\ 
  TF &  & 0 & 0 & 0 & 0 & 0 & 0 & 100 & 0.032 & 8.42 & 63.71 \\ 
  
	\hline
	 B\&P & \multirow{3}{*}{\ref{Model:wave2}} & 16 & 55 & 28 & 1 & 0 & 0 & 0 & 0.336 & 6.48 & 167.31 \\ 
  NOT &  & 0 & 0 & 0 & \textbf{98} & 2 & 0 & 0 & 0.039 & \textbf{2.08} & 0.47 \\ 
  TF &  & 0 & 0 & 0 & 0 & 0 & 0 & 100 & 0.031 & 4.44 & 64.41 \\ 
  
	\hline
	 B\&P & \multirow{3}{*}{\ref{Model:mix}} & 0 & 0 & 8 & \textbf{92} & 0 & 0 & 0 & 0.044 & \textbf{3.31} & 380.84 \\ 
  NOT &  & 0 & 0 & 5 & \textbf{93} & 2 & 0 & 0 & 0.045 & \textbf{3.52} & 0.48 \\ 
  TF &  & 0 & 0 & 0 & 0 & 0 & 0 & 100 & 0.041 & 5.89 & 78.46 \\ 
  
	\hline
	 e-cp3o & \multirow{6}{*}{\ref{Model:vol}} & 95 & 2 & 0 & 3 & 0 & 0 & 0 & 0.372 & 16.55 & 11.67 \\ 
  NMCD &  & 0 & 0 & 15 & 79 & 6 & 0 & 0 & 0.058 & 3.35 & 4.78 \\ 
  NOT &  & 0 & 0 & 10 & \textbf{89} & 1 & 0 & 0 & 0.045 & \textbf{2.07} & 1.22 \\ 
  NP-PELT &  & 0 & 0 & 0 & 22 & 24 & 22 & 32 & 0.12 & 2.97 & 0.61 \\ 
  PELT &  & 11 & 15 & 28 & 44 & 2 & 0 & 0 & 0.075 & 7.83 & 0.02 \\ 
  SegNeigh &  & 0 & 0 & 8 & 60 & 17 & 10 & 5 & 0.054 & 2.5 & 38.05 \\   
  
	\hline
	 B\&P & \multirow{3}{*}{\ref{Model:quad}} & 0 & 0 & 35 & \textbf{65} & 0 & 0 & 0 & 0.066 & 6.47 & 44.26 \\ 
  NOT &  & 0 & 1 & 37 & \textbf{62} & 0 & 0 & 0 & 0.064 & \textbf{5.78} & 0.31 \\ 
  TF &  & 0 & 0 & 0 & 0 & 0 & 1 & 99 & 0.075 & 22.71 & 60.17 \\ 
  
\end{tabular}}
\end{table}

\begin{table}

\caption{Distribution of $\hat{q}-q$ for data generated according to \eqref{Eq:signal+noise} with the noise term $\varepsilon_{t}$ being i.i.d. $\mbox{Laplace}\left(0,(\sqrt{2})^{-1}\right)$ (N.B. $\Var(\varepsilon_{t})=1$ here) for various choices of $f_{t}$ and $\sigma_{t}$ given in Section~\ref{Sec:simulation_models} and competing methods listed in Section~\ref{Sec:simulation_study}. Also, the average Mean-Square Error of the resulting estimate of the signal $f_{t}$, average Hausdorff distance $d_H$ given by \eqref{Eq:hausdorff_distance} and average computation time in seconds using a single core of an Intel Xeon 3.6 GHz CPU with 16 GB of RAM, all calculated over $100$ simulated data sets. Bold: methods with the largest empirical frequency of $\hat{q}-q=0$ or smallest average $d_{H}$ and those within $10\%$ of the highest, or, respectively, within $10\%$ of the lowest. \label{Table:sim_results_laplace_sigma_1}} 
\footnotesize
\centering
\fbox{
	\begin{tabular}{*{11}{c|}c}
	&&\multicolumn{7}{c|}{$\hat{q}-q$}&&&\\
	Method & Model &$\leq-3$ & $-2$ & $-1$ & $0$ & $1$ & $2$ & $\geq 3$ & $\mbox{MSE}$ & $d_{H} \times 10^2$ & time\\
	\hline
	 B\&P & \multirow{11}{*}{\ref{Model:teeth}} & 76 & 4 & 1 & 19 & 0 & 0 & 0 & 0.745 & 13.04 & 0.25 \\ 
  e-cp3o &  & 0 & 0 & 0 & \textbf{100} & 0 & 0 & 0 & 0.097 & 0.87 & 2.13 \\ 
  FDRSeg &  & 0 & 0 & 0 & 5 & 4 & 6 & 85 & 0.199 & 4.78 & 0.13 \\   
  NMCD &  & 0 & 0 & 0 & \textbf{94} & 6 & 0 & 0 & 0.141 & 1.35 & 1.28 \\ 
  NOT &  & 0 & 1 & 0 & \textbf{95} & 3 & 1 & 0 & 0.107 & 1.19 & 0.08 \\ 
  NOT HT &  & 0 & 0 & 0 & \textbf{99} & 0 & 1 & 0 & 0.093 & \textbf{0.79} & 0.09 \\ 
  NP-PELT &  & 0 & 0 & 0 & 71 & 22 & 6 & 1 & 0.141 & 1.57 & 0.04 \\ 
  PELT &  & 0 & 0 & 0 & 69 & 13 & 14 & 4 & 0.145 & 1.4 & 0 \\ 
  S3IB &  & 0 & 1 & 0 & 76 & 10 & 9 & 4 & 0.136 & 1.47 & 0.11 \\ 
  SMUCE &  & 0 & 0 & 1 & 52 & 23 & 14 & 10 & 0.155 & 2.6 & 0.21 \\ 
  WBS &  & 0 & 0 & 0 & 64 & 4 & 23 & 9 & 0.151 & 1.91 & 0.11 \\ 
	\hline
	 B\&P & \multirow{11}{*}{\ref{Model:blocks}} & 100 & 0 & 0 & 0 & 0 & 0 & 0 & 0.311 & 12.55 & 5.36 \\ 
  e-cp3o &  & 100 & 0 & 0 & 0 & 0 & 0 & 0 & 0.147 & 9.1 & 191.73 \\ 
  FDRSeg &  & 0 & 0 & 0 & 0 & 0 & 0 & 100 & 0.1 & 7.96 & 3.06 \\   
  NMCD &  & 15 & 36 & 37 & 12 & 0 & 0 & 0 & 0.06 & \textbf{3.37} & 5.06 \\ 
  NOT &  & 51 & 21 & 17 & 9 & 2 & 0 & 0 & 0.079 & 4.8 & 0.11 \\ 
  NOT HT &  & 23 & 26 & 36 & 15 & 0 & 0 & 0 & 0.054 & \textbf{3.08} & 0.23 \\ 
  NP-PELT &  & 0 & 4 & 10 & 19 & 27 & 19 & 21 & 0.077 & 4.03 & 0.51 \\ 
  PELT &  & 20 & 21 & 19 & 14 & 14 & 6 & 6 & 0.108 & 5.02 & 0.01 \\ 
  S3IB &  & 88 & 8 & 2 & 2 & 0 & 0 & 0 & 0.13 & 10.22 & 0.5 \\ 
  SMUCE &  & 14 & 16 & 23 & \textbf{22} & 6 & 8 & 11 & 0.108 & 6.02 & 0.03 \\ 
  WBS &  & 21 & 12 & 12 & 15 & 15 & 10 & 15 & 0.104 & 4.98 & 0.22 \\  
   
	\hline
	 B\&P & \multirow{3}{*}{\ref{Model:wave1}} & 0 & 0 & 100 & 0 & 0 & 0 & 0 & 0.261 & 4.16 & 147.23 \\ 
  NOT &  & 0 & 0 & 1 & \textbf{96} & 1 & 1 & 1 & 0.037 & \textbf{1.89} & 0.52 \\ 
  TF &  & 0 & 0 & 0 & 0 & 0 & 0 & 100 & 0.035 & 8.42 & 64.08 \\ 
  
	\hline
	 B\&P & \multirow{3}{*}{\ref{Model:wave2}} & 16 & 44 & 37 & 3 & 0 & 0 & 0 & 0.323 & 6.27 & 171.88 \\ 
  NOT &  & 0 & 0 & 0 & \textbf{96} & 3 & 1 & 0 & 0.042 & \textbf{2.24} & 0.44 \\ 
  TF &  & 0 & 0 & 0 & 0 & 0 & 0 & 100 & 0.032 & 4.38 & 66.53 \\ 
  
	\hline
	 B\&P & \multirow{3}{*}{\ref{Model:mix}} & 0 & 1 & 6 & \textbf{93} & 0 & 0 & 0 & 0.045 & \textbf{3.44} & 384.72 \\ 
  NOT &  & 0 & 1 & 2 & \textbf{90} & 3 & 3 & 1 & 0.047 & \textbf{3.48} & 0.5 \\ 
  TF &  & 0 & 0 & 0 & 0 & 0 & 0 & 100 & 0.041 & 5.91 & 78.1 \\ 
  
	\hline
	 e-cp3o & \multirow{6}{*}{\ref{Model:vol}} & 96 & 3 & 1 & 0 & 0 & 0 & 0 & 0.481 & 17.95 & 11.91 \\ 
  NMCD &  & 1 & 28 & 38 & 30 & 2 & 0 & 1 & 0.098 & 9.45 & 4.83 \\ 
  NOT &  & 1 & 10 & 42 & \textbf{35} & 9 & 1 & 2 & 0.188 & 8.17 & 1.24 \\ 
  NP-PELT &  & 0 & 1 & 4 & 14 & 22 & 16 & 43 & 0.359 & \textbf{5.34} & 0.75 \\ 
  PELT &  & 22 & 22 & 35 & 17 & 3 & 1 & 0 & 0.215 & 12.8 & 0.03 \\ 
  SegNeigh &  & 1 & 1 & 13 & 24 & 27 & 20 & 14 & 0.183 & 6.41 & 38.29 \\   
	\hline
	 B\&P & \multirow{3}{*}{\ref{Model:quad}} & 0 & 0 & 41 & \textbf{59} & 0 & 0 & 0 & 0.066 & \textbf{5.93} & 44.19 \\ 
  NOT &  & 0 & 2 & 51 & 44 & 2 & 1 & 0 & 0.077 & 7.7 & 0.32 \\ 
  TF &  & 0 & 0 & 0 & 0 & 0 & 0 & 100 & 0.075 & 22.42 & 60.33 \\ 
  
\end{tabular}}
\end{table}

\begin{table}
\caption{Distribution of $\hat{q}-q$ for data generated according to \eqref{Eq:signal+noise} with the noise term $\varepsilon_{t}$ being a zero-mean unit-variance Gaussian AR(1) process with $\varphi = 0.3$ for various choices of $f_{t}$ and $\sigma_{t}$ given in Section~\ref{Sec:simulation_models} and competing methods listed in Section~\ref{Sec:simulation_study}. Also, the average Mean-Square Error of the resulting estimate of the signal $f_{t}$, average Hausdorff distance $d_H$ given by \eqref{Eq:hausdorff_distance} and average computation time in seconds using a single core of an Intel Xeon 3.6 GHz CPU with 16 GB of RAM, all calculated over $100$ simulated data sets. Bold: methods with the largest empirical frequency of $\hat{q}-q=0$ or smallest average $d_{H}$ and those within $10\%$ of the highest, or, respectively, within $10\%$ of the lowest. \label{Table:sim_results_ar_sigma_1}} 
\footnotesize
\centering
\fbox{
	\begin{tabular}{*{11}{c|}c}
	&&\multicolumn{7}{c|}{$\hat{q}-q$}&&&\\
	Method & Model &$\leq-3$ & $-2$ & $-1$ & $0$ & $1$ & $2$ & $\geq 3$ & $\mbox{MSE}$ & $d_{H} \times 10^2$ & time\\
	\hline
	 B\&P & \multirow{11}{*}{\ref{Model:teeth}} & 78 & 12 & 0 & 10 & 0 & 0 & 0 & 0.783 & 12.87 & 0.25 \\ 
  e-cp3o &  & 0 & 0 & 0 & \textbf{100} & 0 & 0 & 0 & 0.084 & \textbf{0.99} & 2.16 \\ 
  FDRSeg &  & 0 & 0 & 0 & 0 & 2 & 2 & 96 & 0.196 & 5.52 & 0.09 \\  
  NMCD &  & 0 & 0 & 0 & 71 & 18 & 10 & 1 & 0.138 & 1.88 & 1.60 \\ 
  NOT &  & 0 & 0 & 0 & 72 & 13 & 4 & 11 & 0.104 & 1.84 & 0.07 \\ 
  NOT HT &  & 0 & 0 & 0 & 84 & 9 & 4 & 3 & 0.099 & 1.51 & 0.08 \\ 
  NP-PELT &  & 0 & 0 & 0 & 40 & 34 & 16 & 10 & 0.122 & 2.38 & 0.03 \\ 
  PELT &  & 0 & 0 & 0 & 74 & 17 & 6 & 3 & 0.097 & 1.44 & 0.01 \\ 
  S3IB &  & 0 & 0 & 0 & 79 & 14 & 6 & 1 & 0.092 & 1.42 & 0.13 \\ 
  SMUCE &  & 0 & 0 & 0 & 55 & 35 & 5 & 5 & 0.123 & 2.43 & 0.17 \\ 
  WBS &  & 0 & 0 & 0 & 62 & 18 & 10 & 10 & 0.105 & 2.19 & 0.11 \\

	\hline
	 B\&P & \multirow{11}{*}{\ref{Model:blocks}} & 100 & 0 & 0 & 0 & 0 & 0 & 0 & 0.318 & 12.5 & 5.65 \\ 
  e-cp3o &  & 100 & 0 & 0 & 0 & 0 & 0 & 0 & 0.134 & 6.36 & 195.55 \\ 
  FDRSeg &  & 0 & 0 & 0 & 0 & 0 & 0 & 100 & 0.117 & 9.13 & 1.52 \\   
  NMCD &  & 0 & 12 & 43 & 37 & 8 & 0 & 0 & 0.052 & 2.63 & 6.43 \\  
  NOT &  & 2 & 9 & 35 & 28 & 16 & 6 & 4 & 0.048 & 2.99 & 0.11 \\ 
  NOT HT &  & 5 & 14 & 41 & 18 & 16 & 4 & 2 & 0.053 & 3.63 & 0.21 \\ 
  NP-PELT &  & 0 & 1 & 6 & 5 & 18 & 9 & 61 & 0.066 & 4.88 & 0.28 \\ 
  PELT &  & 1 & 6 & 25 & \textbf{48} & 14 & 2 & 4 & 0.046 & \textbf{2.24} & 0.01 \\ 
  S3IB &  & 14 & 26 & 36 & 23 & 1 & 0 & 0 & 0.05 & 3.21 & 0.53 \\ 
  SMUCE &  & 1 & 12 & 35 & 25 & 17 & 6 & 4 & 0.053 & 4.56 & 0.03 \\ 
  WBS &  & 1 & 9 & 36 & 32 & 11 & 7 & 4 & 0.047 & 2.9 & 0.21 \\

	\hline
	 B\&P & \multirow{3}{*}{\ref{Model:wave1}} & 0 & 0 & 92 & 8 & 0 & 0 & 0 & 0.244 & 4.43 & 145.12 \\ 
  NOT &  & 0 & 0 & 0 & \textbf{99} & 1 & 0 & 0 & 0.031 & \textbf{1.41} & 0.55 \\  
  TF &  & 0 & 0 & 0 & 0 & 0 & 0 & 100 & 0.453 & 9.08 & 69.77 \\ 
  
	\hline
	 B\&P & \multirow{3}{*}{\ref{Model:wave2}} & 0 & 3 & 19 & 78 & 0 & 0 & 0 & 0.127 & 3.36 & 174.04 \\ 
  NOT &  & 0 & 0 & 0 & \textbf{97} & 2 & 1 & 0 & 0.035 & \textbf{2} & 0.64 \\ 
  TF &  & 0 & 0 & 0 & 0 & 0 & 0 & 100 & 0.458 & 5 & 72.36 \\  
  
	\hline
	 B\&P & \multirow{3}{*}{\ref{Model:mix}} & 0 & 0 & 0 & \textbf{100} & 0 & 0 & 0 & 0.037 &  \textbf{3.05} & 383.09 \\ 
  NOT &  & 0 & 0 & 0 & \textbf{92} & 7 & 0 & 1 & 0.04 & \textbf{3.32} & 0.52 \\ 
  TF &  & 0 & 0 & 0 & 0 & 0 & 0 & 100 & 0.224 & 6.24 & 80.77 \\ 
  
	\hline
	 e-cp3o & \multirow{6}{*}{\ref{Model:vol}} & 78 & 13 & 3 & 4 & 2 & 0 & 0 & 0.368 & 15.08 & 11.75  \\ 
  NMCD &  & 0 & 0 & 7 & 33 & 30 & 15 & 15 & 0.167 & 5.23 & 4.76 \\ 
  NOT &  & 0 & 1 & 21 & \textbf{67} & 8 & 2 & 1 & 0.099 & 4.14 & 1.21 \\ 
  NP-PELT &  & 0 & 0 & 0 & 0 & 1 & 2 & 97 & 0.457 & 5.48 & 0.63 \\ 
  PELT &  & 11 & 9 & 32 & 43 & 5 & 0 & 0 & 0.107 & 8.18 & 0.02 \\ 
  SegNeigh &  & 0 & 0 & 7 & 31 & 33 & 18 & 11 & 0.125 & \textbf{3.6} & 42.11 \\ 
	\hline
	s B\&P & \multirow{3}{*}{\ref{Model:quad}} & 0 & 0 & 2 & \textbf{88} & 9 & 1 & 0 & 0.046 & 4.17 &  44.01 \\ 
  NOT &  & 0 & 0 & 2 & \textbf{85} & 12 & 1 & 0 & 0.046 & \textbf{3.46} & 0.31 \\ 
  TF &  & 0 & 0 & 0 & 0 & 0 & 0 & 100 & 0.115 & 24.44 & 60.11 \\ 
  
\end{tabular}}
\end{table}

\section{Additional real data example: OPEC Reference Basket oil price}
\label{Sec:oil_analysis}

We perform change-point analysis on the daily Organisation of the Petroleum Exporting Countries (OPEC) Reference Basket oil price from 1 January, 2003 to 15 July, 2016. The data were obtained from the OPEC database through the \proglang{R} package \pkg{Quandl} \citep{mctaggart2016quandl}. Instead of working with the raw price series, we analyse the log-returns series $Y_{t}=100 \log\left(P_{t}/P_{t-1}\right)$, where $P_{t}$ denotes the daily oil price.  One of the stylised facts of the financial time series data is that the autocorrelation of assets returns are weak, while squared returns tend to exhibit strong autocorrelation, which is the case for the oil price time series (see Figure~\ref{Fig:oil_log_returns_sq_acf}). This phenomenon can be possibly explained by the existence of the structural breaks in the mean and variance structure of the data series \citep{mikosch2004nonstationarities,fryzlewicz2006haar}. In this study, we apply NOT with the contrast function given by \eqref{Eq:contrast_volatility}, which is designed to detect changes in both the mean and the volatility, as in Scenario~\ref{Scen:change_in_mean_and_variance}. For comparison, we also report change-points detected with the NMCD method of \cite{zou2014nonparametric}. 

We apply Algorithm~\ref{Alg:not_solution_path} to compute the NOT solution path and choose the model achieving the lowest SIC given by \eqref{Eq:SIC}, setting the number of intervals drawn $M=10000$ and the maximum number of change-points $q_{max}=25$.  Computations for the solution path and model selection are performed using the \proglang{R} package \pkg{not} \citep{baranowski2016notpackage}. For the NMCD procedure, we use the \code{nmcd} routine from the \proglang{R} package \pkg{nmcdr} \citep{zoulancezhange2014nmcdr}, setting the maximum number of change-points to $q_{max}=25$ as well.

Figure~\ref{Fig:oil_analysis} illustrates the results of our analysis. The oil price time series and the locations of the change-points identified by NOT and NMCD can be seen in Figure~\ref{Fig:oil_price}. Both methods discover $7$ change-points, largely agreeing on their locations, in the sense that for $6$ out of $7$ features NOT detects, NMCD detects a change-point nearby. However, NMCD does not indicate any change-point around the first change-point identified by NOT on 29 April 2003. This date could potentially be related to the end of the 2003 invasion of Iraq, which initiated the upward trend in the oil price lasting almost ceaselessly until the beginning of the 2008--09 financial crisis. On the other hand, NMCD indicates two change-points in the first quarter of 2016, while NOT only finds one in that period.  Table~\ref{Table:oil_price_changepoints} lists the exact locations of the change-points detected by the two methods and the events that coincide with some of them.  Figure~\ref{Fig:oil_residuals_sq_acf} shows the autocorrelation function for the squared residuals obtained by subtracting the sample mean and dividing by the standard deviations from the data in each segment. It appears that there is little autocorrelation in the squares of the residuals,  suggesting that Scenario~\ref{Scen:change_in_mean_and_variance} fits the data in this example reasonably well.

\begin{figure}[!ht]
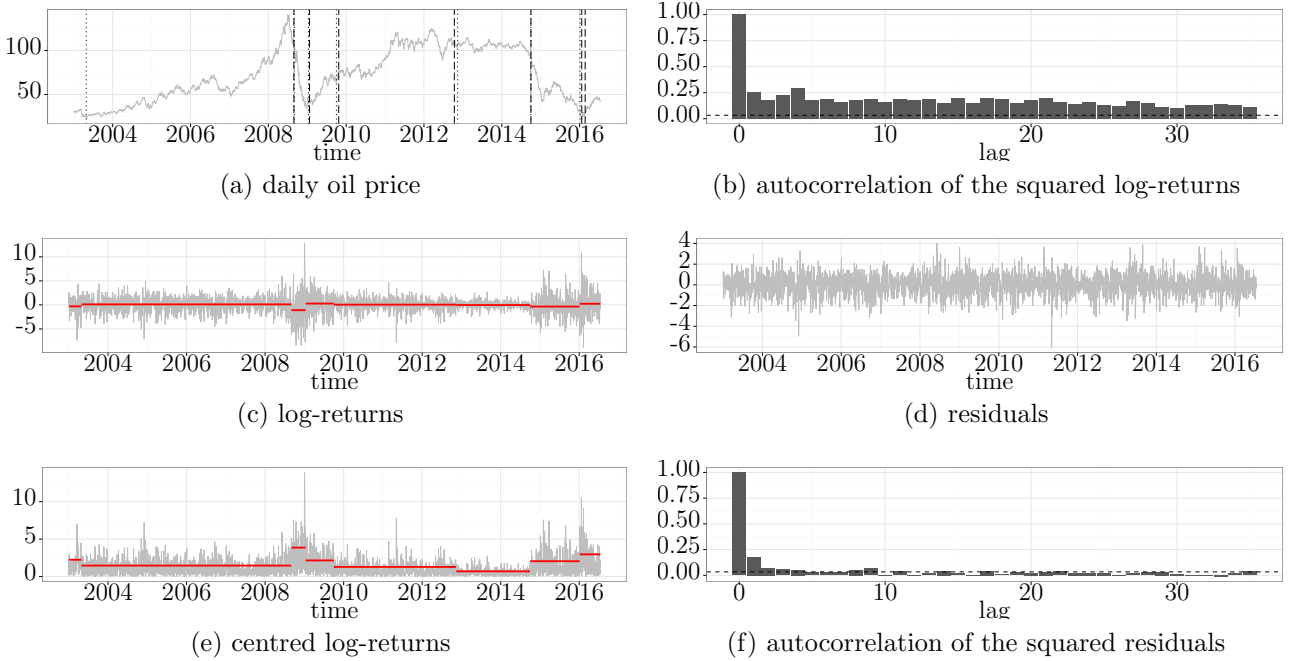

	\centering
	\subfloat[][daily oil price]{
			\resizebox{0.48\textwidth}{!}{\input{tikz/oil_price.tex}}
\label{Fig:oil_price}}
	\subfloat[][autocorrelation of the squared log-returns]{
			\resizebox{0.48\textwidth}{!}{\begin{tikzpicture}[x=1pt,y=1pt]
\definecolor{fillColor}{RGB}{255,255,255}
\path[use as bounding box,fill=fillColor,fill opacity=0.00] (0,0) rectangle (578.16,144.54);
\begin{scope}
\path[clip] (  0.00,  0.00) rectangle (578.16,144.54);
\definecolor{drawColor}{RGB}{255,255,255}
\definecolor{fillColor}{RGB}{255,255,255}

\path[draw=drawColor,line width= 0.6pt,line join=round,line cap=round,fill=fillColor] (  0.00, -0.00) rectangle (578.16,144.54);
\end{scope}
\begin{scope}
\path[clip] ( 43.93, 33.48) rectangle (572.16,138.54);
\definecolor{fillColor}{RGB}{255,255,255}

\path[fill=fillColor] ( 43.93, 33.48) rectangle (572.16,138.54);
\definecolor{drawColor}{gray}{0.98}

\path[draw=drawColor,line width= 0.6pt,line join=round] ( 43.93, 50.19) --
	(572.16, 50.19);

\path[draw=drawColor,line width= 0.6pt,line join=round] ( 43.93, 74.07) --
	(572.16, 74.07);

\path[draw=drawColor,line width= 0.6pt,line join=round] ( 43.93, 97.95) --
	(572.16, 97.95);

\path[draw=drawColor,line width= 0.6pt,line join=round] ( 43.93,121.83) --
	(572.16,121.83);

\path[draw=drawColor,line width= 0.6pt,line join=round] (140.84, 33.48) --
	(140.84,138.54);

\path[draw=drawColor,line width= 0.6pt,line join=round] (274.60, 33.48) --
	(274.60,138.54);

\path[draw=drawColor,line width= 0.6pt,line join=round] (408.37, 33.48) --
	(408.37,138.54);

\path[draw=drawColor,line width= 0.6pt,line join=round] (542.13, 33.48) --
	(542.13,138.54);
\definecolor{drawColor}{gray}{0.90}

\path[draw=drawColor,line width= 0.2pt,line join=round] ( 43.93, 38.25) --
	(572.16, 38.25);

\path[draw=drawColor,line width= 0.2pt,line join=round] ( 43.93, 62.13) --
	(572.16, 62.13);

\path[draw=drawColor,line width= 0.2pt,line join=round] ( 43.93, 86.01) --
	(572.16, 86.01);

\path[draw=drawColor,line width= 0.2pt,line join=round] ( 43.93,109.89) --
	(572.16,109.89);

\path[draw=drawColor,line width= 0.2pt,line join=round] ( 43.93,133.76) --
	(572.16,133.76);

\path[draw=drawColor,line width= 0.2pt,line join=round] ( 73.96, 33.48) --
	( 73.96,138.54);

\path[draw=drawColor,line width= 0.2pt,line join=round] (207.72, 33.48) --
	(207.72,138.54);

\path[draw=drawColor,line width= 0.2pt,line join=round] (341.48, 33.48) --
	(341.48,138.54);

\path[draw=drawColor,line width= 0.2pt,line join=round] (475.25, 33.48) --
	(475.25,138.54);
\definecolor{fillColor}{gray}{0.35}

\path[fill=fillColor] ( 67.94, 38.25) rectangle ( 79.98,133.76);

\path[fill=fillColor] ( 81.31, 38.25) rectangle ( 93.35, 62.51);

\path[fill=fillColor] ( 94.69, 38.25) rectangle (106.73, 55.37);

\path[fill=fillColor] (108.07, 38.25) rectangle (120.11, 59.50);

\path[fill=fillColor] (121.44, 38.25) rectangle (133.48, 65.90);

\path[fill=fillColor] (134.82, 38.25) rectangle (146.86, 55.35);

\path[fill=fillColor] (148.20, 38.25) rectangle (160.23, 56.14);

\path[fill=fillColor] (161.57, 38.25) rectangle (173.61, 53.12);

\path[fill=fillColor] (174.95, 38.25) rectangle (186.99, 55.27);

\path[fill=fillColor] (188.33, 38.25) rectangle (200.36, 56.43);

\path[fill=fillColor] (201.70, 38.25) rectangle (213.74, 53.09);

\path[fill=fillColor] (215.08, 38.25) rectangle (227.12, 55.74);

\path[fill=fillColor] (228.45, 38.25) rectangle (240.49, 54.78);

\path[fill=fillColor] (241.83, 38.25) rectangle (253.87, 56.15);

\path[fill=fillColor] (255.21, 38.25) rectangle (267.25, 52.48);

\path[fill=fillColor] (268.58, 38.25) rectangle (280.62, 57.41);

\path[fill=fillColor] (281.96, 38.25) rectangle (294.00, 52.62);

\path[fill=fillColor] (295.34, 38.25) rectangle (307.37, 57.48);

\path[fill=fillColor] (308.71, 38.25) rectangle (320.75, 56.36);

\path[fill=fillColor] (322.09, 38.25) rectangle (334.13, 52.85);

\path[fill=fillColor] (335.47, 38.25) rectangle (347.50, 55.26);

\path[fill=fillColor] (348.84, 38.25) rectangle (360.88, 57.15);

\path[fill=fillColor] (362.22, 38.25) rectangle (374.26, 53.50);

\path[fill=fillColor] (375.59, 38.25) rectangle (387.63, 51.99);

\path[fill=fillColor] (388.97, 38.25) rectangle (401.01, 53.37);

\path[fill=fillColor] (402.35, 38.25) rectangle (414.39, 50.39);

\path[fill=fillColor] (415.72, 38.25) rectangle (427.76, 49.69);

\path[fill=fillColor] (429.10, 38.25) rectangle (441.14, 53.88);

\path[fill=fillColor] (442.48, 38.25) rectangle (454.51, 52.37);

\path[fill=fillColor] (455.85, 38.25) rectangle (467.89, 48.48);

\path[fill=fillColor] (469.23, 38.25) rectangle (481.27, 47.75);

\path[fill=fillColor] (482.61, 38.25) rectangle (494.64, 50.87);

\path[fill=fillColor] (495.98, 38.25) rectangle (508.02, 50.88);

\path[fill=fillColor] (509.36, 38.25) rectangle (521.40, 51.68);

\path[fill=fillColor] (522.73, 38.25) rectangle (534.77, 50.32);

\path[fill=fillColor] (536.11, 38.25) rectangle (548.15, 48.48);
\definecolor{drawColor}{RGB}{0,0,0}

\path[draw=drawColor,line width= 0.6pt,dash pattern=on 4pt off 4pt ,line join=round] ( 43.93, 41.42) -- (572.16, 41.42);
\definecolor{drawColor}{gray}{0.50}

\path[draw=drawColor,line width= 0.6pt,line join=round,line cap=round] ( 43.93, 33.48) rectangle (572.16,138.54);
\end{scope}
\begin{scope}
\path[clip] (  0.00,  0.00) rectangle (578.16,144.54);
\definecolor{drawColor}{RGB}{0,0,0}

\node[text=drawColor,anchor=base east,inner sep=0pt, outer sep=0pt, scale=  0.96*2] at ( 38.53, 34.95) {0.00};

\node[text=drawColor,anchor=base east,inner sep=0pt, outer sep=0pt, scale=  0.96*2] at ( 38.53, 58.82) {0.25};

\node[text=drawColor,anchor=base east,inner sep=0pt, outer sep=0pt, scale=  0.96*2] at ( 38.53, 82.70) {0.50};

\node[text=drawColor,anchor=base east,inner sep=0pt, outer sep=0pt, scale=  0.96*2] at ( 38.53,106.58) {0.75};

\node[text=drawColor,anchor=base east,inner sep=0pt, outer sep=0pt, scale=  0.96*2] at ( 38.53,130.46) {1.00};
\end{scope}
\begin{scope}
\path[clip] (  0.00,  0.00) rectangle (578.16,144.54);
\definecolor{drawColor}{RGB}{0,0,0}

\path[draw=drawColor,line width= 0.6pt,line join=round] ( 40.93, 38.25) --
	( 43.93, 38.25);

\path[draw=drawColor,line width= 0.6pt,line join=round] ( 40.93, 62.13) --
	( 43.93, 62.13);

\path[draw=drawColor,line width= 0.6pt,line join=round] ( 40.93, 86.01) --
	( 43.93, 86.01);

\path[draw=drawColor,line width= 0.6pt,line join=round] ( 40.93,109.89) --
	( 43.93,109.89);

\path[draw=drawColor,line width= 0.6pt,line join=round] ( 40.93,133.76) --
	( 43.93,133.76);
\end{scope}
\begin{scope}
\path[clip] (  0.00,  0.00) rectangle (578.16,144.54);
\definecolor{drawColor}{RGB}{0,0,0}

\path[draw=drawColor,line width= 0.6pt,line join=round] ( 73.96, 30.48) --
	( 73.96, 33.48);

\path[draw=drawColor,line width= 0.6pt,line join=round] (207.72, 30.48) --
	(207.72, 33.48);

\path[draw=drawColor,line width= 0.6pt,line join=round] (341.48, 30.48) --
	(341.48, 33.48);

\path[draw=drawColor,line width= 0.6pt,line join=round] (475.25, 30.48) --
	(475.25, 33.48);
\end{scope}
\begin{scope}
\path[clip] (  0.00,  0.00) rectangle (578.16,144.54);
\definecolor{drawColor}{RGB}{0,0,0}

\node[text=drawColor,anchor=base,inner sep=0pt, outer sep=0pt, scale=  0.96*2] at ( 73.96, 21.46-5) {0};

\node[text=drawColor,anchor=base,inner sep=0pt, outer sep=0pt, scale=  0.96*2] at (207.72, 21.46-5) {10};

\node[text=drawColor,anchor=base,inner sep=0pt, outer sep=0pt, scale=  0.96*2] at (341.48, 21.46-5) {20};

\node[text=drawColor,anchor=base,inner sep=0pt, outer sep=0pt, scale=  0.96*2] at (475.25, 21.46-5) {30};
\end{scope}
\begin{scope}
\path[clip] (  0.00,  0.00) rectangle (578.16,144.54);
\definecolor{drawColor}{RGB}{0,0,0}

\node[text=drawColor,anchor=base,inner sep=0pt, outer sep=0pt, scale=  0.96*2] at (308.04,  8.40-8.4) {lag};
\end{scope}
\end{tikzpicture}}
\label{Fig:oil_log_returns_sq_acf}}\\
	\subfloat[][log-returns]{
			\resizebox{0.48\textwidth}{!}{\input{tikz/oil_log_returns.tex}}
\label{Fig:oil_log_returns}}
	\subfloat[][residuals]{
			\resizebox{0.48\textwidth}{!}{\input{tikz/oil_residuals.tex}}
\label{Fig:oil_residuals}}\\
	\subfloat[][centred log-returns]{
			\resizebox{0.48\textwidth}{!}{\input{tikz/oil_centered_log_returns.tex}}
\label{Fig:oil_centered_log_returns}}
	\subfloat[][autocorrelation of the squared residuals]{
			\resizebox{0.48	\textwidth}{!}{\begin{tikzpicture}[x=1pt,y=1pt]
\definecolor{fillColor}{RGB}{255,255,255}
\path[use as bounding box,fill=fillColor,fill opacity=0.00] (0,0) rectangle (578.16,144.54);
\begin{scope}
\path[clip] (  0.00,  0.00) rectangle (578.16,144.54);
\definecolor{drawColor}{RGB}{255,255,255}
\definecolor{fillColor}{RGB}{255,255,255}

\path[draw=drawColor,line width= 0.6pt,line join=round,line cap=round,fill=fillColor] (  0.00, -0.00) rectangle (578.16,144.54);
\end{scope}
\begin{scope}
\path[clip] ( 43.93, 33.48) rectangle (572.16,138.54);
\definecolor{fillColor}{RGB}{255,255,255}

\path[fill=fillColor] ( 43.93, 33.48) rectangle (572.16,138.54);
\definecolor{drawColor}{gray}{0.98}

\path[draw=drawColor,line width= 0.6pt,line join=round] ( 43.93, 51.22) --
	(572.16, 51.22);

\path[draw=drawColor,line width= 0.6pt,line join=round] ( 43.93, 74.81) --
	(572.16, 74.81);

\path[draw=drawColor,line width= 0.6pt,line join=round] ( 43.93, 98.39) --
	(572.16, 98.39);

\path[draw=drawColor,line width= 0.6pt,line join=round] ( 43.93,121.97) --
	(572.16,121.97);

\path[draw=drawColor,line width= 0.6pt,line join=round] (140.84, 33.48) --
	(140.84,138.54);

\path[draw=drawColor,line width= 0.6pt,line join=round] (274.60, 33.48) --
	(274.60,138.54);

\path[draw=drawColor,line width= 0.6pt,line join=round] (408.37, 33.48) --
	(408.37,138.54);

\path[draw=drawColor,line width= 0.6pt,line join=round] (542.13, 33.48) --
	(542.13,138.54);
\definecolor{drawColor}{gray}{0.90}

\path[draw=drawColor,line width= 0.2pt,line join=round] ( 43.93, 39.43) --
	(572.16, 39.43);

\path[draw=drawColor,line width= 0.2pt,line join=round] ( 43.93, 63.01) --
	(572.16, 63.01);

\path[draw=drawColor,line width= 0.2pt,line join=round] ( 43.93, 86.60) --
	(572.16, 86.60);

\path[draw=drawColor,line width= 0.2pt,line join=round] ( 43.93,110.18) --
	(572.16,110.18);

\path[draw=drawColor,line width= 0.2pt,line join=round] ( 43.93,133.76) --
	(572.16,133.76);

\path[draw=drawColor,line width= 0.2pt,line join=round] ( 73.96, 33.48) --
	( 73.96,138.54);

\path[draw=drawColor,line width= 0.2pt,line join=round] (207.72, 33.48) --
	(207.72,138.54);

\path[draw=drawColor,line width= 0.2pt,line join=round] (341.48, 33.48) --
	(341.48,138.54);

\path[draw=drawColor,line width= 0.2pt,line join=round] (475.25, 33.48) --
	(475.25,138.54);
\definecolor{fillColor}{gray}{0.35}

\path[fill=fillColor] ( 67.94, 39.43) rectangle ( 79.98,133.76);

\path[fill=fillColor] ( 81.31, 39.43) rectangle ( 93.35, 56.05);

\path[fill=fillColor] ( 94.69, 39.43) rectangle (106.73, 46.10);

\path[fill=fillColor] (108.07, 39.43) rectangle (120.11, 44.82);

\path[fill=fillColor] (121.44, 39.43) rectangle (133.48, 44.04);

\path[fill=fillColor] (134.82, 39.43) rectangle (146.86, 41.80);

\path[fill=fillColor] (148.20, 39.43) rectangle (160.23, 42.09);

\path[fill=fillColor] (161.57, 39.43) rectangle (173.61, 42.16);

\path[fill=fillColor] (174.95, 39.43) rectangle (186.99, 44.15);

\path[fill=fillColor] (188.33, 39.43) rectangle (200.36, 45.73);

\path[fill=fillColor] (201.70, 39.43) rectangle (213.74, 39.68);

\path[fill=fillColor] (215.08, 39.43) rectangle (227.12, 43.07);

\path[fill=fillColor] (228.45, 39.43) rectangle (240.49, 39.73);

\path[fill=fillColor] (241.83, 39.43) rectangle (253.87, 41.02);

\path[fill=fillColor] (255.21, 39.43) rectangle (267.25, 43.54);

\path[fill=fillColor] (268.58, 39.43) rectangle (280.62, 41.20);

\path[fill=fillColor] (281.96, 39.43) rectangle (294.00, 39.51);

\path[fill=fillColor] (295.34, 39.43) rectangle (307.37, 43.21);

\path[fill=fillColor] (308.71, 39.43) rectangle (320.75, 40.22);

\path[fill=fillColor] (322.09, 39.43) rectangle (334.13, 41.88);

\path[fill=fillColor] (335.47, 39.43) rectangle (347.50, 41.54);

\path[fill=fillColor] (348.84, 39.43) rectangle (360.88, 42.21);

\path[fill=fillColor] (362.22, 39.43) rectangle (374.26, 42.84);

\path[fill=fillColor] (375.59, 39.43) rectangle (387.63, 41.42);

\path[fill=fillColor] (388.97, 39.43) rectangle (401.01, 41.29);

\path[fill=fillColor] (402.35, 39.43) rectangle (414.39, 41.39);

\path[fill=fillColor] (415.72, 39.32) rectangle (427.76, 39.43);

\path[fill=fillColor] (429.10, 39.43) rectangle (441.14, 40.93);

\path[fill=fillColor] (442.48, 39.43) rectangle (454.51, 41.83);

\path[fill=fillColor] (455.85, 39.43) rectangle (467.89, 42.58);

\path[fill=fillColor] (469.23, 39.43) rectangle (481.27, 40.73);

\path[fill=fillColor] (482.61, 39.43) rectangle (494.64, 39.84);

\path[fill=fillColor] (495.98, 38.62) rectangle (508.02, 39.43);

\path[fill=fillColor] (509.36, 38.25) rectangle (521.40, 39.43);

\path[fill=fillColor] (522.73, 39.43) rectangle (534.77, 41.12);

\path[fill=fillColor] (536.11, 39.43) rectangle (548.15, 42.75);
\definecolor{drawColor}{RGB}{0,0,0}

\path[draw=drawColor,line width= 0.6pt,dash pattern=on 4pt off 4pt ,line join=round] ( 43.93, 42.56) -- (572.16, 42.56);
\definecolor{drawColor}{gray}{0.50}

\path[draw=drawColor,line width= 0.6pt,line join=round,line cap=round] ( 43.93, 33.48) rectangle (572.16,138.54);
\end{scope}
\begin{scope}
\path[clip] (  0.00,  0.00) rectangle (578.16,144.54);
\definecolor{drawColor}{RGB}{0,0,0}

\node[text=drawColor,anchor=base east,inner sep=0pt, outer sep=0pt, scale=  0.96*2] at ( 38.53, 36.13) {0.00};

\node[text=drawColor,anchor=base east,inner sep=0pt, outer sep=0pt, scale=  0.96*2] at ( 38.53, 59.71) {0.25};

\node[text=drawColor,anchor=base east,inner sep=0pt, outer sep=0pt, scale=  0.96*2] at ( 38.53, 83.29) {0.50};

\node[text=drawColor,anchor=base east,inner sep=0pt, outer sep=0pt, scale=  0.96*2] at ( 38.53,106.88) {0.75};

\node[text=drawColor,anchor=base east,inner sep=0pt, outer sep=0pt, scale=  0.96*2] at ( 38.53,130.46) {1.00};
\end{scope}
\begin{scope}
\path[clip] (  0.00,  0.00) rectangle (578.16,144.54);
\definecolor{drawColor}{RGB}{0,0,0}

\path[draw=drawColor,line width= 0.6pt,line join=round] ( 40.93, 39.43) --
	( 43.93, 39.43);

\path[draw=drawColor,line width= 0.6pt,line join=round] ( 40.93, 63.01) --
	( 43.93, 63.01);

\path[draw=drawColor,line width= 0.6pt,line join=round] ( 40.93, 86.60) --
	( 43.93, 86.60);

\path[draw=drawColor,line width= 0.6pt,line join=round] ( 40.93,110.18) --
	( 43.93,110.18);

\path[draw=drawColor,line width= 0.6pt,line join=round] ( 40.93,133.76) --
	( 43.93,133.76);
\end{scope}
\begin{scope}
\path[clip] (  0.00,  0.00) rectangle (578.16,144.54);
\definecolor{drawColor}{RGB}{0,0,0}

\path[draw=drawColor,line width= 0.6pt,line join=round] ( 73.96, 30.48) --
	( 73.96, 33.48);

\path[draw=drawColor,line width= 0.6pt,line join=round] (207.72, 30.48) --
	(207.72, 33.48);

\path[draw=drawColor,line width= 0.6pt,line join=round] (341.48, 30.48) --
	(341.48, 33.48);

\path[draw=drawColor,line width= 0.6pt,line join=round] (475.25, 30.48) --
	(475.25, 33.48);
\end{scope}
\begin{scope}
\path[clip] (  0.00,  0.00) rectangle (578.16,144.54);
\definecolor{drawColor}{RGB}{0,0,0}

\node[text=drawColor,anchor=base,inner sep=0pt, outer sep=0pt, scale=  0.96*2] at ( 73.96, 21.46-5) {0};

\node[text=drawColor,anchor=base,inner sep=0pt, outer sep=0pt, scale=  0.96*2] at (207.72, 21.46-5) {10};

\node[text=drawColor,anchor=base,inner sep=0pt, outer sep=0pt, scale=  0.96*2] at (341.48, 21.46-5) {20};

\node[text=drawColor,anchor=base,inner sep=0pt, outer sep=0pt, scale=  0.96*2] at (475.25, 21.46-5) {30};
\end{scope}
\begin{scope}
\path[clip] (  0.00,  0.00) rectangle (578.16,144.54);
\definecolor{drawColor}{RGB}{0,0,0}

\node[text=drawColor,anchor=base,inner sep=0pt, outer sep=0pt, scale=  0.96*2] at (308.04,  8.40-8.4) {lag};
\end{scope}
\end{tikzpicture}}
\label{Fig:oil_residuals_sq_acf}}\\
	\caption{Change-point analysis on the daily OPEC Reference Basket oil price in USD from  1 January, 2003 to 15 July, 2016. Figure~\ref{Fig:oil_price}: price series $P_{t}$ (thin grey), locations of the change-points detected with NOT (vertical dotted lines) and NMCD (vertical dashed lines). Figure~\ref{Fig:oil_log_returns_sq_acf}: autocorrelation function of $Y_{t}^2$. Figure~\ref{Fig:oil_log_returns}:  log-returns $Y_{t}=100 \log\left(P_{t}/P_{t-1}\right)$ (thin grey), the fitted piecewise-constant mean via NOT, $\hat{f}_{t}$  (thick red). Figure~\ref{Fig:oil_residuals}: estimated residuals via NOT, $\hat{\varepsilon}_{t}=(Y_t-\hat{f}_{t})/\hat{\sigma}_{t}$. Figure~\ref{Fig:oil_centered_log_returns}: the centred log-returns $|Y_t-\hat{f}_{t}|$ (thin grey), fitted piecewise-constant volatility $\hat{\sigma}_{t}$ (thick red). Figure~\ref{Fig:oil_residuals_sq_acf}: autocorrelation of $\hat{\varepsilon}_{t}^2$. The exact locations of the change-points detected via NOT are given in Table~\ref{Table:oil_price_changepoints}.\label{Fig:oil_analysis}}
\end{figure} 

\begin{table}
\caption{\label{Table:oil_price_changepoints}Change-points detected using NOT and NMCD methods in the daily OPEC Reference Basket oil price data from 1 January 2003 to 15 July 2016, with some of them dated.} 

\centering
\footnotesize
\fbox{
\begin{tabular}{p{0.18\textwidth}|p{0.17\textwidth}|p{0.52\textwidth}}
	 NOT  &  NMCD &  Event that coincides\\
	\hline
	 29 April 2003  & N/A & Invasion of Iraq\\
	 1 September 2008 & 28 August 2008 & critical stage of the subprime mortgage crisis\\
	 27 January 2009 & 22 January 2009 & tensions in the Gaza Strip\\
	 1 October 2009 & 23 October 2009 & \\
	 12 November 2012 & 12 October 2012 & beginning of a period of low volatility\\
	 30 September 2014 & 1 October 2014 & \\
	 5 January 2016 & 21 January 2016 & beginning of a sell-off leading the price to 12-year low\\
	 N/A & 22 February 2016& \\
\end{tabular}}
\end{table}

\section{Proofs}
\subsection{Some useful lemmas}
\label{Sec:lemmas}
\subsubsection{The piecewise-constant case}
\begin{Lemma}
	\label{Lem:basic}
	Let $g(x,y) = \frac{xy}{x+y}$ and suppose that $\min (x,y) > 0$. Then 
	\begin{align*}
	g(x,y)\geq \frac{1}{2}\min(x,y).
	\end{align*}
\end{Lemma}
\begin{proof}
	Without loss of generality, assume that $x \ge y$. Then $g(x,y) \ge \frac{xy}{2x} \ge y/2 = \min(x,y)/2$.
\end{proof}
\begin{Lemma}
	\label{Lem:cusum_at_cp_jump}
	Suppose $\fb=(f_{1},\ldots,f_{T})'$ is piecewise-constant vector as in Scenario \ref{Scen:change_in_mean}, and $\tau_{1},\ldots,\tau_{q}$ are the locations of the change-points. Suppose $1\leq s < e\leq T$, such that $\tau_{j-1} <  s \leq \tau_{j} <  e \leq \tau_{j+1}$  for some $j=1\ldots, q$. Let $\eta = \min\{\tau_{j} - s +1, e - \tau_{j}\}$ and $\Delta_j^\fb = |f_{\tau_{j}+1}-f_{\tau_{j}}|$. Then
	\begin{align*}
	\cont{s}{e}{\tau_{j}}{\fb} = \max_{s\leq b <  e}\cont{s}{e}{b}{\fb} \; \begin{cases} \geq \frac{1}{\sqrt{2}} \eta^{1/2} \Delta_j^\fb,\\  \leq \eta^{1/2} \Delta_j^\fb. \end{cases}
	\end{align*}
\end{Lemma}
\begin{proof}
	For any $s\leq b < e$, by simple algebra, we have
	\begin{align}
	\cont{s}{e}{b}{\fb} =
	\begin{cases}
	\sqrt{\frac{b-s+1}{l(e-b)}}(e-\tau_{j}) |f_{\tau_{j}+1}-f_{\tau_{j}}|, &b \le  \tau_{j};\\
	\sqrt{\frac{(\tau_{j}-s+1)(e-\tau_{j})}{l}} |f_{\tau_{j}+1}-f_{\tau_{j}}|, &b=\tau_{j};\\
	\sqrt{\frac{e-b}{l(b-s+1)}}(\tau_{j}-s+1) |f_{\tau_{j}+1}-f_{\tau_{j}}|, &b\ge \tau_{j},
	\end{cases}\label{Eq:cusum_one_cpt}
	\end{align}
	where $l = s - e + 1$.
	Now $\cont{s}{e}{\tau_{j}}{\fb}=\max_{s\leq b \le  e}\cont{s}{e}{b}{\fb}$ follows from the fact that $\cont{s}{e}{b}{\fb}$ is increasing (as a function of $b$) for $1 \leq b \le \tau_j$ and decreasing for $\tau_j \le b \le  e$. To prove the lower bound, we set $\eta_L=\tau_{j}-s+1$ and $\eta_R=e-\tau_{j}$ and observe that $\eta_L \geq \eta$ and $\eta_R\geq \eta$. Therefore by  Lemma~\ref{Lem:basic}, $\frac{\eta_L \eta_R}{\eta_L+\eta_R} \geq \frac{\eta}{2} $. Noting that $l=\eta_L+\eta_R$ we bound
	\begin{align*}
	\cont{s}{e}{\tau_j}{\fb}  = \sqrt{\frac{(\tau_{j}-s+1)(e-\tau_{j})}{l}} |f_{\tau_{j}+1}-f_{\tau_{j}}| \begin{cases} \geq (\eta/2)^{1/2} \Delta_j^\fb;\\
	\leq \eta^{1/2}\Delta_j^\fb. \end{cases}
	\end{align*}
	which completes the proof.
\end{proof}

\begin{Lemma}
	\label{Lem:cusum_two_cp_jump}
	Suppose $\fb=(f_{1},\ldots,f_{T})'$ is piecewise-constant vector as in Scenario \ref{Scen:change_in_mean}, and $\tau_{1},\ldots,\tau_{q}$ are the locations of the change-points. Suppose $1\leq s < e\leq T$ such that $\tau_{j-1} < s  \leq \tau_{j}$ and $\tau_{j+1} <  e \leq \tau_{j+2}$  for some $j=1\ldots, q-1$. Then
	\begin{align*}
	\max_{s\leq b < e}\cont{s}{e}{b}{\fb}  \le  (\tau_j - s + 1)^{1/2} \Delta_j^\fb + (e - \tau_{j+1})^{1/2} \Delta_{j+1}^{\fb}
	\end{align*}
	where $\Delta_j^\fb = |f_{\tau_{j}+1}-f_{\tau_{j}}|$.
\end{Lemma}
\begin{proof}
	Suppose that $b^* = \argmax_{s\leq b < e}\cont{s}{e}{b}{\fb}$. Then
	\begin{align*}
	0 &\le \|	\fb -  \langle\fb, \psib_{s,e}^{b^*}\rangle \psib_{s,e}^{b^*} - \langle \fb,  \mathbf{1}_{s,e} \rangle  \mathbf{1}_{s,e} \|^2	 =  \| \fb -  \langle \fb,  \mathbf{1}_{s,e} \rangle  \mathbf{1}_{s,e} \|^2 -  \langle\fb, \psib_{s,e}^{b^*}\rangle ^2 \\
	& \le \|\fb -  f_{\tau_j+1} \sqrt{e-s+1} \mathbf{1}_{s,e}   \|^2 -  \langle\fb, \psib_{s,e}^{b^*}\rangle ^2 \\
	& = (\tau_j - s + 1) (\Delta_j^\fb)^2 + (e - \tau_{j+1}) (\Delta_{j+1}^\fb)^2 -  \Big(\max_{s\leq b < e}\cont{s}{e}{b}{\fb}\Big)^2.
	\end{align*}
	It then follows that
	\[
	\max_{s\leq b < e}\cont{s}{e}{b}{\fb} \le \sqrt{(\tau_j - s + 1) (\Delta_j^\fb)^2 + (e - \tau_{j+1}) (\Delta_{j+1}^\fb)^2} \le  (\tau_j - s + 1)^{1/2} \Delta_j^\fb + (e - \tau_{j+1})^{1/2} \Delta_{j+1}^{\fb}.
	\]
\end{proof}	

\begin{Lemma}
	\label{Lem:cusum_cp_jump_distance}
	Suppose $\fb=(f_{1},\ldots,f_{T})'$ is piecewise-constant vector as in Scenario \ref{Scen:change_in_mean}. Pick any interval $[s,e]\subset [1,T]$ such that $[s,e-1]$ contains exactly one change-point $\tau_{j}$. Let $\rho = |\tau_{j} -b|$, $\Delta_j^\fb = |f_{\tau_{j}+1}-f_{\tau_{j}}|$, $\eta_L = \tau_{j}-s+1$ and $\eta_R = e-\tau_{j}$. Then,
	\[
	\| \psib_{s,e}^b \langle\fb,\psib_{s,e}^b \rangle - \psib_{s,e}^{\tau_j} \langle\fb,\psib_{s,e}^{\tau_j}\rangle\|_2^2 = (\cont{s}{e}{\tau_j}{\fb})^2-(\cont{s}{e}{b}{\fb})^2.
	\]
	Moreover, 
	\begin{enumerate}
		\item for any $\tau_j \leq b < e$, $(\cont{s}{e}{\tau_j}{\fb})^2-(\cont{s}{e}{b}{\fb})^2 = \frac{ \rho \; \eta_L}{\rho+\eta_L} (\Delta_j^\fb)^2$;
		\item for any $s \leq b < \tau_j$, $(\cont{s}{e}{\tau_j}{\fb})^2-(\cont{s}{e}{b}{\fb})^2 = \frac{ \rho \; \eta_R}{\rho+\eta_R} (\Delta_j^\fb)^2$.
	\end{enumerate}
\end{Lemma}
\begin{proof}
	First, we note that since there is only one change-point in $[s,e-1]$, the restriction of $\fb$ on $[s,e]$, i.e. $\fb|_{[s,e]} = (0,\ldots,0,f_{s},\ldots,f_{e},0,\ldots,0)'$ can be decomposed into
	\[
	\fb|_{[s,e]} = \psib_{s,e}^{\tau_j} \langle\fb, \psib_{s,e}^{\tau_j}\rangle +  \mathbf{1}_{s,e} \langle\fb, \mathbf{1}_{s,e}\rangle,
	\]
	where we also used the fact that $\psib_{s,e}^{\tau_j}$ and $\mathbf{1}_{s,e}$ are orthonormal. Note that $\psib_{s,e}^{b}$ and $\mathbf{1}_{s,e}$ are also orthonormal, it follows that
	\[
	\langle\fb,\psib_{s,e}^b \rangle = \langle\fb|_{[s,e]},\psib_{s,e}^b \rangle = \big\langle \psib_{s,e}^{\tau_j}\langle\fb, \psib_{s,e}^{\tau_j}\rangle +  \mathbf{1}_{s,e} \langle\fb, \mathbf{1}_{s,e}\rangle  ,\psib_{s,e}^b \big\rangle = \langle \psib_{s,e}^{\tau_j}, \psib_{s,e}^{b}\rangle\langle\fb, \psib_{s,e}^{\tau_j}\rangle.
	\]
	Therefore,
	\[
	\langle\fb,\psib_{s,e}^b \rangle ^2 =  \langle\fb,\psib_{s,e}^b \rangle \langle \psib_{s,e}^{\tau_j}, \psib_{s,e}^{b}\rangle\langle\fb, \psib_{s,e}^{\tau_j}\rangle,
	\]
	and thus
	\begin{align*}
	\langle\fb,\psib_{s,e}^{\tau_j} \rangle ^2 - \langle\fb,\psib_{s,e}^b \rangle ^2 &= \langle\fb,\psib_{s,e}^{\tau_j} \rangle ^2 + \langle\fb,\psib_{s,e}^b \rangle ^2 - 2\langle\fb,\psib_{s,e}^b \rangle \langle \psib_{s,e}^{\tau_j}, \psib_{s,e}^{b}\rangle\langle\fb, \psib_{s,e}^{\tau_j}\rangle \\
	&= \|  \psib_{s,e}^b \langle\fb, \psib_{s,e}^b \rangle - \psib_{s,e}^{\tau_j} \langle\fb, \psib_{s,e}^{\tau_j}\rangle\|_2^2.
	\end{align*}
	Here in the above final step, we used the fact that $\|\psib_{s,e}^{\tau_j}\|^2_2 = \|\psib_{s,e}^{b}\|^2_2 = 1$.
	
	Second, for the sake of brevity, we only prove the case of $b\ge \tau_{j}$. Let $l = e -s +1$, $x=b-s+1$, and thus $\rho = x-\eta_L$.
	Using \eqref{Eq:cusum_one_cpt}, we get
	\begin{align*}
	(\cont{s}{e}{\tau_j}{\fb})^2-(\cont{s}{e}{b}{\fb})^2&= \left(\frac{\eta_L(l-\eta_L)}{l} -\frac{\eta_L^2(l-x)}{lx} \right) |f_{\tau_{j}+1}-f_{\tau_{j}}|^2 \\
	&= \frac{\eta_L(x-\eta_L)}{x}(\Delta_j^\fb)^2 =  \left(\frac{\rho \eta_L }{\eta_L + \rho}\right)(\Delta_j^\fb)^2.
	\end{align*}
\end{proof}

\subsubsection{The piecewise-linear continuous case}
\begin{Lemma}
	\label{Lem:cusum_at_cp_kink}
	Suppose $\fb=(f_{1},\ldots,f_{T})'$ is piecewise-linear vector as in Scenario \ref{Scen:change_in_slope}, and $\tau_{1},\ldots,\tau_{q}$ are the locations of the change-points. Suppose $1\leq s < e\leq T$, such that $\tau_{j-1} \le  s < \tau_{j} <  e \leq \tau_{j+1}$  for some $j=1\ldots, q$. Let $\eta = \min\{\tau_{j} - s, e - \tau_{j}\}$ and $\Delta_j^\fb = |2f_{\tau_{j}}-f_{\tau_{j}-1}-f_{\tau_{j}+1}|$. Then
	\begin{align*}
	\cont{s}{e}{\tau_{j}}{\fb} = \max_{s < b <  e}\cont{s}{e}{b}{\fb} \; \begin{cases} \geq \frac{1}{\sqrt{24}}\eta^{3/2}\Delta_j^\fb,\\  \leq \frac{1}{\sqrt{3}}(\eta+1)^{3/2} \Delta_j^\fb. \end{cases}
	\end{align*}
\end{Lemma}
\begin{proof} 
	First, we show that $\cont{s}{e}{b}{\fb}$ is maximised at $b = \tau_j$. Using the notation from the proof of Lemma~\ref{Lem:cusum_cp_jump_distance}, we have that
	\[
	\fb|_{[s,e]} = \phib_{s,e}^{\tau_j} \langle \fb, \phib_{s,e}^{\tau_j} \rangle +   \gammab_{s,e} \langle \fb, \mathbf{1}_{s,e} \rangle + \mathbf{1}_{s,e} \langle \fb, \mathbf{1}_{s,e} \rangle.
	\]
	Therefore, it follows that
	\begin{align}
	\label{Eq:cusum_at_cp_kink_decomp1}	
	\| 	\fb|_{[s,e]}  \|_2^2 = \langle\fb, \phib_{s,e}^{\tau_j}\rangle^2 + \langle\fb, \gammab_{s,e}\rangle^2 + \langle\fb, \mathbf{1}_{s,e}\rangle^2.
	\end{align}
	For any $b \in \{s+1,\ldots, \tau_j - 1, \tau_j + 1, \ldots, e-1\}$, it is clear that $\fb|_{[s,e]}$ does not lie in the span of $\phib_{s,e}^b$, $\gammab_{s,e}$ and $\mathbf{1}_{s,e}$. Consequently, by projecting  $\fb|_{[s,e]}$ onto these three bases, we have that 
	\begin{align}
	\label{Eq:cusum_at_cp_kink_decomp2}	
	\| 	\fb|_{[s,e]}  \|^2 > \langle\fb, \phib_{s,e}^{b}\rangle^2 + \langle\fb, \gammab_{s,e}\rangle^2 + \langle\fb, \mathbf{1}_{s,e}\rangle^2.
	\end{align}
	Comparing (\ref{Eq:cusum_at_cp_kink_decomp2}) with (\ref{Eq:cusum_at_cp_kink_decomp1}) entails that $| \langle\fb, \phib_{s,e}^{\tau_j}\rangle \big | > \big | \langle\fb, \phib_{s,e}^{b}\rangle \big |$ for any $b \neq \tau_j$.
	
	Secondly, set $\eta_{L}=\tau_{j}-s$ and $\eta_{R}=e-\tau_{j}$. After some calculation, we get that	
	\begin{align*}
	\cont{s}{e}{\tau_j}{\fb} = \left\{\frac{\eta_L(\eta_{L}+1)\eta_{R}(\eta_{R}+1)(2\eta_{L}\eta_{R}+\eta_{L}+\eta_{R}+2)}{6l(l^2-1)}\right\}\Delta_j^\fb,
	\end{align*}
	where $l = e -s +1$. Also, we have $\eta_L \geq \eta$, $\eta_R\geq \eta$ and $l=\eta_L+\eta_R + 1$. To prove the lower bound, we observe that
	\begin{align*}
	&\left\{\frac{\eta_L(\eta_{L}+1)\eta_{R}(\eta_{R}+1)(2\eta_{L}\eta_{R}+\eta_{L}+\eta_{R}+2)}{6l(l^2-1)}\right\} \\
	&\geq \left\{\frac{1}{6}\frac{(\eta_{L}+1)\eta_{R}}{l}\frac{\eta_{L}(\eta_{R}+1)}{l}\frac{2 \min(\eta_{L},\eta_{R})\{\max(\eta_{L},\eta_{R})+1\}}{l}\right\} 
	\geq \left\{\frac{\eta^3}{24}\right\},
	\end{align*}
	where the last inequality is obtained applying Lemma~\ref{Lem:basic} three times. For the upper bound, we notice that 
	$2\eta_{L}\eta_{R}+\eta_{L}+\eta_{R}+2 \leq 2 (\eta_{L}+1)(\eta_{R}+1)$ which implies
	\begin{align*}
	\left\{\frac{\eta_L(\eta_{L}+1)\eta_{R}(\eta_{R}+1)(2\eta_{L}\eta_{R}+\eta_{L}+\eta_{R}+2)}{6l(l^2-1)}\right\} 
	\leq \left\{\frac{1}{3}\frac{\eta_{L}\eta_{R}(\eta_{L}+1)^2(\eta_{R}+1)^2}{(l-1)l^2}\right\}
	\leq \left\{\frac{(\eta+1)^3}{3}\right\}.
	\end{align*}
	
\end{proof}

\begin{Lemma}
	\label{Lem:cusum_two_cp_kink}
	Suppose $\fb=(f_{1},\ldots,f_{T})'$ is piecewise-linear vector as in Scenario \ref{Scen:change_in_slope}, and $\tau_{1},\ldots,\tau_{q}$ are the locations of the change-points. Suppose $1\leq s < e\leq T$ such that $\tau_{j-1} \leq s  \leq \tau_{j}$ and $\tau_{j+1} \le  e \leq \tau_{j+2}$  for some $j=1\ldots, q-1$. Then
	\begin{align*}
	\max_{s\leq b < e}\cont{s}{e}{b}{\fb}  \le  \frac{1}{\sqrt{3}} (\tau_j - s + 1)^{3/2} \Delta_j^\fb + \frac{1}{\sqrt{3}} (e - \tau_{j+1} + 1)^{3/2} \Delta_{j+1}^{\fb},
	\end{align*}
	where $\Delta_j^\fb = |2f_{\tau_{j}}-f_{\tau_{j}-1}-f_{\tau_{j}+1}|$.
\end{Lemma}
\begin{proof}
	Suppose that $b^* = \argmax_{s\leq b \leq e}\cont{s}{e}{b}{\fb}$. Then
	\begin{align*}
	0 &\le \|	\fb|_{[s,e]} -  \langle\fb, \phib_{s,e}^{b^*}\rangle \phib_{s,e}^{b^*} - \langle \fb,  \gammab_{s,e}\rangle  \gammab_{s,e} - \langle \fb,  \mathbf{1}_{s,e} \rangle  \mathbf{1}_{s,e} \|^2	 =  \| \fb|_{[s,e]} -  \langle \fb,  \gammab_{s,e}\rangle  \gammab_{s,e} - \langle \fb,  \mathbf{1}_{s,e} \rangle  \mathbf{1}_{s,e} \|^2 -  \langle\fb, \phib_{s,e}^{b^*}\rangle ^2 \\
	& = \frac{1}{6}(\tau_j - s)(\tau_j - s+1)(2\tau_j - 2s+1) (\Delta_j^\fb)^2 +  \frac{1}{6}(e - \tau_{j+1})(e - \tau_{j+1}+1)(2e - 2\tau_{j+1}+1) (\Delta_{j+1}^\fb)^2 \\
	&\qquad -  \Big(\max_{s\leq b < e}\cont{s}{e}{b}{\fb}\Big)^2.
	\end{align*}
	It then follows that
	\begin{align*}
	\max_{s\leq b < e}\cont{s}{e}{b}{\fb} &\le \left\{(\tau_j - s + 1)^3 (\Delta_j^\fb)^2/3 + (e - \tau_{j+1}+1)^3 (\Delta_{j+1}^\fb)^2/3 \right\} \\
	&  \le  \frac{1}{\sqrt{3}} (\tau_j - s + 1)^{3/2} \Delta_j^\fb + \frac{1}{\sqrt{3}} (e - \tau_{j+1} + 1)^{3/2} \Delta_{j+1}^{\fb}.
	\end{align*}
\end{proof}	

\begin{Lemma}
	\label{Lem:cusum_cp_kink_distance}
	Suppose $\fb=(f_{1},\ldots,f_{T})'$ is piecewise-linear vector as in Scenario \ref{Scen:change_in_slope}, and $\tau_{1},\ldots,\tau_{q}$ are the locations of the change-points. Suppose $1\leq s < e\leq T$, such that $\tau_{j-1} \le  s < \tau_{j} <  e \leq \tau_{j+1}$  for some $j=1\ldots, q$.  Let $\rho = |\tau_{j} -b|$, $\eta_L = \tau_{j} - s$, $\eta_R = e - \tau_{j}$ and $\Delta_j^\fb = |2f_{\tau_{j}}-f_{\tau_{j}-1}-f_{\tau_{j}+1}|$. Then,
	\begin{align}
	\label{Eq:cusum_cp_kink_distance_1}
	\| \phib_{s,e}^b \langle\fb,\phib_{s,e}^b \rangle - \phib_{s,e}^{\tau_j} \langle\fb,\phib_{s,e}^{\tau_j}\rangle\|_2^2 = (\cont{s}{e}{\tau_j}{\fb})^2-(\cont{s}{e}{b}{\fb})^2.
	\end{align}
	Moreover, 
	\begin{enumerate}
		\item for any $\tau_j \le b < e$, $(\cont{s}{e}{\tau_j}{\fb})^2-(\cont{s}{e}{b}{\fb})^2 \ge \frac{1}{63} \min(\rho,\eta_L)^3 (\Delta^{\fb}_j)^2$;
		\item for any $s < b \le \tau_j$, $(\cont{s}{e}{\tau_j}{\fb})^2-(\cont{s}{e}{b}{\fb})^2 \ge \frac{1}{63} \min(\rho,\eta_R)^3 (\Delta^{\fb}_j)^2$.
	\end{enumerate}
\end{Lemma}
\begin{proof} 
	The proof of (\ref{Eq:cusum_cp_kink_distance_1}) is very similar to that shown in Lemma~\ref{Lem:cusum_cp_jump_distance}, so is omitted for brevity. In the following, we only deal with the case of $\tau_j \le b < e$. Note that
	\begin{align*}
	&\| \phib_{s,e}^b \langle\fb,\phib_{s,e}^b \rangle - \phib_{s,e}^{\tau_j} \langle\fb,\phib_{s,e}^{\tau_j}\rangle\|_2^2 = \big\| \phib_{s,e}^b \langle\fb,\phib_{s,e}^b \rangle + \gammab_{s,e} \langle\fb,\gammab_{s,e} \rangle + \mathbf{1}_{s,e} \langle \fb, \mathbf{1}_{s,e} \rangle - \fb|_{[s,e]} \big\|_2^2  \\
	&\ge \min_{a_0,a_1 \in \R} \big\|\fb|_{[s,b]} - a_0 \mathbf{1}_{s,b} - a_1 \gammab_{s,b} \big\|_2^2 + \min_{a_0,a_1 \in \R}\big\|\fb|_{[b+1,e]} - a_0 \mathbf{1}_{b+1,e} - a_1 \gammab_{b+1,e} \big\|_2^2 \\
	&\ge \min_{a_0,a_1 \in \R} \big\|\fb|_{[s,b]} - a_0 \mathbf{1}_{s,b} - a_1 \gammab_{s,b} \big\|_2^2.
	\end{align*}
	Recalling the definitions of $\alpha_{s,b}^{\tau_j}$ and $\beta_{s,b}^{\tau_j}$ in (\ref{Eq:basis_function_kink}), and writing $d = b-s+1$. After some calculations (similar to what has already been carried out in deriving $\phib_{s,e}^b$), we obtain that
	\begin{align*}
	&\min_{a_0,a_1 \in \R} \big\|\fb|_{[s,b]} - a_0 \mathbf{1}_{s,b} - a_1 \gammab_{s,b} \big\|_2^2 
	= \Big[(3\eta_L+\rho+2) \alpha_{s,b}^{\tau_j} \beta_{s,b}^{\tau_j} + (3\rho+\eta_L+2)\alpha_{s,b}^{\tau_j} (\beta_{s,b}^{\tau_j})^{-1}\Big]^{-2} (\Delta^{\fb}_j)^2 \\
	&= \frac{1}{6}(\Delta^{\fb}_j)^2 d(d^2-1)\big[1+\rho\eta_L + (\rho+1)(\eta_L+1)\big] \times \\
	&\qquad \bigg[(d+2\eta_L+1)^2\frac{\rho(\rho+1)}{\eta_L(\eta_L+1)}+(d+2\rho+1)^2\frac{\eta_L(\eta_L+1)}{\rho(\rho+1)} + 2(d+2\eta_L+1)(d+2\rho+1)\bigg]^{-1}.
	\end{align*} 
	Notice that the above equation is symmetric with respect to $\eta_L$ and $\rho$. Without loss of generality, here we proceed by assuming that $\eta_L \ge \rho$. Since $(d + 2\eta_L+1) + (d + 2\rho + 1) = 4d$, it follows that $(d + 2\eta_L+1)(d + 2\rho + 1) \le 4d^2$. Therefore,
	\begin{align*}
	&\min_{a_0,a_1 \in \R} \big\|\fb|_{[s,b]} - a_0 \mathbf{1}_{s,b} - a_1 \gammab_{s,b} \big\|_2^2 \\
	& \ge \frac{1}{6}(\Delta^{\fb}_j)^2 d(d^2-1) [2(\eta_L+1)\rho] \bigg[(3d)^2 + (2d)^2 \frac{(\eta_L+1)^2}{\rho^2} + 8 d^2\bigg]^{-1} \\
	& \ge \frac{1}{6}(\Delta^{\fb}_j)^2 d^2(d-1) [2(\eta_L+1) \rho] \bigg[21 d^2 \frac{(\eta_L+1)^2}{\rho^2} \bigg]^{-1} \ge \frac{1}{63}\rho^3 (\Delta^{\fb}_j)^2,
	\end{align*}
	where in the last step, we used the fact that $\frac{d-1}{\eta_L+1} \ge 1$ for $\rho \ge 1$ (and note that the last above-displayed equation also holds if $\rho = 0$).
	
	Finally, we remark that the case of $s < b \le \tau_j$ can also be handled by symmetry. 
\end{proof}

\begin{Lemma}
	\label{Lem:cusum_cp_kink_distance_2}
	Suppose $\fb=(f_{1},\ldots,f_{T})'$ is piecewise-linear vector as in Scenario \ref{Scen:change_in_slope}, and $\tau_{1},\ldots,\tau_{q}$ are the locations of the change-points. Suppose $1\leq s < e\leq T$, such that $\tau_{j-1} \le  s < \tau_{j} <  e \leq \tau_{j+1}$  for some $j=1\ldots, q$.  Let $\rho = |\tau_{j} -b|$, $\eta_L = \tau_{j} - s$, $\eta_R = e - \tau_{j}$ and $\Delta_j^\fb = |2f_{\tau_{j}}-f_{\tau_{j}-1}-f_{\tau_{j}+1}|$. Then, for any $b$ satisfying $\tau_j- \min (\eta_L, \eta_R)/2 < b < \tau_j + \min (\eta_L, \eta_R)/2$, we have that 
	\[
	(\cont{s}{e}{\tau_j}{\fb})^2-(\cont{s}{e}{b}{\fb})^2 \ge \frac{(\Delta^{\fb}_j)^2}{96} \big\{ \min(\eta_L,\eta_R) -1 \big\}  \rho^2.
	\]
\end{Lemma}
\begin{proof} Here we focus on the scenario where $b > \tau_j$. By Lemma~\ref{Lem:cusum_cp_kink_distance}, 
	\begin{align*}
	(\cont{s}{e}{\tau_j}{\fb})^2-(\cont{s}{e}{b}{\fb})^2  &= 	\| \phib_{s,e}^b \langle\fb,\phib_{s,e}^b \rangle - \phib_{s,e}^{\tau_j} \langle\fb,\phib_{s,e}^{\tau_j}\rangle\|_2^2 =  \min_{a_0, a_1, a_2 \in \R} \big\|\fb|_{[s,e]} - a_0 \mathbf{1}_{s,e} - a_1 \gammab_{s,e} - a_2  \phib_{s,e}^{b} \big\|_2^2\\
	&= (\Delta^{\fb}_j)^2 \min_{a_0, a_1, a_2 \in \R} \big\|\tilde{\fb}|_{[s,e]} - a_0 \mathbf{1}_{s,e} - a_1 \gammab_{s,e} - a_2  \phib_{s,e}^{b}\big\|_2^2,
	\end{align*}
	where $\tilde{\fb}|_{[s,e]} := (0,\ldots,0,1,\ldots,e-\tau_j,0,\ldots,0)'$, in which ``$1$'' appears at the $(\tau_j+1)$-th position. In the following, our aim is that bound the residual sum of squares of fitting $ \tilde{\fb}|_{[s,e]}$ using a piecewise-linear and continuous function with only one kink at $b$ on $[s,e]$. Assuming that the fitted value of this vector at the $b$-th position is $m$, then, we have that 
	\begin{align*}
	&\min_{a_0, a_1, a_2 \in \R} \big\|\tilde{\fb}|_{[s,e]} - a_0 \mathbf{1}_{s,e} - a_1 \gammab_{s,e} - a_2  \phib_{s,e}^{b}\big\|_2^2\\
	&\quad	\ge \Big(\frac{2m}{\eta_L+2\rho}\Big)^2 \times \frac{1}{6} \frac{\eta_L}{2}\Big(\frac{\eta_L}{2}+1\Big) (\eta_L+1)+ \Big\{\frac{2(\rho-m)}{e-b}\Big\}^2 \times   \frac{1}{6}\Big(\frac{e-b}{2}-1\Big)\frac{e-b}{2} (e-b-1).\\
	\end{align*}
	Since  $b < \tau_j + \eta_R/2$, it follows that $e-b > \eta_R/2$. Moreover, the fact of $\rho < \min(\eta_L,\eta_R)/2$ yields $\eta_L + 2\rho \le 2\eta_L$. Plugging these two inequalities into the previous equation, we have that 
	\begin{align*}
	&\min_{a_0, a_1, a_2 \in \R} \big\|\tilde{\fb}|_{[s,e]} - a_0 \mathbf{1}_{s,e} - a_1 \gammab_{s,e} - a_2  \phib_{s,e}^{b}\big\|_2^2\\
	& \quad \ge \frac{m^2\eta_L}{24} + (\rho-m)^2\frac{\eta_R - 1}{48}		\ge\frac{1}{2} \min\Big(\frac{\eta_L}{24},\frac{\eta_R-1}{48}\Big) \rho^2
	\end{align*}
	Consequently,
	\[
	(\cont{s}{e}{\tau_j}{\fb})^2-(\cont{s}{e}{b}{\fb})^2 \ge \frac{(\Delta^{\fb}_j)^2}{96} \big\{ \min(\eta_L,\eta_R) -1 \big\}  \rho^2.
	\]
	By symmetry, the scenario of $b < \tau_j$ can be dealt with in a similar fashion. Finally, we remark that the constants here are not sharp, as we will only use this lemma to establish rate-type results later.
\end{proof}	

\subsection{Proof of Theorem~\ref{Thm:consistency_gaussian_const}}
\label{Sec:proof_ext}
Here we informally discuss our proof strategy, which could be generalised to other scenarios. 
\begin{itemize}
	\item Intuitively speaking, lemmas from Appendix~\ref{Sec:lemmas} deal with noiseless versions of the change-point estimation problems. In order to apply these results to show the consistency of estimated number of change-points, we need to control $\|\cont{s}{e}{b}{\Yb} - \cont{s}{e}{b}{\fb}\|$ for every $(s,e,b)$, which can be achieved using Bonferroni in Step~One. 
	\item Note that for any fixed interval with start-point $s$ and end-point $e$, to decide whether $b_1$ or $b_2$ is a more suitable change-point candidate inside this interval, we only need to look at the value of $\cont{s}{e}{b_1}{\Yb} - \cont{s}{e}{b_2}{\Yb}$. Therefore, when establishing the convergence rate of the estimated change-point location , we control the distance between $\cont{s}{e}{b_1}{\Yb} - \cont{s}{e}{b_2}{\Yb}$ and its noiseless analogue $\cont{s}{e}{b_1}{\fb} - \cont{s}{e}{b_2}{\fb}$ (after proper normalisation) for all tuples $(s,e,b_1,b_2)$ in Step~Two. 
	\item In Step~Three, we show that given a properly chosen threshold and a large enough $M$, both bounds in Step~One and Step~Two hold, and for each change-point $\tau_j$, there exists an interval from $F_T^M$ that contains only this change-point and both its start- and end- points are sufficiently far away from other change-points. Since we are dealing with the narrowest-over-threshold intervals, the actual intervals that our NOT algorithm pick must have length no longer than the ones we considered in Step~Three, thus could only contain precisely one change-point. 
	\item So in Step~Four, it suffices to investigate a single change-point detection problem, where we can use lemmas from Appendix~\ref{Sec:lemmas} and the bound in Step~Two to establish the convergence rate for its location estimation. 
	\item Finally, in Step~Five, we show that after detecting all the change-points, the NOT algorithm stops with no further detection. This is because the remaining elements $[s,e]\in F_T^M$ to be considered either have no change-point inside, or have one/two change-points that are very close to its start- or/and end- points, thus their corresponding $\max_{b} \cont{s}{e}{b}{\Yb}$ cannot exceed the given threshold in views of the property of its noiseless analogue and the bound from Step~One.
\end{itemize}

Now we proceed to the technical details.
\label{Sec:proof_of_consistency_theorem_const}
\begin{proof} 
	We shall prove the following more specific result, which in turn implies (\ref{Eq:consistency_gaussian_const}).
	\begin{align}
	\Pb{\hat{q}=q,\; \max_{j=1,\ldots,q} \Big(|\hat{\tau}_{j}-\tau_{j}|(\Delta_j^\fb)^2 \Big) \leq C_3 \log T} \geq 1- T^{-1}/(6\sqrt{\pi}) - T\delta_{T}^{-1}(1-\delta_{T}^{2}T^{-2}/36)^{M},
	\end{align}	
	
	\subsubsection*{Step One.}
	
	Let $\varepsilonb = (\varepsilon_1,\ldots,\varepsilon_T)'$ and $\lambda_T = \sqrt{ 8 \log T}$. Define the set
	\[
	A_T = \Big \{ \max_{s,b,e: 1\le s \le b < e \le T} |\cont{s}{e}{b}{\varepsilonb}| \le \lambda_T \Big \}.
	\]
	Note that for any $ 1\le s \le b < e \le T$, $\cont{s}{e}{b}{\varepsilonb}$ follows a standard normal distribution. Therefore, using the Bonferroni bound, we get
	\[
	\Pb{A_T^c} \leq \frac{T^3}{6} \; \frac{2 e^{-(\sqrt{8 \log T}) ^2/2}}{ \sqrt{8 \log T} \sqrt{2\pi}}  \le \frac{T^{-1}}{12 \sqrt{\pi}}.
	\]
	Moreover, because $\cont{s}{e}{b}{\Yb} - \cont{s}{e}{b}{\fb} = \cont{s}{e}{b}{\varepsilonb}$, so $A_T$ also implies that
	\[
	\Big \{ \max_{s,b,e: 1\le s \le b < e \le T} |\cont{s}{e}{b}{\Yb} -\cont{s}{e}{b}{\fb}| \le \lambda_T \Big \}.
	\]
	We remark that though the constant in $\lambda_T$ (i.e. $\sqrt{8}$) does not appear sharp (as it is rooted in the simple Bonferroni bound), it is sufficient for our purpose of establishing consistency and rate-type results later. We refer the readers to \citet{DS2001} and \citet{RW2010} for possible improvement over this constant.
	
	\subsubsection*{Step Two.}
	
	Define the set 
	\[
	B_T = \Big \{ \max_{j=1,\ldots, q} \;\; \max_{\substack{\tau_{j-1} < s \le \tau_j \\ \tau_j < e \le \tau_{j+1} \\ s \le b < e}} \frac{\Big|\Big\langle \psib_{s,e}^b \langle\fb,\psib_{s,e}^b \rangle - \psib_{s,e}^{\tau_j} \langle\fb,\psib_{s,e}^{\tau_j} \rangle  ,\varepsilonb\Big\rangle\Big|}{\|\psib_{s,e}^b \langle\fb,\psib_{s,e}^b \rangle - \psib_{s,e}^{\tau_j} \langle\fb,\psib_{s,e}^{\tau_j} \rangle\|_2} \le \lambda_T \Big \}.
	\]
	Again, for any $ 1\le s \le b < e \le T$, $\frac{|\langle \psib_{s,e}^b \langle\fb,\psib_{s,e}^b \rangle - \psib_{s,e}^{\tau_j} \langle\fb,\psib_{s,e}^{\tau_j} \rangle  ,\varepsilonb\rangle|}{\|\psib_{s,e}^b \langle\fb,\psib_{s,e}^b \rangle - \psib_{s,e}^{\tau_j} \langle\fb,\psib_{s,e}^{\tau_j} \rangle\|_2}$ follows a standard normal distribution, so using a similar argument, we get 
	\[
	\Pb{B_T^c} \le \frac{T^{-1}}{12 \sqrt{\pi}}.
	\]
	
	\subsubsection*{Step Three.}
	
	To fix the ideas, for $j=1,\ldots,q$, we define intervals 
	\begin{align}
	\Ic_{j}^{L} &= (\tau_{j} -\delta_T/3, \tau_{j} - \delta_T/6) \label{Eq:intervals_arnd_cpts_left}\\
	\Ic_{j}^{R} &= (\tau_{j} +\delta_T/6, \tau_{j}+\delta_T/3)\label{Eq:intervals_arnd_cpts_right}
	\end{align}
	Note that these intervals all contain at least one integer as long as $\delta_T > 6$. This is always true for sufficiently large $T$, as it follows from Conditions 1 and 2 that $\delta_T > \Cl \log T / \fl$. Recall that $F_{T}^{M}$ is the set of $M$ randomly drawn intervals with endpoints in $\{1,\ldots, T\}$. Denote by $[s_{1},e_{1}], \ldots, [s_{M},e_{M}]$ the elements of $F_{T}^{M}$ and let 
	\begin{align}
	D^{M}_{T}  = \Big\{\forall j=1,\ldots,q, \exists k \in \{1,\ldots,M\}, \; \mbox{s.t.} \; s_k \times e_k \in \Ic_{j}^{L}\times\Ic_{j}^{R}\Big\} \label{Eq:event_D_M_T}.
	\end{align}
	We have that 
	\begin{align*}
	\Pb{(D^{M}_{T})^{c}} &\leq \sum_{j=1}^{q}\Pi_{m=1}^{M}\Big(1-\Pb{s_m \times e_m \in \Ic_{j}^{L}\times\Ic_{j}^{R}}\Big) \\
	&\le q  \left(1-\frac{\delta_{T}^{2}}{6^2 T^{2}}\right)^{M} \leq \frac{T}{\delta_{T}} \left(1 - \frac{\delta_{T}^{2}}{36 T^{2}}\right)^{M}.
	\end{align*}
	Therefore, $\Pb{A_{T}\cap B_{T}\cap D^{M}_T} \geq 1-  T^{-1}/(6 \sqrt{\pi}) - T \delta_{T}^{-1}(1-\delta_T^{2}T^{-2}/36)^M.$ 
	
	In the rest of the proof, we assume that  $A_T, B_T$ and $D^{M}_T$ all hold. We give the constants as follows:
	\[\Cl=\sqrt{6}\big(2\sqrt{C_3} + 4\sqrt{2}\big)+1, \quad C_1 = 2\sqrt{C_3} + 2\sqrt{2}, \quad C_2 = \frac{1}{\sqrt{6}} - \frac{2 \sqrt{2}}{\Cl}, \quad C_3 = 32\sqrt{2}+48.\]
	These constants could be further refined by applying the Bonferroni bound more carefully. See also our remark at the end of Step~One. But since our main aim is to establish the rate, we chose not to pursue this direction further. In addition, here we set $\Cl$ in such a way that $\Cl C_2 > C_1$ (as well as $C_2 > 0$). This means that given $\delta_T^{1/2}\fl_{T} \geq \Cl \sqrt{\log T}$, one have that $C_2 \delta_{T}^{1/2} \fl_{T} > C_1 \sqrt{\log T }$, i.e. we can select $\zeta_{T} \in [C_1 \sqrt{\log T }, C_2 \delta_{T}^{1/2} \fl_{T})$.
	
	\subsubsection*{Step Four.}
	
	We focus on a generic interval $[s,e]$ such that
	\begin{align}
	\exists j \in \{1,\ldots,q\}, \; \exists k \in \{1,\ldots,M\}, \; \mbox{s.t.}\; [s_{k},e_{k}]\subset [s,e] \mbox{ and } s_k \times e_k \in\Ic_{j}^{L}\times\Ic_{j}^{R} \label{Eq:main_th_proof_there_exists_interval}
	\end{align}
	Fix such an interval $[s,e]$ and let $j \in \{1,\ldots,q\}$ and $k \in \{1,\ldots,M\}$ be such that \eqref{Eq:main_th_proof_there_exists_interval} is satisfied. 
	Let $b_k^{*}=\argmax_{s_{k}\leq b\le  e_{k}} \cont{s_k}{e_k}{b}{\Yb}$. By construction, $[s_k,e_k]$ satisfies $\tau_j-s_{k}+1 \ge \delta_{T}/6$ and $e_{k}-\tau_j > \delta_{T}/6$. 
	Denote by 
	\begin{align*}
	\Mc_{s,e}&=\left\{m:[s_{m},e_{m}]\in F_{T}^{M}, [s_{m},e_{m}]\subset[s,e] \right\};\\
	\Oc_{s,e}&=\{m\in\Mc_{s,e}:\max_{s_m \leq b < e_m} \cont{s_m}{e_m}{b}{\Yb} > \zeta_{T}\}
	\end{align*}
	Our first aim is to show that $\Oc_{s,e}$ is non-empty. This follows from Lemma~\ref{Lem:cusum_at_cp_jump} and the calculation below.
	\begin{align*}
	\cont{s_k}{e_k}{b^{*}_k}{\Yb} &\geq \cont{s_k}{e_k}{\tau_{j}}{\Yb} \\
	&\geq \cont{s_k}{e_k}{b^{*}_k}{\fb}-\lambda_{T} \geq \left(\frac{\delta_{T}}{6}\right)^{1/2} |f_{\tau_{j}+1}-f_{\tau_{j}}|- \lambda_{T}\geq \left(\frac{\delta_{T}}{6}\right)^{1/2} \fl_{T} - \lambda_{T} \\
	& = \left( \frac{1}{\sqrt{6}} - \frac{\lambda_{T}}{\delta_{T}^{1/2}\fl_{T}}\right)\delta_{T}^{1/2}\fl_{T}\geq \left(\frac{1}{\sqrt{6}} - \frac{2 \sqrt{2}}{\Cl} \right)\delta_{T}^{1/2}\fl_{T}  =  C_2 \delta_{T}^{1/2}\fl_{T} > \zeta_T.
	\end{align*}

	Let $m^{*}=\argmin_{m\in\Oc_{s,e}} (e_{m}-s_{m}+1)$ and $b^{*} = \argmax_{s_{m^*} \leq b <e_{m^*}} \cont{s_{m^*}}{e_{m^*}}{b}{\Yb}$. Observe that $[s_{m^*},e_{m^*})$ must contain at least one change-point. Indeed, if that was not the case, we would have $\cont{s_{m^*}}{e_{m^*}}{b}{\fb}=0$ and
	\begin{align*}
	\cont{s_{m^*}}{e_{m^*}}{b^{*}}{\Yb}=	|\cont{s_{m^*}}{e_{m^*}}{b^{*}}{\Yb} - 	\cont{s_{m^*}}{e_{m^*}}{b^{*}}{\fb}| \le \lambda_{T} \le \zeta_{T}
	\end{align*}
	which contradicts $\cont{s_{m^*}}{e_{m^*}}{b^{*}}{\Yb} >  \zeta_{T}$. On the other hand, $[s_{m^*},e_{m^*})$ cannot contain more than one change-points, because $e_{m^*}- s_{m^*} + 1 \le  e_{k}- s_{k} + 1 \le \delta_{T}$, as we picked the \emph{narrowest}-over-threshold interval. 
	
	Without loss of generality, assume $\tau_{j} \in [s_{m^*}, e_{m^*}]$. Denote by  $\eta_L=\tau_{j}-s_{m^*}+1$, $\eta_R = e_{m^*}-\tau_{j}$ and $\eta_{T}= (C_1 - \sqrt{8})^2 (\Delta_j^\fb)^{-2}\log T$, where $\Delta_j^\fb= |f_{\tau_{j}+1}-f_{\tau_{j}}|$. We claim that $\min(\eta_L, \eta_R) > \eta_{T}$, because $\min(\eta_L, \eta_R) \le  \eta_{T}$ and Lemma~\ref{Lem:cusum_at_cp_jump} result in  
	\begin{align*}	
	\cont{s_{m^*}}{e_{m^*}}{b^{*}}{\Yb} &\leq \cont{s_{m^*}}{e_{m^*}}{b^{*}}{\fb} +\lambda_{T} \leq  	\cont{s_{m^*}}{e_{m^*}}{\tau_{j}}{\fb} +\lambda_{T} \leq \eta_{T}^{1/2}\Delta_j^\fb+ \lambda_{T} \\
	&= (C_1 - \sqrt{8}+ \sqrt{8}) \sqrt{\log T} = C_1\sqrt{\log T} \leq \zeta_{T},
	\end{align*}
	which contradicts  $\cont{s_{m^*}}{e_{m^*}}{b^{*}}{\Yb} > \zeta_{T}$.  
	
	We are now in the position to prove $|b^* - \tau_j| \le  C_3 \log T / (\Delta_j^\fb)^2$. The arguments we use here are simpler and slightly more general than Lemma A.3 of \citet{fryzlewicz2014wild}. Our aim is to find $\epsilon_T$ such that for any $b \in \{s_{m^*}, s_{m^*}+1,\ldots,e_{m^*}-1\}$ with $|b-\tau_j| > \epsilon_{T}$, we always have
	\begin{align}
	(\cont{s_{m^*}}{e_{m^*}}{\tau_j}{\Yb})^2 -  (\cont{s_{m^*}}{e_{m^*}}{b}{\Yb})^2 > 0. 
	\label{Eq:main_th_proof_interval_rate1}
	\end{align}
	This would then imply that $|b^* - \tau_j| \le \epsilon_T$. By expansion and rearranging the terms (using the fact that $f_t = Y_t + \varepsilon_t$), we see that (\ref{Eq:main_th_proof_interval_rate1}) is equivalent to 
	\begin{align}
	\notag
	\langle\fb, \psib_{s_{m^*},e_{m^*}}^{\tau_j} \rangle^2 - \langle\fb, \psib_{s_{m^*},e_{m^*}}^{b} \rangle^2  >  &\langle \varepsilonb, \psib_{s_{m^*},e_{m^*}}^{b} \rangle^2  -  \langle \varepsilonb, \psib_{s_{m^*},e_{m^*}}^{\tau_j} \rangle^2 \\
	& + 2 \Big\langle \varepsilonb, \psib_{s_{m^*},e_{m^*}}^b  \langle\fb,\psib_{s_{m^*},e_{m^*}}^b \rangle - \psib_{s_{m^*},e_{m^*}}^{\tau_j} \langle\fb,\psib_{s_{m^*},e_{m^*}}^{\tau_j} \rangle \Big\rangle.
	\label{Eq:main_th_proof_interval_rate2}
	\end{align}
	In the following, we assume that $b \ge \tau_j$. The case that $b < \tau_j$ can be handled in a similar fashion. By Lemma~\ref{Lem:cusum_cp_jump_distance}, we  have
	\begin{align*}
	\langle\fb, \psib_{s_{m^*},e_{m^*}}^{\tau_j} \rangle^2 - \langle\fb, \psib_{s_{m^*},e_{m^*}}^{b} \rangle^2 = (\cont{s_{m^*}}{e_{m^*}}{\tau_j}{\fb})^2 -  (\cont{s_{m^*}}{e_{m^*}}{b}{\fb})^2 = \frac{|b -\tau_j| \eta_L}{|b -\tau_j| +  \eta_L} (\Delta_j^\fb)^2 := \kappa.
	\end{align*}
	In addition, since $A_T$ and $B_T$ hold, we have that
	\begin{align*}
	&\langle \varepsilonb, \psib_{s_{m^*},e_{m^*}}^{b} \rangle^2  -  \langle \varepsilonb, \psib_{s_{m^*},e_{m^*}}^{\tau_j} \rangle^2  \le \lambda_T^2, \\
	&2 \Big\langle \varepsilonb, \psib_{s_{m^*},e_{m^*}}^b \langle\fb,\psib_{s_{m^*},e_{m^*}}^b \rangle - \psib_{s_{m^*},e_{m^*}}^{\tau_j} \langle\fb,\psib_{s_{m^*},e_{m^*}}^{\tau_j} \rangle \Big\rangle  \\ 
	& \qquad \qquad  \qquad \le 2\| \psib_{s_{m^*},e_{m^*}}^b \langle\fb,\psib_{s_{m^*},e_{m^*}}^b \rangle - \psib_{s_{m^*},e_{m^*}}^{\tau_j} \langle\fb,\psib_{s_{m^*},e_{m^*}}^{\tau_j}\rangle\|_2 \lambda_T 
	= 2 \kappa^{1/2} \lambda_T,
	\end{align*}
	where the last equality also comes from Lemma~\ref{Lem:cusum_cp_jump_distance}. Consequently, (\ref{Eq:main_th_proof_interval_rate2}) can be deducted from the stronger inequality $\kappa - 2\lambda_T\kappa^{1/2} - \lambda_T^2 > 0$. This quadratic inequality is implied by $\kappa > (\sqrt{2}+1)^2\lambda_T^2$, and could be restricted further to  
	\begin{align}
	\frac{2|b -\tau_j| \eta_L}{|b -\tau_j| +  \eta_L} \ge \min(|b -\tau_j|, \eta_L) > (32\sqrt{2}+48) (\Delta_j^\fb)^{-2}\log T = C_3  (\Delta_j^\fb)^{-2}\log T.
	\label{Eq:main_th_proof_interval_rate3}
	\end{align}
	But since 
	\[
	\eta_L \ge \eta_T = (C_1 - \sqrt{8})^2 (\Delta_j^\fb)^{-2}\log T = (2\sqrt{C_3})^2  (\Delta_j^\fb)^{-2}\log T > C_3 (\Delta_j^\fb)^{-2}\log T,
	\]
	we see that (\ref{Eq:main_th_proof_interval_rate3}) is equivalent to $|b -\tau_j| > C_3 (\Delta_j^\fb)^{-2}\log T$. To sum up, 
	$|b^{*} -\tau_j| (\Delta_j^\fb)^{2} > C_3 \log T$ would result in (\ref{Eq:main_th_proof_interval_rate1}), a contradiction. So we have proved that $|b^{*} -\tau_j| (\Delta_j^\fb)^2 \leq C_3  \log T$.
	
	\subsubsection*{Step Five.}
	Using the arguments given above which are valid on the event $A_{T}\cap B_{T}\cap D_{T}^M$, we can now proceed with the proof of the theorem as follows. At the start of Algorithm~\ref{Alg:not_algorithm} we have $s=1$ and $e=T$ and, provided that $q \ge 1$, condition \eqref{Eq:main_th_proof_there_exists_interval} is satisfied. Therefore the algorithm detects a change-point $b^{*}$ in that interval such that $|b^{*}-\tau_{j}| \leq C_3 \log T  (\Delta_j^\fb)^{-2}$. By construction, we also have that $|b^{*}-\tau_{j}| < 2/3 \delta_{T}$. This in turn implies that for all $l=1,\ldots,q$ such that $\tau_{l}\in [s,e]$ and $l\neq j$ we have either $\Ic_{l}^{L}, \Ic_{l}^{R} \subset [s,b^{*}]$ or $\Ic_{l}^{L}, \Ic_{l}^{R} \subset [b^{*}+1,e]$. Therefore \eqref{Eq:main_th_proof_there_exists_interval} is satisfied within each segment containing at least one change-point. Note that before all $q$ change-points are detected, each change-point will not be detected twice. To see this, we suppose that $\tau_j$ has already been detected by $b$, then for all intervals $[s_k,e_k] \subset [\tau_j - C_3 \log T  (\Delta_j^\fb)^{-2} + 1, \tau_j + 2/3\delta_T] \cup [\tau_j - 2/3\delta_T, \tau_j + C_3 \log T  (\Delta_j^\fb)^{-2}]$, Lemma~\ref{Lem:cusum_at_cp_jump}, together with the event $A_T$, guarantees that 
	\[
	\max_{s_k \leq b < e_k}\cont{s_k}{e_k}{b}{\Yb} \le \max_{s\leq b <  e}\cont{s_k}{e_k}{b}{\fb} + \lambda_T \le \sqrt{ C_3 \log T  (\Delta_j^\fb)^{-2}} \Delta_j^\fb + \lambda_T \leq C_1 \sqrt{\log T} \le \zeta_{T}.
	\] 
	Once all the change-points are detected, we then only need to consider $[s_k,e_k]$ such that
	\[
	[s_k,e_k] \subset [\tau_j - C_3 \log T  (\Delta_j^\fb)^{-2}+1, \tau_{j+1} + C_3 \log T  (\Delta_{j+1}^{f})^{-2}]
	\]
	for $j = 0,\ldots,q$, where we set $\Delta_{0}^{f} = \Delta_{q+1}^{f} = \infty$ for notational convenience. It follows from Lemma~\ref{Lem:cusum_two_cp_jump} (within $A_T$) that
	\begin{align*}
	\max_{s_k \leq b < e}\cont{s_k}{e_k}{b}{\Yb} &\le \max_{s\leq b <  e}\cont{s_k}{e_k}{b}{\fb} + \lambda_T \\
	&\le \sqrt{ C_3 \log T  (\Delta_j^\fb)^{-2}} \Delta_j^\fb + \sqrt{ C_3 \log T  (\Delta_{j+1}^{f})^{-2}} \Delta_{j+1}^{\fb} +  \lambda_T \\
	& < (2\sqrt{C_3} + \sqrt{8}) \sqrt{\log T} = C_1 \sqrt{\log T} \le \zeta_{T}.
	\end{align*}
	Hence the algorithm terminates and no further change-points are detected. 
\end{proof}

\subsection{Proof of Theorem~\ref{Thm:consistency_gaussian_linear}}
\label{Sec:proof_of_consistency_theorem_slope}
\begin{proof}
	The proof proceeds in analogy to the proof of Theorem~\ref{Thm:consistency_gaussian_const}. In five steps we shall establish the following result,
	\begin{align}
	\Pb{\hat{q}=q,\; \max_{j=1,\ldots,q} \Big(|\hat{\tau}_{j}-\tau_{j}|(\Delta_{j}^{\fb})^{2/3}  \Big) \leq C_{3} (\log T)^{1/3}} \geq 1- T^{-1}/(6\sqrt{\pi}) - T\delta_{T}^{-1}(1-\delta_{T}^{2}T^{-2}/36)^{M},
	\end{align}	
	which in turn implies \eqref{Eq:consistency_gaussian_linear}.
	
	\subsubsection*{Step One and Step Two}
	We define the following two events 
	\begin{align*}
	A_T &= \Big \{ \max_{s,b,e: 1\le s \le b < e \le T} |\cont{s}{e}{b}{\varepsilonb}| \le \lambda_T \Big \},\\
	B_T &= \Big \{ \max_{j=1,\ldots, q} \;\; \max_{\substack{\tau_{j-1} < s \le \tau_j \\ \tau_j < e \le \tau_{j+1} \\ s \le b < e}} \frac{\Big|\Big\langle \phib_{s,e}^b \langle\fb,\phib_{s,e}^b \rangle - \phib_{s,e}^{\tau_j} \langle\fb,\phib_{s,e}^{\tau_j} \rangle  ,\varepsilonb\Big\rangle\Big|}{\|\phib_{s,e}^b \langle\fb,\phib_{s,e}^b \rangle - \phib_{s,e}^{\tau_j} \langle\fb,\phib_{s,e}^{\tau_j} \rangle\|_2} \le \lambda_T \Big \}, 
	\end{align*}
	where $\lambda_{T} = \sqrt{8\log T}$. Arguments as those used in Step One and Step Two of the proof of Theorem~\ref{Thm:consistency_gaussian_const} show that $\Pb{A_T^c} \le \frac{T^{-1}}{12 \sqrt{\pi}}$ and $\Pb{B_T^c} \le \frac{T^{-1}}{12 \sqrt{\pi}}$.
	
	\subsubsection*{Step Three}
	In the rest of the proof, we assume that $A_T$, $B_T$ and $D_T^M$ all hold, where the last event is given by \eqref{Eq:event_D_M_T}. Exactly as in the proof of Theorem~\ref{Eq:consistency_gaussian_const}, we show that $\Pb{A_{T}\cap B_{T}\cap D^{M}_T} \geq 1-  T^{-1}/(6 \sqrt{\pi}) - T \delta_{T}^{-1}(1-\delta_T^{2}T^{-2}/36)^M.$ 
	
	We give the constants as follows:
	\[\Cl = 72\Big(4\sqrt{2}+2\sqrt{\frac{2}{3}} C_{3}^{3/2}\Big)+1 , \quad C_1 = 2 \sqrt{\frac{2}{3}} C_{3}^{3/2} + 2\sqrt{2}, \quad C_2 = \frac{1}{72} - \frac{2 \sqrt{2}}{\Cl}, \quad C_3 = 2 \sqrt[3]{7} \left(3 \left(1+\sqrt{2}\right)\right)^{2/3}.\]
	Here we set $\Cl$ in such a way that $\Cl C_2 > C_1$ (which also implies that $C_2 >0$). Consequently, given $\delta_T^{3/2}\fl_{T} \geq \Cl \sqrt{\log T}$ it is possible to select $\zeta_{T}\in \left[C_{1}\sqrt{\log T}, C_{2}\delta_T^{3/2}\fl_{T}\right)$.  
	
	Again, these constants could be further refined. But since our main aim is to establish the rate, we chose not to pursue this direction here. 
	
	\subsubsection*{Step Four}
	Consider a generic interval $[s,e]$ satisfying 
	\begin{align}
	\exists j \in \{1,\ldots,q\}, \; \exists k \in \{1,\ldots,M\}, \; \mbox{s.t.}\; [s_{k},e_{k}]\subset [s,e] \mbox{ and } s_k \times e_k \in\Ic_{j}^{L}\times\Ic_{j}^{R} \label{Eq:th_linear_case_proof_there_exists_interval}
	\end{align}
	and define events 
	\begin{align*}
	\Mc_{s,e}&=\left\{m:[s_{m},e_{m}]\in F_{T}^{M}, [s_{m},e_{m}]\subset[s,e] \right\},\\
	\Oc_{s,e}&=\{m\in\Mc_{s,e}:\max_{s_m \leq b < e_m} \cont{s_m}{e_m}{b}{\Yb} > \zeta_{T}\}.
	\end{align*}
	Let $b_k^{*}=\argmax_{s_{k}\leq b\le  e_{k}} \cont{s_k}{e_k}{b}{\Yb}$. We have
	\begin{align*}
	\cont{s_k}{e_k}{b^{*}_k}{\Yb} &\geq \cont{s_k}{e_k}{\tau_{j}}{\Yb} \\
	&\geq \cont{s_k}{e_k}{b^{*}_k}{\fb}-\lambda_{T} \geq \frac{1}{\sqrt{24}}\left(\delta_{T}/6\right)^{3/2} \Delta_j^\fb- \lambda_{T}\geq \frac{1}{72}\delta_{T}^{3/2} \fl_{T} - \lambda_{T} \\
	& = \left( \frac{1}{72} - \frac{\lambda_{T}}{\delta_{T}^{3/2}\fl_{T}}\right)\delta_{T}^{1/2}\fl_{T}\geq \left(\frac{1}{72} - \frac{2 \sqrt{2}}{\Cl} \right)\delta_{T}^{3/2}\fl_{T}  =  C_2 \delta_{T}^{3/2}\fl_{T} > \zeta_T,
	\end{align*}
	where the third inequality above  follows from Lemma~\ref{Lem:cusum_at_cp_kink}, therefore $\Oc_{s,e}$ is non-empty.
	
	Let $m^{*} = \argmin_{m\in\Oc_{s,e}} (e_m - s_m+1)$ and $b^{*} = \argmax_{s_{m^*} \leq b < e_{m^*}} \cont{s_{m^*}}{e_{m^*}}{b}{\Yb}$. Arguing exactly as in Step~Four in the proof of Theorem~\ref{Thm:consistency_gaussian_const}, we show that $[s_{m^*}, e_{m^*})$ must contain exactly one change-point. Without loss of generality, assume that  $\tau_{j}\in [s_{m^*}, e_{m^*})$. Let $\eta_{L}=\tau_{j} -s_{m^*}$, $\eta_{R}=e_{m^{*}}-\tau_{j}$ and $\eta_{T}=\left(\sqrt{3}(C_{1}-\sqrt{8})\sqrt{\log T} (\Delta_{j}^{\fb})^{-1})\right)^{2/3}-1$.  We observe that $\min(\eta_L, \eta_R) > \eta_{T}$, as  $\min(\eta_L, \eta_R) \leq \eta_{T}$ and Lemma~\ref{Lem:cusum_at_cp_kink} implies that 
	\begin{align*}	
	\cont{s_{m^*}}{e_{m^*}}{b^{*}}{\Yb} &\leq \cont{s_{m^*}}{e_{m^*}}{b^{*}}{\fb} +\lambda_{T} \leq  	\cont{s_{m^*}}{e_{m^*}}{\tau_{j}}{\fb} +\lambda_{T} \leq \frac{1}{\sqrt{3}}(\eta_{T}+1)^{3/2}\Delta_j^\fb+ \lambda_{T} \\
	&= (C_1 - \sqrt{8}+ \sqrt{8}) \sqrt{\log T} = C_1\sqrt{\log T} \leq \zeta_{T},
	\end{align*}
	contradicting $\cont{s_{m^*}}{e_{m^*}}{b^{*}}{\Yb} > \zeta_{T}$.
	
	We are now in the position to prove that $|b^{*} - \tau_{j}| \leq  C_{3} (\Delta_{j}^{\fb})^{-2/3} (\log T)^{1/3} := \epsilon_{T}$. Let $b\in\{s_{m^{*}}+1,\ldots, e_{m^{*}}-2\}$.  We claim that  when $|b-\tau_j|>\epsilon_{T}$,
	\begin{align}
	(\cont{s_{m^*}}{e_{m^*}}{\tau_{j}}{\Yb})^{2}-(\cont{s_{m^*}}{e_{m^*}}{b}{\Yb})^{2} > 0. \label{Eq:linear_th_proof_interval_rate1}
	\end{align}
	Since inequality \eqref{Eq:linear_th_proof_interval_rate1} does not hold for $b=b^{*}$, so proving this claim consequently demonstrates that $|b^* - \tau_j|\leq \epsilon_{T}$.  
	
	Without loss of generality, we consider the case of $b>\tau_{j}$. Using arguments as those in Step~Four of the proof of Theorem~\ref{Thm:consistency_gaussian_const} we can show that \eqref{Eq:linear_th_proof_interval_rate1} is implied by $\kappa > (\sqrt{2}+1)^2\lambda_T^2$, where $\kappa=(\cont{s_{m^*}}{e_{m^*}}{\tau_{j}}{\fb})^{2}-(\cont{s_{m^*}}{e_{m^*}}{b}{\fb})^{2}$.  By Lemma~\ref{Lem:cusum_cp_kink_distance}, $\kappa > (\sqrt{2}+1)^2\lambda_T^2$ is implied by
	\begin{align*}
	\min\left(|b-\tau_j|, \eta_{L}\right) > \left(63 (\Delta_{j}^{\fb})^{-2} \cdot 8 (\sqrt{2}+1)^{2}\log T\right) ^ {1/3} = C_{3} (\Delta_{j}^{\fb})^{-2/3} (\log T)^{1/3}
	\end{align*} 
	However, for sufficiently large $T$, 
	\begin{align*}
	\eta_L > \eta_{T} & =  (\sqrt{3}(C_{1}-\sqrt{8}))^{2/3} (\Delta_{j}^{\fb})^{-2/3}(\log T)^{1/3} -1 >  (C_{1}-\sqrt{8})^{2/3} (\Delta_{j}^{\fb})^{-2/3}(\log T)^{1/3}\\
	& =   (C_{3}^{3/2}+\sqrt{8}-\sqrt{8})^{2/3} (\Delta_{j}^{\fb})^{-2/3} = C_{3} (\Delta_{j}^{\fb})^{-2/3} (\log T)^{1/3} = \epsilon_{T},
	\end{align*}
	hence $|b-\tau_{j}| >  \epsilon_{T}$ implies \eqref{Eq:linear_th_proof_interval_rate1}, so it must hold that $|b^{*}-\tau_{j}| \leq \epsilon_{T}$. 
	
	\subsubsection*{Step Five}
	
	Using the arguments given above which are valid on the event $A_{T}\cap B_{T}\cap D_{T}^M$, we can now proceed with the proof of the theorem as follows. At the start of Algorithm~\ref{Alg:not_algorithm} we have $s=1$ and $e=T$ and, provided that $q \ge 1$, condition \eqref{Eq:main_th_proof_there_exists_interval} is satisfied. Therefore the algorithm detects a change-point $b^{*}$ in that interval such that $|b^{*}-\tau_{j}| \leq C_{3} (\Delta_{j}^{\fb})^{-2/3} (\log T)^{1/3}$. By construction, we also have that $|b^{*}-\tau_{j}| < 2/3 \delta_{T}$. This in turn implies that for all $l=1,\ldots,q$ such that $\tau_{l}\in [s,e]$ and $l\neq j$ we have either $\Ic_{l}^{L}, \Ic_{l}^{R} \subset [s,b^{*}]$ or $\Ic_{l}^{L}, \Ic_{l}^{R} \subset [b^{*}+1,e]$. Therefore \eqref{Eq:main_th_proof_there_exists_interval} is satisfied within each segment containing at least one change-point. Note that before all $q$ change-points are detected, each change-point will not be detected twice. To see this, we suppose that $\tau_j$ has already been detected by $b$, then for all intervals $[s_k,e_k] \subset [\tau_j - \epsilon_T + 1, \tau_j - \epsilon_T + 2/3\delta_T + 1] \cup [\tau_j + \epsilon_T - 2/3\delta_T, \tau_j + \epsilon_T]$, Lemma~\ref{Lem:cusum_at_cp_kink}, together with the event $A_T$, guarantees that
	\begin{align*}
	\max_{s_k \leq b < e_k}\cont{s_k}{e_k}{b}{\Yb} &\le \max_{s\leq b <  e} \cont{s_k}{e_k}{b}{\fb} + \sqrt{8 \log T} \le \frac{1}{\sqrt{3}}( C_{3} (\Delta_{j}^{\fb})^{-2/3} (\log T)^{1/3}+1)^{3/2} \Delta_j^\fb + \sqrt{8 \log T}\\ 
	&\leq (2 \sqrt{\frac{2}{3}} C_{3}^{3/2} +\sqrt{8}) \sqrt{\log T} = C_1\sqrt{\log T}  \le \zeta_{T}
	\end{align*}  
	Once all the change-points are detected, we then only need to consider $[s_k,e_k]$ such that
	\[
	[s_k,e_k] \subset [\tau_j - C_{3} (\Delta_{j}^{\fb})^{-2/3} (\log T)^{1/3}+1, \tau_{j+1} + C_{3} (\Delta_{j+1}^{\fb})^{-2/3} (\log T)^{1/3}]
	\]
	for $j = 0,\ldots,q$, where we set $\Delta_{0}^{f} = \Delta_{q+1}^{f} = \infty$ for notational convenience. It follows from Lemma~\ref{Lem:cusum_two_cp_kink} (within $A_T$) that
	\begin{align*}
	\max_{s_k \leq b < e}\cont{s_k}{e_k}{b}{\Yb} &\le \max_{s\leq b <  e}\cont{s_k}{e_k}{b}{\fb} + \sqrt{8 \log T} \\
	&\le \frac{1}{\sqrt{3}} (C_{3} (\Delta_{j}^{\fb})^{-2/3} (\log T)^{1/3})^{3/2} \Delta_j^\fb + \frac{1}{\sqrt{3}} (C_{3} (\Delta_{j}^{\fb})^{-2/3} (\log T)^{1/3})^{3/2} \Delta_{j+1}^{\fb}+ \sqrt{8 \log T}\\
	& = (\frac{2}{\sqrt{3}}C_3^{3/2} + \sqrt{8}) \sqrt{\log T} \leq C_1 \sqrt{\log T} \le \zeta_{T}.
	\end{align*}
	Hence the algorithm terminates and no further change-points are detected. 
\end{proof}

\subsection{Proof of Theorem~\ref{Thm:consistency_gaussian_const_sSIC}}
\label{Sec:proof_of_consistency_theorem_sSIC_constant}
\begin{proof}
	Recall that $\{\varepsilon_t\}_{t=1}^T$ are i.i.d. $N(0,\sigma_0^2)$ with $\sigma_0=1$. For any candidate $\Tc(\zeta^{(k)})$ on the NOT solution path,  the sSIC criterion function in \ref{Scen:change_in_mean} can be written as 
	\[
	T \hat{\sigma}^2_k +  (\hat{q}_k+1) \log^\alpha(T) + \mathrm{constant}
	\]
	where $\hat{\sigma}^2_k$ is the estimated variance of the noise (i.e. the residual sum of squares divided by $T$) based on $\Tc(\zeta^{(k)})$, and $\hat{q}_k$ is the estimated number of change-points. 
	
	We now divide our proof into three parts. 
	
	\subsubsection*{Part I. About a particular model candidate on the NOT solution path}
	By Theorem~\ref{Thm:consistency_gaussian_const}, we know that with arbitrarily high probability for sufficiently large $T$, there exists $k^*$ such that $\Tc(\zeta^{(k^*)})$ on the NOT solution path is a ``good'' candidate with $\hat{\tau}_{1},\ldots,\hat{\tau}_{\hat{q}_{k^*}} \in \Tc(\zeta^{(k^*)})$ satisfying $\hat{q}_{k^*} = q$ and $\max_{i=1}^{q}|\hat{\tau}_{i}-\tau_i| \le C' \log T$ for some $C' > 0$. In the rest of the proof, for presentational convenience, we condition on the event that such $k^*$ does exist throughout our analysis. 
	
	In addition, recall that  $\mathbf{1}_{s,e} = \big(\mathbf{1}_{s,e}(1),\ldots,\mathbf{1}_{s,e}(T)\big)'$ with 
	\begin{align} 
	\mathbf{1}_{s,e}(t) = \begin{cases}	(e-s+1)^{-1/2}, & t=s,\ldots, e\\ 0, &\mbox{otherwise} \end{cases},
	\end{align}
	and define the set
	\[
	E_T = \Big \{ \max_{s,e: 1\le s \le e \le T} |\langle \mathbf{1}_{s,e}, {\varepsilonb}\rangle| \le \sqrt{6 \log T} \Big \}.
	\]
	Using an argument similar to Step One of the proof of Theorem~\ref{Thm:consistency_gaussian_const}, we see that $\Pb{E_T^c} = O(T^{-1})$. Since we are only interested in proving a certain type of probabilistic statement for $T \rightarrow \infty$, here we could also assume that $E_T$ holds.
	
	Let $\{\hat{f}_t\}_{t=1}^T$ be the fitted values using the candidate on the solution path with $\hat{\tau}_{1},\ldots,\hat{\tau}_{\hat{q}_{k^*}} \in \Tc(\zeta^{(k^*)})$ , and define $\tilde{f}_t = f_{\tau_j}$ for $t = \hat{\tau}_j,\ldots,\hat{\tau}_{j+1}-1$ for every $j = 0, 1,\ldots q$. Here for notational convenience, we suppressed the dependence of $\{\hat{f}_t\}_{t=1}^T$ and $\{\tilde{f}_t\}_{t=1}^T$ on $k^*$. It is easy to see that $f_t- \tilde{f}_t$ is piecewise-constant, only non-zero for $t$ between the true location of the change-point $\tau_j$ and its estimation $\hat{\tau}_j$, and exactly zero elsewhere. Write $\tilde{\fb}=(\tilde{f}_1,\ldots,\tilde{f}_T)'$. Then
	\begin{align*}
	T \hat{\sigma}^2_{k^*} &= \sum_{t=1}^T(\varepsilon_t+f_t-\hat{f}_t)^2  \le \sum_{t=1}^T(\varepsilon_t+f_t-\tilde{f}_t)^2 =  \sum_{t=1}^T \varepsilon_t^2 + 2\langle \varepsilonb, \fb -\tilde{\fb} \rangle + \|\fb -\tilde{\fb}\|^2 \\
	&=  \sum_{t=1}^T \varepsilon_t^2 +  4 q \bar{C}  \sqrt{6\log T} \sqrt{C'\log T}  + q (2\bar{C})^2 C' \log T  \\
	&= \sum_{t=1}^T\varepsilon_t^2 +  ( 4q \bar{C} \sqrt{6C'} + 4q C' \bar{C}^2 )\log T
	\end{align*}
	where the second last step follows from $E_T$, linearity of the inner product, and the fact that $\max_{i=1}^{q}|\hat{\tau}_{i}-\tau_i| \le C' \log T$. Consequently, it follows that $\Pb{\hat{\sigma}^2_{k^*} < 1+ \delta} = 1$ for any $\delta >0$ as $T \rightarrow \infty$.

	\subsubsection*{Part II. Estimation of the number of change-points}
	In this part, we prove that for NOT with the sSIC, $\Pb{\hat{q}=q} \rightarrow 1$ as $T \rightarrow \infty$. We accomplish this by showing separately that (i) $\Pb{\hat{q}<q} \rightarrow 0$ and (ii) $\Pb{\hat{q}>q} \rightarrow 0$. 
	
	First, we note that it follows from Lemma~3 of \citet{Yao1988} that there exists $\delta >0$ such that as $T \rightarrow \infty$,
	\[
	\min_{k: \hat{q}_k < q} \Pb{\hat{\sigma}^2_k > 1+ \delta} \rightarrow 1.
	\]
	This means that for all $k$ with $\hat{q}_k < q$,
	\[
	\mbox{sSIC}(k) - \mbox{sSIC}(k^*) = T(\hat{\sigma}_k^2 - \hat{\sigma}_{k^*}^2) + (\hat{q}_k - q) \log^\alpha(T) \ge \delta T -  q\log^\alpha(T) > 0
	\] 
	for large enough $T$, which implies $\Pb{\hat{q}<q} \rightarrow 0$.
	
	Second, for all $k$ with $\hat{q}_k > q$ and $\hat{\tau}_{1},\ldots,\hat{\tau}_{\hat{q}_{k}} \in \Tc(\zeta^{(k)})$, we consider a ``saturated oracle'' candidate model with $\hat{q}_k+q$ change-points at $\hat{\tau}_{1},\ldots,\hat{\tau}_{\hat{q}_k},\tau_1,\ldots,\tau_q$ respectively. We reorder these $\hat{q}_k+q$ locations as $0 = \mathring{\tau}_{0} < \mathring{\tau}_{1} \le \ldots \le  \mathring{\tau}_{\hat{q}_k+q} <  \mathring{\tau}_{\hat{q}_k+q+ 1} = T$, and denote the estimated variance of the errors corresponding this saturated oracle candidate by $\mathring{\sigma}_k^2$. Since for each $j = 0, \ldots, \hat{q}_k+q$, $f_t$ is constant over $\{ 1+\mathring{\tau}_j,\ldots,\mathring{\tau}_{j+1} \}$,  it then follows that
	\begin{align*} 
	T \hat{\sigma}^2_{k} \ge T \mathring{\sigma}^2_{k} &= \sum_{j=0}^{\hat{q}_k+q} \sum_{t=1+\mathring{\tau}_j}^{\mathring{\tau}_{j+1}}\Big\{\varepsilon_t - \frac{1}{\mathring{\tau}_{j+1}-\mathring{\tau}_{j}}\sum_{b=1+\mathring{\tau}_j}^{\mathring{\tau}_{j+1}}\varepsilon_b\Big\}^2 \\
	&= \sum_{t=1}^T \varepsilon_t^2 - \sum_{j=0}^{\hat{q}_k+q} \langle \varepsilonb,\mathbf{1}_{1+\mathring{\tau}_j,\mathring{\tau}_{j+1}} \rangle^2 \ge \sum_{t=1}^T \varepsilon_t^2 - 6 (q+\hat{q}_k+1) \log T,
	\end{align*}
	where the last line again follows from $E_T$. This means that for all $k$ with $\hat{q}_k > q$,
	\begin{align*}
	\mbox{sSIC}(k) - \mbox{sSIC}(k^*) &\ge T(\mathring{\sigma}_k^2 - \hat{\sigma}_{k^*}^2) + (\hat{q}_k - q) \log^\alpha(T) \\
	& \ge \Big\{\sum_{t=1}^T \varepsilon_t^2 - 6 (q+\hat{q}_k+1) \log T\Big\} - \Big\{\sum_{t=1}^T\varepsilon_t^2  + ( 4q \bar{C} \sqrt{6C'} + 4q C' \bar{C}^2 )\log T\Big\} \\
	& \qquad+  (\hat{q}_k-q) \log^\alpha(T) \\
	&= (\hat{q}_k - q) \{\log^\alpha(T)  - 6 \log T\} - (12q + 4q \bar{C} \sqrt{6C'} + 4q C' \bar{C}^2+6) \log T  \\
	& \ge \{\log^\alpha(T)  - 6 \log T\} - (12q + 4q \bar{C} \sqrt{6C'} + 4q C' \bar{C}^2+6) \log T  > 0
	\end{align*}
	for large enough $T$, which implies $\Pb{\hat{q}>q} \rightarrow 0$.
	
	In conclusion, we have established $\Pb{\hat{q}=q} \rightarrow 1$.
	
	\subsubsection*{Part III. Estimation of the change-point locations}
	
	In view of the conclusion of Part II, in the rest of the proof we could assume that $E_T$ holds and $\hat{q} = q$. Suppose that the model picked via NOT with the sSIC is $\hat{\tau}_{1},\ldots,\hat{\tau}_{q} \in \Tc(\zeta^{(\hat{k})})$. Furthermore, let
	\[
	j^* = \argmax_{j=1,\ldots,q} \min_{i = 1,\ldots,q} |\hat{\tau}_{i} - {\tau}_{j}| \quad\mbox{ and }\quad C:=  \frac{\min\big(\lfloor\delta_T/2\rfloor, \min_{i = 1,\ldots,q}  |\hat{\tau}_{i} - {\tau}_{j^*}|\big)}{\log T}.
	\]
	Our aim is to show that $C$ is finite (more precisely, has an upper bound independent of $T$). Now consider a 	``near-saturated oracle'' candidate model with $2q+1$ change-points at 
	\[
	\{\hat{\tau}_{1},\ldots,\hat{\tau}_{q}, {\tau}_{1},\ldots,{\tau}_{j^*-1},{\tau}_{j^*+1},\ldots,\hat{\tau}_{q},{\tau}_{j^*}-C\log T, {\tau}_{j^*}+C\log T\} 
	\]
	with the corresponding estimated variance of the errors denoted as $\dot{\sigma}^2_{\hat{k}}$. So here instead of adding all the true change-points to the set of estimated change-points as before (which generates the so-called ``saturated oracle''), we add all true change-points apart from $\tau_{j^*}$, and replace it by ${\tau}_{j^*}\pm C\log T$.
	
	Note that by construction (i.e. via $\delta_T$ in the definition of $C$), $f_t$ is constant on $\{{\tau}_{j^*}-C\log T+1,\ldots, {\tau}_{j^*}\}$ and $\{{\tau}_{j^*}+1,\ldots, {\tau}_{j^*}+C\log T\}$. In addition, $\Delta_{j^*}^\fb = |f_{\tau_{j^*}+1}-f_{\tau_{j^*}}| \ge \fl_{T}$. Write 
	\[
	\bar{\varepsilon}_* = \frac{1}{2C\log T}\sum_{t={\tau}_{j^*}-C\log T+1}^{{\tau}_{j^*}+C\log T}{\varepsilon_t}.
	\]
	
	Without loss of generality, assume that $f_{\tau_{j^*}+1} > f_{\tau_{j^*}}$. Now using the argument similar to that in Part II, we see that
	\begin{align*} 
	T \hat{\sigma}^2_{\hat{k}} \ge  T\dot{\sigma}^2_{\hat{k}} &\ge \sum_{t=1}^{{\tau}_{j^*}-C\log T} \varepsilon_t^2 +   \sum_{t={\tau}_{j^*}+C\log T+1}^{T} \varepsilon_t^2 -    (2q) 6\log T \\
	&\qquad + \sum_{t={\tau}_{j^*}-C\log T+1}^{\tau_{j^*}} (\varepsilon_t-\Delta_{j^*}^\fb/2- \bar{\varepsilon}_*)^2  +  \sum_{t={\tau}_{j^*}+1}^{{\tau}_{j^*}+C\log T} (\varepsilon_t+\Delta_{j^*}^\fb/2- \bar{\varepsilon}_*)^2  \\
	& =  \sum_{t=1}^{T} \varepsilon_t^2 -  12q \log T  + \Delta_{j^*}^\fb \Big(\sum_{t={\tau}_{j^*}+1}^{{\tau}_{j^*}+C\log T} \varepsilon_t - \sum_{t={\tau}_{j^*}-C\log T+1}^{\tau_{j^*}}\varepsilon_t\Big) \\
	&\qquad +  (\Delta_{j^*}^\fb /2)^2 (2C \log T) - (2C\log T) \bar{\varepsilon}_*^2\\
	& =  \sum_{t=1}^{T} \varepsilon_t^2 -   12q \log T  + \Delta_{j^*}^\fb  \sqrt{C \log T} \Big\{\langle \varepsilonb ,\mathbf{1}_{{\tau}_{j^*}+1,\tau_{j^*}+C\log T}\rangle - \langle \varepsilonb , \mathbf{1}_{{\tau}_{j^*}-C\log T+1,\tau_{j^*}}\rangle \Big\}\\
	& \qquad  +  (\Delta_{j^*}^\fb /2)^2 (2C \log T) -  \langle \varepsilonb , \mathbf{1}_{{\tau}_{j^*}-C\log T+1,\tau_{j^*}+C\log T}\rangle^2\\
	& \ge \sum_{t=1}^{T} \varepsilon_t^2 -   \{6 (2q+1) + 2 \sqrt{6C}\Delta_{j^*}^\fb\} \log T + (\Delta_{j^*}^\fb /2)^2 (2C \log T)
	\end{align*} 
	However,
	\[
	T \hat{\sigma}^2_{\hat{k}} \le T \hat{\sigma}^2_{k^*} \le  \sum_{t=1}^T\varepsilon_t^2 +  ( 4q \bar{C} \sqrt{6C'} + 4q C' \bar{C}^2 )\log T
	\]
	Combining the above two inequalities, and after some algebraic manipulations, we get
	\[
	2q \bar{C} \sqrt{6C'} + 2q C' \bar{C}^2  \ge   C (\Delta_{j^*}^\fb /2)^2  - 3 (2q+1) -  \sqrt{6C}\Delta_{j^*}^\fb,
	\]
	and thus
	\[
	2q \bar{C} \sqrt{6C'} + 2q C' \bar{C}^2 +  3 (2q+1) +6 \ge (\sqrt{C} \Delta_{j^*}^\fb /2 - \sqrt{6})^2,
	\]
	which entails
	\[
	C \le  4\Big[\Big\{2q \bar{C} \sqrt{6C'} + 2q C' \bar{C}^2 +  3 (2q+1) +6 \Big\}^{1/2}+ \sqrt{6}\Big]^2/{\Cl}_2^2.
	\]
	
	Finally, we remark that since $\delta_T = \min_{j=1,\ldots,q+1}(\tau_{j}-\tau_{j-1})\ge T{\Cl}_1$, for sufficiently large $T$,
	\[
	C\log T \ge \min\Big(\lfloor\delta_T/2\rfloor,\max_{j=1,\ldots,q} \min_{i = 1,\ldots,q} |\hat{\tau}_{i} - {\tau}_{j}|\Big) =  \max_{j=1,\ldots,q} |\hat{\tau}_{j} - {\tau}_{j}|.
	\]
	Therefore, $\Pb{\max_{j=1,\ldots,q} |\hat{\tau}_{j}-\tau_{j}| \leq C \log T} \rightarrow 1$, as required.
\end{proof}

\subsection{Proof of Theorem~\ref{Thm:consistency_gaussian_linear_sSIC}}
\label{Sec:consistency_gaussian_linear_sSIC}
First, we strengthen Theorem~\ref{Thm:consistency_gaussian_linear} in the scenario where the true signal has finitely many kinks (with spacings of $O(T)$).
\begin{Lemma}
	\label{lem:consistency_gaussian_linear}
	Under the assumptions of Theorem~\ref{Thm:consistency_gaussian_linear_sSIC},  there exist constants $C'$ and $\tilde{C}$ such that by setting $\zeta_{T} = \tilde{C} \sqrt{T}$ and $M \geq 36 {\Cl}_1^{-2} \log ({\Cl}_1^{-1} T)$, we have that
	\begin{align}
	\Pb{\hat{q}=q,\; \max_{j=1,\ldots,q} |\hat{\tau}_{j}-\tau_{j}| \leq C' \sqrt{T\log T}} \rightarrow 1,
	\end{align}
	as $T \rightarrow \infty$.
\end{Lemma}

\begin{proof}
	Let $\Cl, C_1, C_2, C_3 > 0$ be the constants upon applying Theorem~\ref{Thm:consistency_gaussian_linear}. For simplicity, here we shall take
	\[
	\widetilde{C} = C_2  {\Cl}_1^{3/2}  {\Cl}_2 /2 \quad \mbox{ and } \quad C'= \frac{32\sqrt{6}(\sqrt{2}+1)}{\Cl_2 \big\{\sqrt{3}\Cl_1\widetilde{C}/\bar{C}\big\}^{1/3}}
	\] 
	
	First, we verify that the conditions in Theorem~\ref{Thm:consistency_gaussian_linear} are satisfied. Specifically, we note that under the additional assumptions of Theorem~\ref{Thm:consistency_gaussian_linear_sSIC}, for sufficiently large $T$,
	\begin{enumerate}
		\item $\delta_{T}^{3/2} \fl_{T} \ge  {\Cl}_1^{3/2}  {\Cl}_2 \sqrt{T} > \Cl \sqrt{\log T}$,
		\item $\zeta_T = \widetilde{C} \sqrt{T} \in [C_1 \sqrt{\log T}, C_2 \delta_{T}^{3/2} \fl_{T})$,
		\item $M  \geq 36 {\Cl}_1^{-2} \log ({\Cl}_1^{-1} T) \ge 36 (T/\delta_T)^{2} \log \{(T/\delta_T)T\}$.
	\end{enumerate}	
	This means that
	\[
	\Pb{\hat{q}=q,\; \max_{j=1,\ldots,q} |\hat{\tau}_{j}-\tau_{j}| \leq C_3 {\Cl}_2^{-2/3} (T^2 \log T)^{1/3} } \rightarrow 1.
	\]
	
	Second, to strengthen the convergence rate of $\max_{j=1,\ldots,q} |\hat{\tau}_{j}-\tau_{j}|$, we make some minor modifications to Step~Four in the proof of Theorem~\ref{Thm:consistency_gaussian_linear}. 
	
	We still let $m^{*} = \argmin_{m\in\Oc_{s,e}} (e_m - s_m+1)$ and $b^{*} = \argmax_{s_{m^*} \leq b < e_{m^*}} \cont{s_{m^*}}{e_{m^*}}{b}{\Yb}$, where $[s_{m^*}, e_{m^*})$ must contain exactly one change-point. Again, we consider  $\tau_{j}\in [s_{m^*}, e_{m^*})$, and let $\eta_{L}=\tau_{j} -s_{m^*}$ and $\eta_{R}=e_{m^{*}}-\tau_{j}$.  
	Note that 
	\[
	\max_{j = 1,\ldots,q} \Delta_j^\fb \le \frac{4 \max_{i=1,\ldots,T} |f_i| }{\delta_T} \le \frac{4\bar{C}}{\Cl_1}\frac{1}{T} 
	\]
	By setting $\eta_T = \big\{\sqrt{3}\Cl_1\widetilde{C}/(8\bar{C})\big\}^{2/3}T-1$ (\emph{different} from the proof of Theorem~\ref{Thm:consistency_gaussian_linear}), we observe that $\min(\eta_L,\eta_R) > \eta_T$ for sufficiently large $T$ (satisfying $8 \log T < \widetilde{C}^2T/4$). It is because otherwise $\min(\eta_L, \eta_R) \leq \eta_{T}$ and Lemma~\ref{Lem:cusum_at_cp_kink} would imply that 
	\begin{align*}	
	\cont{s_{m^*}}{e_{m^*}}{b^{*}}{\Yb} &\leq \cont{s_{m^*}}{e_{m^*}}{b^{*}}{\fb} +\lambda_{T} \leq  	\cont{s_{m^*}}{e_{m^*}}{\tau_{j}}{\fb} +\lambda_{T} \leq \frac{1}{\sqrt{3}}(\eta_{T}+1)^{3/2}\frac{4\bar{C}}{\Cl_1}\frac{1}{T} + \lambda_{T} \\
	&= \frac{\widetilde{C}}{2}\sqrt{T} + \sqrt{8 \log T}  < \widetilde{C}\sqrt{T} = \zeta_T,
	\end{align*}
	which leads to a contradiction.
	
	We are now in the position to prove that $|b^{*} - \tau_{j}| \leq  C' \sqrt{T\log T} := \epsilon_{T}$. Note that in view of Theorem~\ref{Thm:consistency_gaussian_linear}, it suffices to only consider 
	\[
	b\in\Big\{s_{m^{*}}+1,\ldots, e_{m^{*}}-1\Big\} \cap \Big\{\tau_{j} - \lceil C_{3} (\Delta_{j}^{\fb})^{-2/3} (\log T)^{1/3}\rceil, \ldots, \tau_{j} + \lceil C_{3} (\Delta_{j}^{\fb})^{-2/3} (\log T)^{1/3}\rceil\Big\}
	\]
	Our aim is to show that given $|b-\tau_j|>\epsilon_{T}$ (as well as $|b-\tau_j| \le  C_{3} (\Delta_{j}^{\fb})^{-2/3} (\log T)^{1/3}$, according to Theorem~\ref{Thm:consistency_gaussian_linear}),
	\begin{align}
	\label{Eq:linear_th_proof_interval_rate1_strong}
	(\cont{s_{m^*}}{e_{m^*}}{\tau_{j}}{\Yb})^{2}-(\cont{s_{m^*}}{e_{m^*}}{b}{\Yb})^{2} > 0.
	\end{align}
	Inequality \eqref{Eq:linear_th_proof_interval_rate1_strong} does not hold for $b=b^{*}$, so proving this claim demonstrates that $|b^* - \tau_j|\leq \epsilon_{T}$.  
	
	Using arguments as those in Step~Four of the proof of Theorem~\ref{Thm:consistency_gaussian_const} (or Theorem~\ref{Thm:consistency_gaussian_linear}), we can show that \eqref{Eq:linear_th_proof_interval_rate1_strong} is implied by $\kappa > (\sqrt{2}+1)^2\lambda_T^2$, where $\kappa=(\cont{s_{m^*}}{e_{m^*}}{\tau_{j}}{\fb})^{2}-(\cont{s_{m^*}}{e_{m^*}}{b}{\fb})^{2}$. By Lemma~\ref{Lem:cusum_cp_kink_distance_2}, $\kappa > (\sqrt{2}+1)^2\lambda_T^2$ is implied by
	\begin{align}
	\label{Eq:linear_th_proof_interval_rate1_strong_2}
	\frac{(\Delta^{\fb}_j)^2}{96} \big\{ \min(\eta_L,\eta_R) -1 \big\}  |b-\tau_j|^2 > (\sqrt{2}+1)^2\lambda_T^2,
	\end{align}
	In view of the fact that 
	\begin{align*}
	\min(\eta_L,\eta_R) - 1 > \eta_T -1 = \big\{\sqrt{3}\Cl_1\widetilde{C}/(8\bar{C})\big\}^{2/3}T-2 > \big\{\sqrt{3}\Cl_1\widetilde{C}/(8\bar{C})\big\}^{2/3}T/2
	\end{align*}
	for sufficiently large $T$,  \eqref{Eq:linear_th_proof_interval_rate1_strong_2} is further implied by 
	\[
	|b-\tau_j| > \frac{16\sqrt{3}(\sqrt{2}+1)\sqrt{\log T}}{\Cl_2/T \big\{\sqrt{3}\Cl_1\widetilde{C}/(8\bar{C})\big\}^{1/3} \sqrt{T/2} } = \frac{32\sqrt{6}(\sqrt{2}+1)}{\Cl_2 \big\{\sqrt{3}\Cl_1\widetilde{C}/\bar{C}\big\}^{1/3}}\sqrt{T\log T} = C'\sqrt{T\log T}.
	\]
	In conclusion, $|b-\tau_{j}| >  \epsilon_{T}$ implies \eqref{Eq:linear_th_proof_interval_rate1_strong}, leading to a contradiction. So it must hold that $|b^{*}-\tau_{j}| \leq \epsilon_{T}$ for  large $T$. 
	
	Finally, since $\Pb{\hat{q}=q}\rightarrow 1$, we have that
	\[
	\Pb{\hat{q}=q,\; \max_{j=1,\ldots,q} |\hat{\tau}_{j}-\tau_{j}| \leq C' \sqrt{T\log T}} \rightarrow 1,
	\]
	as required. 
	
\end{proof}

Now we are in the position to prove Theorem~\ref{Thm:consistency_gaussian_linear_sSIC}.
\begin{proof}
	The proof proceeds in analogy to the proof of Theorem~\ref{Thm:consistency_gaussian_const_sSIC}. In the following, we present details of the main steps. 
	
	Again, thanks to the standard Gaussianity of the noise, for any candidate $\Tc(\zeta^{(k)})$ on the NOT solution path,  the sSIC criterion function in \ref{Scen:change_in_slope} can be written as 
	\[
	T \hat{\sigma}^2_k +  (\hat{q}_k+2) \log^\alpha(T) + \mathrm{constant}
	\]
	where $\hat{\sigma}^2_k$ is the estimated variance of the noise (i.e. the residual sum of squares divided by $T$) based on $\Tc(\zeta^{(k)})$, and $\hat{q}_k$ is the estimated number of kinks. 
	
	\subsubsection*{Part I. About a particular model candidate on the NOT solution path}
	By Lemma~\ref{lem:consistency_gaussian_linear}, we know that with arbitrarily high probability for sufficiently large $T$, there exists $k^*$ such that $\Tc(\zeta^{(k^*)})$ on the NOT solution path is a ``good'' candidate with $\hat{\tau}_{1},\ldots,\hat{\tau}_{\hat{q}_{k^*}}  \in \Tc(\zeta^{(k^*)})$ satisfying $\hat{q}_{k^*} = q$ and $\max_{i=1}^{q}|\hat{\tau}_{i}-\tau_i| \le C' \sqrt{T\log T}$ for some $C' > 0$. In the rest of the proof, for presentational convenience, we assume the existence of such $k^*$. 
	
	Define the set
	\[
	E_T = \Big \{ \max_{s,e: 1\le s \le e \le T} \max\Big(|\langle \gammab_{s,e}, {\varepsilonb}\rangle|,|\langle \mathbf{1}_{s,e}, {\varepsilonb}\rangle|\Big) \le \sqrt{6 \log T} \Big \}.
	\]
	Using the Bonferroni bound, we see that $\Pb{E_T^c} = O(T^{-1})$. Again, in the following, we could assume that $E_T$ holds.
	
	Let $\{\hat{f}_t\}_{t=1}^T$ be the fitted values using the candidate on the solution path with $\hat{\tau}_{1},\ldots,\hat{\tau}_{\hat{q}_{k^*}} \in \Tc(\zeta^{(k^*)})$ , and define $\tilde{f}_1 = \hat{f}_1$, $\tilde{f}_{t+1} = \tilde{f}_{t}+(f_{\tau_j+1}-f_{\tau_j})$  for $t = \hat{\tau}_j,\ldots,\hat{\tau}_{j+1}-1$ for every $j = 0, 1,\ldots q$. Again, here for notational convenience, we suppressed the dependence of $\{\hat{f}_t\}_{t=1}^T$ and $\{\tilde{f}_t\}_{t=1}^T$ on $k^*$. It is easy to see that $f_t- \tilde{f}_t$ is piecewise-linear and continuous, with at most $2q$ kinks and
	\[
	\max_{t=1,\ldots,T}|f_t- \tilde{f}_t|\le q \max_j(\Delta_j^\fb) C'\sqrt{T \log T} \le \frac{4\bar{C}}{\Cl_1 T} C'q\sqrt{T \log T} =  \frac{4q\bar{C}C'}{\Cl_1} \sqrt{\log T/T}.
	\]
	
	Write $\tilde{\fb}=(\tilde{f}_1,\ldots,\tilde{f}_T)'$, then $\|\fb -\tilde{\fb}\|^2 \le  (4q\bar{C}C'/\Cl_1)^2 \log T$. Furthermore, it is easy to verify (under $E_T$) that 
	\begin{align*}
	T \hat{\sigma}^2_{k^*} &= \sum_{t=1}^T(\varepsilon_t+f_t-\hat{f}_t)^2  \le \sum_{t=1}^T(\varepsilon_t+f_t-\tilde{f}_t)^2 =  \sum_{t=1}^T \varepsilon_t^2 + 2\langle \varepsilonb, \fb -\tilde{\fb} \rangle + \|\fb -\tilde{\fb}\|^2 \\
	&=  \sum_{t=1}^T \varepsilon_t^2 + M \log T 
	\end{align*}
	for some constant $M$ that does not depend on $T$. Consequently, as $T \rightarrow \infty$, it follows that $\Pb{\hat{\sigma}^2_{k^*} < 1+ \delta} = 1$ for any $\delta >0$. 
	
	\subsubsection*{Part II. Estimation of the number of change-points}
	Our aim in this part is to show that $\Pb{\hat{q}=q} \rightarrow 1$ as $T \rightarrow \infty$. We accomplish this by showing separately that (i) $\Pb{\hat{q}<q} \rightarrow 0$ and (ii) $\Pb{\hat{q}>q} \rightarrow 0$. 
	
	First, we note that it follows from Lemma~5.3~and~5.4 of \cite{LWZ1997} that there exists $\delta >0$ such that as $T \rightarrow \infty$,
	\[
	\min_{k: \hat{q}_k < q} \Pb{\hat{\sigma}^2_k > 1+ \delta} \rightarrow 1.
	\]
	This means that for all $k$ with $\hat{q}_k < q$,
	\[
	\mbox{sSIC}(k) - \mbox{sSIC}(k^*) = T(\hat{\sigma}_k^2 - \hat{\sigma}_{k^*}^2) + (\hat{q}_k - q) \log^\alpha(T) \ge \delta T -  q\log^\alpha(T) > 0
	\] 
	for large enough $T$, which implies $\Pb{\hat{q}<q} \rightarrow 0$.
	
	Second, for all $k$ with $\hat{q}_k > q$ and $\hat{\tau}_{1},\ldots,\hat{\tau}_{\hat{q}_{k}} \in \Tc(\zeta^{(k)})$, we consider a ``saturated oracle'' candidate model with $\hat{q}_k+q$ kinks at $\hat{\tau}_{1},\ldots,\hat{\tau}_{\hat{q}_k},\tau_1,\ldots,\tau_q$ respectively. We reorder these $\hat{q}_k+q$ locations as $0 = \mathring{\tau}_{0} < \mathring{\tau}_{1} \le \ldots \le  \mathring{\tau}_{\hat{q}_k+q} <  \mathring{\tau}_{\hat{q}_k+q+ 1} = T$, and denote by $\mathring{\sigma}_k^2$ the estimated variance of the errors corresponding to a piecewise-linear model with features at these locations but \textbf{without} the continuity constraint (so effectively the way of estimating this quantity under Scenario~\ref{Scen:change_in_mean_and_slope}). Let $\varepsilonb = (\varepsilon_1,\ldots,\varepsilon_T)'$, 
	\[
	\Gammab_{s,e} := [\mathbf{1}_{s,e}, \gammab_{s,e}] \quad \mbox{and} \quad \mathbf{H}(s,e) = \Gammab_{s,e}  (\Gammab_{s,e}' \Gammab_{s,e})^{-1}  \Gammab_{s,e}'
	\]
	for $1 \le s \le e \le T$, where $\Gammab_{s,e}$ is a $T \times 2$ matrix and $\mathbf{H}(s,e)$ is a $T \times T$ matrix. Furthermore, denote by $\mathbf{I}(s,e)$ a $T\times T$ matrix with 1 on the $(s,s)$-th to the $(e,e)$-th entries and zero elsewhere. Here both $ \mathbf{H}(s,e)$ and $\mathbf{H}(s,e)-\mathbf{I}(s,e)$ are idempotent matrices.
	
	Then the residual sum of squares for fitting a linear line on $\{\mathring{\tau}_{j}+1,\ldots,\mathring{\tau}_{j+1}\}$ (on which $f_t$ is linear as well) is
	\[
	(\fb+\varepsilonb)' \{\mathbf{I}(\mathring{\tau}_{j}+1,\mathring{\tau}_{j+1}) - \mathbf{H}(\mathring{\tau}_{j}+1,\mathring{\tau}_{j+1})\}(\fb+\varepsilonb) = \varepsilonb' \{\mathbf{I}(\mathring{\tau}_{j}+1,\mathring{\tau}_{j+1}) - \mathbf{H}(\mathring{\tau}_{j}+1,\mathring{\tau}_{j+1})\} \varepsilonb.
	\]
	It then follows that
	\begin{align*} 
	T \hat{\sigma}^2_{k} \ge T \mathring{\sigma}^2_{k} &= \sum_{j=0}^{\hat{q}_k+q} \varepsilonb' \{\mathbf{I}(\mathring{\tau}_{j}+1,\mathring{\tau}_{j+1}) - \mathbf{H}(\mathring{\tau}_{j}+1,\mathring{\tau}_{j+1})\} \varepsilonb. \\
	&= \sum_{t=1}^T \varepsilon_t^2 - \sum_{j=0}^{\hat{q}_k+q} \varepsilonb' \mathbf{H}(\mathring{\tau}_{j}+1,\mathring{\tau}_{j+1}) \varepsilonb.
	\end{align*}
	Note that $\varepsilonb' \mathbf{H}(s,e) \varepsilonb$ follows a $\chi_2^2$ distribution. For any $Z \sim \chi_2^2$, $\Pb{Z>z} \le e^{-z/2}$. Therefore, by defining the set
	\[
	G_T = \Big \{ \max_{s,e: 1\le s \le e \le T} \varepsilonb' \mathbf{H}(s,e) \varepsilonb \le 6 \log T \Big \},
	\]
	we have that $\Pb{G_T^c} = O(T^{-1})$  using the Bonferroni bound. Now assume that $G_T$ holds, it follows that
	\[
	T \hat{\sigma}^2_{k} \ge \sum_{t=1}^T \varepsilon_t^2  - 6(\hat{q}_k+q+1)\log T
	\]
	This means that for all $k$ with $\hat{q}_k > q$,
	\begin{align*}
	\mbox{sSIC}(k) - \mbox{sSIC}(k^*) &\ge T(\mathring{\sigma}_k^2 - \hat{\sigma}_{k^*}^2) + (\hat{q}_k - q) \log^\alpha(T) \\
	& \ge  (\hat{q}_k - q) \log^\alpha(T)  - \{6(\hat{q}_k+q+1)+M\}\log T \\
	& = (\hat{q}_k - q) \{\log^\alpha(T)  - 6 \log T\} - (12q+6+M) \log T \\
	& \ge  \log^\alpha(T) -(12q+12+M) \log T >  0
	\end{align*}
	for large enough $T$, which in turn implies $\Pb{\hat{q}>q} \rightarrow 0$.
	
	In conclusion, we have established that $\Pb{\hat{q}=q} \rightarrow 1$.
	
	\subsubsection*{Part III. Estimation of the change-point locations}
	In view of the conclusion of Part II, in the rest of the proof we could assume that $A_T \cap B_T \cap D_T \cap E_T \cap G_T$ holds and $\hat{q} = q$. 
	
	Suppose that the model picked via NOT with the sSIC is $\hat{\tau}_{1},\ldots,\hat{\tau}_{q} \in \Tc(\zeta^{(\hat{k})})$. Comparing the residual sum of squares of this candidate with $\Tc(\zeta^{(k^*)})$ yields that $\hat{\tau}_{j} \in \{\tau_j-\lfloor\delta_T/6\rfloor+1,\ldots,\tau_j+\lfloor\delta_T/6\rfloor-1\}$. It is because otherwise one could find an interval of length roughly $\delta_T/3$ (so of $O(T)$) with a true kink in the middle of but with no kinks in its estimates, leading to $\hat{\sigma}^2 \rightarrow 1+ \delta$ (see Lemma~5.3~and~5.4 of \cite{LWZ1997}), and thus a contradiction. Moreover, it is easy to see that $\hat{\tau}_{j}$ is the only estimated kink over $\{\tau_j-\lceil\delta_T/3\rceil-1,\ldots,\tau_j+\lceil\delta_T/3\rceil+1\}$ for every $j=1,\ldots,q$.
	
	Let
	\[
	j^* = \argmax_{j=1,\ldots,q} |\hat{\tau}_{j} - {\tau}_{j}| .
	\]
	Now consider a 	``near-saturated oracle'' candidate model with $2q+1$ kinks at 
	\[
	\{\hat{\tau}_{1},\ldots,\hat{\tau}_{q}, {\tau}_{1},\ldots,{\tau}_{j^*-1},{\tau}_{j^*+1},\ldots,\hat{\tau}_{q},{\tau}_{j^*} - \lceil\delta_T/3\rceil-1, {\tau}_{j^*}+\lceil\delta_T/3\rceil+1\} 
	\]
	with the corresponding estimated variance of the errors denoted as $\dot{\sigma}^2_{\hat{k}}$. So again, instead of adding all the true kinks to the set of estimated kinks as before (which generates the so-called ``saturated oracle''), we add all true kinks apart from $\tau_{j^*}$, and replace it by ${\tau}_{j^*}\pm (\lceil\delta_T/3\rceil+1)$.
	
	Note that $\dot{\sigma}^2_{\hat{k}}$ is no smaller than the estimated variance of the errors from a model with the features at the same $2q+1$ locations, but with the continuity constraint only enforced at $\hat{\tau}_{j^*}$. More precisely, in the rest of the proof we could effectively follow a model with Scenario~\ref{Scen:change_in_slope} over $\{\tau_{j^*}-\lceil\delta_T/3\rceil,\ldots,\tau_{j^*}+\lceil\delta_T/3\rceil+1\}$ and Scenario~\ref{Scen:change_in_mean_and_slope} elsewhere.
	
	In addition, for any $1 \le s\le b \le e \le T$, 
	\begin{align*} 
	&\|	\Yb|_{[s,e]} -  \langle\Yb, \phib_{s,e}^{b}\rangle \phib_{s,e}^{b} - \langle \Yb,  \gammab_{s,e}\rangle  \gammab_{s,e} - \langle \Yb,  \mathbf{1}_{s,e} \rangle  \mathbf{1}_{s,e} \|^2	 \\
	&=  \| \Yb|_{[s,e]} -  \langle \Yb,  \gammab_{s,e}\rangle  \gammab_{s,e} - \langle \Yb,  \mathbf{1}_{s,e} \rangle  \mathbf{1}_{s,e} \|^2 -  \langle\Yb, \phib_{s,e}^{b}\rangle ^2 \\
	&=  \| \Yb|_{[s,e]} -  \langle \Yb,  \gammab_{s,e}\rangle  \gammab_{s,e} - \langle \Yb,  \mathbf{1}_{s,e} \rangle  \mathbf{1}_{s,e} \|^2 -  (\cont{s}{e}{b}{\Yb})^2
	\end{align*}
	Applying this result on $s = \tau_{j^*}-\lceil\delta_T/3\rceil$, $e=\tau_{j^*}+\lceil\delta_T/3\rceil+1$ and $b = \tau_{j^*}$ or $\hat{\tau}_{j^*}$, and using the argument similar to that in Part II, we obtain that
	\begin{align*} 
	T \hat{\sigma}^2_{\hat{k}} \ge  T\dot{\sigma}^2_{\hat{k}} &\ge \sum_{t=1}^{\tau_{j^*}-\lceil\delta_T/3\rceil-1} \varepsilon_t^2 +   \sum_{t={\tau}_{j^*}+\lceil\delta_T/3\rceil+2}^{T} \varepsilon_t^2 -   (2q) 6\log T \\
	&\qquad + (\cont{s}{e}{{\tau}_{j^*}}{\Yb})^2-(\cont{s}{e}{\hat{\tau}_{j^*}}{\Yb})^2  + \Big(\sum_{\tau_{j^*}-\lceil\delta_T/3\rceil}^{{\tau}_{j^*}+\lceil\delta_T/3\rceil+1} \varepsilon_t^2 - 12 \log T\Big),  
	\end{align*}
	where $\sum_{\tau_{j^*}-\lceil\delta_T/3\rceil}^{{\tau}_{j^*}+\lceil\delta_T/3\rceil+1} \varepsilon_t^2 - 12 \log T$ is the lower-bound of the residual sum of squares for fitting a piecewise-linear function over $\{\tau_{j^*}-\lceil\delta_T/3\rceil,\ldots,{\tau}_{j^*}+\lceil\delta_T/3\rceil+1\}$ with only one feature at $\tau_{j^*}$.
	Consequently, it follows from the argument in Step Four of the proof of Theorem~\ref{Thm:consistency_gaussian_const} that
	\begin{align*} 
	T \hat{\sigma}^2_{\hat{k}} & \ge  \sum_{t=1}^{T} \varepsilon_t^2 -   6 (2q+2) \log T   + (\cont{s}{e}{{\tau}_{j^*}}{\fb})^2-(\cont{s}{e}{\hat{\tau}_{j^*}}{\fb})^2 - 2 \sqrt{8 \log T}\sqrt{(\cont{s}{e}{{\tau}_{j^*}}{\fb})^2-(\cont{s}{e}{\hat{\tau}_{j^*}}{\fb})^2} - 8 \log T \\
	& =  \sum_{t=1}^{T} \varepsilon_t^2 -   6 (2q+2) \log T   + \Big(\sqrt{(\cont{s}{e}{{\tau}_{j^*}}{\fb})^2-(\cont{s}{e}{\hat{\tau}_{j^*}}{\fb})^2} - \sqrt{8 \log T}\Big)^2  - 16 \log T \\
	& \ge  \sum_{t=1}^{T} \varepsilon_t^2 -   (12q+28)\log T   + \Big(\frac{\Cl_2}{\sqrt{96} T}(\Cl_1T/3+1-1)^{1/2} |\hat{\tau}_{j^*} - \tau_{j^*}| - \sqrt{8 \log T}\Big)^2\\
	& =  \sum_{t=1}^{T} \varepsilon_t^2 -   (12q+28)\log T   +  \Big(\sqrt{\frac{\Cl_1\Cl_2^2}{288T}} |\hat{\tau}_{j^*} - \tau_{j^*}| -  \sqrt{8 \log T}\Big)^2,
	\end{align*} 
	where we used the fact that $|\hat{\tau}_{j^*} - \tau_{j^*}| < \delta_T/6 = \frac{1}{2}\frac{\delta_T}{3}$ and Lemma~\ref{Lem:cusum_cp_kink_distance_2} in the second last line above.
	
	However,
	\[
	T \hat{\sigma}^2_{\hat{k}} \le T \hat{\sigma}^2_{k^*} \le  \sum_{t=1}^T\varepsilon_t^2 +  M \log T
	\]
	Combining the above two inequalities, and after some algebraic manipulations, we get
	\[
	|\hat{\tau}_{j^*} - \tau_{j^*}| \le \sqrt{\frac{288}{\Cl_1\Cl_2^2}}(\sqrt{M+12q+28}+\sqrt{8}) \sqrt{T \log T} =: C \sqrt{T \log T}
	\]
	
	Therefore, $\Pb{\max_{j=1,\ldots,q} |\hat{\tau}_{j}-\tau_{j}| \leq C \sqrt{T \log T}} \rightarrow 1$, as required.
\end{proof}

\subsection{Proof of Corollary~\ref{Cor:consistency_gaussian_const}}
\label{Sec:proof_of_consistency_cor_const}
\begin{proof}
	Without loss of generality, we assume that $\sigma_0=1$. In addition, we set $P := \sum_{k=-\infty}^\infty |\rho_k|$, where $\rho_k$ is the auto-correlation function of $\{\varepsilon_t\}$.
	
	We modify our proof of Theorem~\ref{Thm:consistency_gaussian_const} in the following way:
	
	\subsubsection*{Step One and Two}
	Let $\lambda_T = \sqrt{8 P \log T}$ and define the set $A_T$ as before. Denote the autocorrelation matrix of $\{\varepsilon_t\}$ by $\mathbf{P}_T = [\rho_{i-j}]_{i,j=1,\ldots,T}$ (which is also the autocovariance matrix, since $\varepsilon_t$ has unit-variance). Then since $\mathbf{P}_T$ is symmetric, we have that
	\[
	\|\mathbf{P}_T\|_\infty = \|\mathbf{P}_T\|_1 = \max_{j} \sum_i |P_{ij}| \le P,
	\]
	where $\|\cdot\|_\infty$ and $\|\cdot\|_1$ are the operator norms of a matrix. Consequently, by H\"{o}lder's inequality, $\|\mathbf{P}_T\|_2 \le \sqrt{\|\mathbf{P}_T\|_1 \; \|\mathbf{P}_T\|_\infty} \le P$, i.e., the largest eigenvalue of $\mathbf{P}_T$ is bounded above by $P$, which is irrelevant of $T$.
	
	For any $s,b,e$ such that $1 \le s < b < e \le T$, since $\big\langle \psib_{s,e}^b  ,\varepsilonb\big\rangle$ has a normal distribution, with zero-mean and 
	\[
	\Var(\big\langle \psib_{s,e}^b  ,\varepsilonb\big\rangle) =  (\psib_{s,e}^b)^T \mathbf{P}_T \psib_{s,e}^b \le P \|\psib_{s,e}^b\|_2^2 \le P,
	\]
	we have that
	\[
	\Pb{|\cont{s}{e}{b}{\varepsilonb}| \ge \lambda_T}  = \Pb{|\cont{s}{e}{b}{\varepsilonb}|/\sqrt{P} \ge \sqrt{8\log T}} \le \frac{2 e^{-8\log T/2}}{\sqrt{8\log T} \sqrt{2\pi}}.
	\]
	It follows from the Bonferroni bound that $\Pb{A_T^c} \le {12 \sqrt{\pi} T}^{-1}$.
	
	Using the same argument as above, we can show that $\frac{|\langle \psib_{s,e}^b \langle\fb,\psib_{s,e}^b \rangle - \psib_{s,e}^{\tau_j} \langle\fb,\psib_{s,e}^{\tau_j} \rangle  ,\varepsilonb\rangle|}{\|\psib_{s,e}^b \langle\fb,\psib_{s,e}^b \rangle - \psib_{s,e}^{\tau_j} \langle\fb,\psib_{s,e}^{\tau_j} \rangle\|_2}$ is normal distributed, with zero-mean and variance bounded above by $P$ for any $ 1\le s \le b < e \le T$. Thus, $\Pb{B_T^c} \le {12 \sqrt{\pi} T}^{-1}$.
	
	\subsubsection*{Step Three, Four and Five}
	The rest of the proof goes through by changing the constants as 
	\[\Cl=\sqrt{6}\big(2\sqrt{C_3} + \sqrt{32P}\big)+1 , \quad C_1 = 2\sqrt{C_3} + \sqrt{8P}, \quad C_2 = \frac{1}{\sqrt{6}} - \frac{\sqrt{8P}}{\Cl}, \quad C_3 = (32\sqrt{2}+48)P\]
	and setting
	\[
	\eta_{T}= (C_1 - \sqrt{8P})^2.
	\]
\end{proof} 

Finally, we remark that the proof of Corollary~\ref{Cor:consistency_gaussian_linear} is similar to that of Corollary~\ref{Cor:consistency_gaussian_const}, so is omitted for brevity.

\end{document}